\documentclass[preprint,12pt]{article}
 \UseRawInputEncoding 
\usepackage{layout}
\usepackage{url}
\usepackage{proof}
\usepackage{mathrsfs}
\usepackage{amsmath,amsthm,amssymb}
\usepackage{braket}
\theoremstyle{definition}
\newtheorem{thm}{Theorem}
\newtheorem{prop}[thm]{Proposition}
\newtheorem{lemma}[thm]{Lemma}
\newtheorem{claim}[thm]{Claim}
\newtheorem{defn}[thm]{Definition}
\newtheorem{cor}[thm]{Corollary}
\newtheorem{eg}[thm]{Example}
\newtheorem{rmk}[thm]{Remark}

\newtheorem{question}{Question}
\newtheorem{convention}[thm]{Convention}

\newcommand{\QQ}{\mathbb{Q}}
\newcommand{\NN}{\mathbb{N}}
\newcommand{\ZZ}{\mathbb{Z}}

\newcommand{\FF}{\mathbb{F}}
\newcommand{\MM}{\mathbb{M}}

\newcommand{\PP}{\mathbb{P}}

\DeclareMathOperator{\mmax}{\bold{max}}
\DeclareMathOperator{\zeromat}{\bold{O}}
\DeclareMathOperator{\idmat}{\bold{I}}
\DeclareMathOperator{\row}{\bold{r}}
\DeclareMathOperator{\column}{\bold{c}}
\DeclareMathOperator{\extract}{\bold{e}}
\DeclareMathOperator{\summation}{\bold{\sum}}
\DeclareMathOperator{\lamt}{\bold{\lambda}}
\DeclareMathOperator{\chmat}{\bold{chmat}}
\DeclareMathOperator{\ch}{\bold{ch}}
\DeclareMathOperator{\detrat}{\bold{det}_{rat}}
\DeclareMathOperator{\num}{\bold{num}}
\DeclareMathOperator{\den}{\bold{den}}
\DeclareMathOperator{\adj}{\bold{adj}}
\DeclareMathOperator{\mul}{\bold{mul}}

\DeclareMathOperator{\symm}{\bold{symm}}
\DeclareMathOperator{\detpol}{\bold{det}_{pol}}
\DeclareMathOperator{\cond}{\bold{cond}}
\DeclareMathOperator{\conv}{\bold{conv}}
\DeclareMathOperator{\pconv}{\bold{pconv}}
\DeclareMathOperator{\rem}{\bold{rem}}
\DeclareMathOperator{\divi}{\bold{div}}

\DeclareMathOperator{\coeff}{\bold{coeff}}
\DeclareMathOperator{\mcoeff}{\bold{mcoeff}}

\DeclareMathOperator{\basis}{\bold{basis}}
\DeclareMathOperator{\sol}{\bold{sol}}
\DeclareMathOperator{\solve}{\bold{solve}}
\DeclareMathOperator{\im}{im}
\DeclareMathOperator{\polize}{\bold{polize}}
\DeclareMathOperator{\Mat}{\bold{Mat}}
\DeclareMathOperator{\LAPPD}{\bold{LA}\mathtt{P}\bold{PD}}

\DeclareMathOperator{\ct}{\bold{ct}}
\DeclareMathOperator{\rk}{\bold{rk}}
\DeclareMathOperator{\LA}{\bold{LA}}
\DeclareMathOperator{\LAP}{\bold{LA}\mathtt{P}}
\DeclareMathOperator{\rank}{\bold{rank}}
\DeclareMathOperator{\VNC}{\bold{VNC}}

\DeclareMathOperator{\lincoeff}{lincoeff}

\DeclareMathOperator{\power}{\mathtt{P}}
\DeclareMathOperator{\indexsort}{\mathtt{index}}
\DeclareMathOperator{\fieldsort}{\mathtt{field}}
\DeclareMathOperator{\matrixsort}{\mathtt{matrix}}
\DeclareMathOperator{\DET}{DET}
\DeclareMathOperator{\Det}{Det}
\DeclareMathOperator{\Den}{Den}
\DeclareMathOperator{\Num}{Num}
\DeclareMathOperator{\Eval}{Eval}
\usepackage{hyperref}

\usepackage{amssymb,stmaryrd}

\title{On matrix rank function over bounded arithmetics}

\author{Eitetsu KEN
\footnote{Institute: Graduate School of Mathematical Sciences, the University of Tokyo,
3-8-1 Komaba,
Meguro-ku, Tokyo-to 153-0041, Japan.
email: \url{yeongcheol-kwon@g.ecc.u-tokyo.ac.jp}}
 \& Satoru KURODA
 \footnote{Institute: Department of Culture and Informatics, Gunma Prefectural Women's University,
1395-1 Kaminote Tamamura-machi, Sawa-gun, Gumma-ken 370-1127, Japan.
email: \url{satoru@mail.gpwu.ac.jp}}}

\begin{document}
\maketitle




\begin{abstract}

In \cite{Mulmuley's algorithm}, Mulmuley gave an algorithm reducing the computation of the matrix rank function to that of determinants, of which the proof for the verification is elementary. 
In this article, we formalize this argument in the bounded arithmetic $\LAP$; that is, we show that 
\[\det(AB)=\det(A)\det(B)\]
 for matrices $A,B$ with $\FF(X)$-coefficients implies 
 \[\rank(M)=\dim(\im M),\]
  where $\FF$ is the universe of the $\fieldsort$ of the theory, $M$ is a matrix with $\FF$-coefficients, and $\rank(M)$ is the rank function computed by Mulmuley's algorithm.
Furthermore, interpreting $\LAP$ by $\VNC^{2}$ with $\FF=\QQ$ and using the result of \cite{Uniform}, we see that $\VNC^{2}$ can formalize $\rank(M)$ and prove $\rank(M)=\dim(\im M)$.

Lastly, we give several examples of combinatorial statements provable in $\VNC^{2}$, using the formalized linear algebra.
\end{abstract}









\section{Introduction}\label{Introduction}
Based on the characterization of the $NP$ v.s. $coNP$ problem \cite{CookReckhow}, proof complexity of propositional logic has been vigorously investigated. 
Among various concrete proof systems, the Frege system has been playing an important role in the field. 
Despite decades of intensive efforts, it is still open whether the Frege system is polynomially bounded or not \cite{proofcomplexity}.

If $NP \neq coNP$, then the Frege system cannot be polynomially bounded. 
Upon this observation, \cite{Hard} gave numerous examples of combinatorial statements which can be formalized as a family of propositional tautologies but yet seemed to be difficult to prove efficiently in the Frege system.
 
However, studies have been showing that the Frege system is quite strong.
It is known that there are:
\begin{itemize}
\item Polynomial-sized Frege proofs for Bondy's theorem (\cite{Hard})
\item Polynomial-sized Frege proofs for Bollob\'{a}s' theorem (\cite{Bollobas}) 
\item Quasipolynomial-sized Frege proofs for Frankl's theorem (\cite{Frankl})
\item Quasipolynomial-sized Frege proofs for multiplicativity of matrix determinant, Cayley-Hamilton theorem, and several other hard matrix identities (\cite{Uniform}). 
 \end{itemize}

This paper can be located in this line of research. 
Roughly speaking, we show that the quasipolynomial Frege system can formalize the matrix rank function $\rank(M)$ for square matrices $M$ with $\QQ$-coefficients and prove that $\rank(M)=\dim(\im M)$.

To explain it more precisely, we need the notion of bounded arithmetics.
It is known that many systems $T$ of bounded arithmetic have a counterpart $P_{T}$ in propositional proof systems so that if a $\forall \Sigma^{B}_{0}$-theorem is provable in a theory of bounded arithmetic $T$, then there are (quasi)-polynomial-sized propositional proofs is $P_{T}$ (\cite{Cook}).
In terms of bounded arithmetics, our work and the above results on feasible provability of combinatorial statements are located in the research program ``Bounded Reverse Mathematics,'' proposed by Cook.

As for linear algebra, the first and perhaps the most natural bounded arithmetics related to it are $\LA$ and $\LAP$ introduced in \cite{Soltys} and \cite{The proof complexity of linear algebra}, consisting of three-sorts: $\indexsort$-sort, $\fieldsort$-sort, and $\matrixsort$-sort.
It is model theoretically interesting that they treat general coefficient fields including those which are not even definable in the standard model of arithmetics. 
\cite{Soltys} formalized matrix determinant in $\LAP$, following \cite{Berkowitz}, and proved that
\begin{enumerate}
\item The following are equivalent over $\LAP$:
\begin{enumerate}
 \item\label{hardstart} The existence of cofactor expansion,
 \item An axiomatic definition of the determinant,
 \item\label{hardend} The Cayley-Hamilton Theorem
\end{enumerate}
\item\label{detmult} The multiplicativity of matrix determinants implies all of them above.
\end{enumerate}

Moreover, they gave a stronger theory $\forall \LAP$, which is capable to prove the above matrix identities, and a translation theorem which converts proofs of universal sentences of $\forall \LAP$ to the extended Frege system.
This corresponds to the fact that the above hard matrix identities are all provable in the bounded arithmetic $V^{1}_{1}$.

However, whether the hard matrix identities (\ref{hardstart})-(\ref{hardend}) and (\ref{detmult}) above are provable in bounded arithmetics (seeming to be) weaker than $V^{1}_{1}$ had been a long-standing open problem.
The next breakthrough was given by \cite{Short proofs for the determinant identities} \cite{Uniform}.
 Formalizing the matrix determinant via Schur complement by arithmetical circuits with division, they showed that, for integer coefficients, multiplicativity of determinants and therefore the hard matrix equalities mentioned above can be proven in the bounded arithmetic $\VNC^{2}$, which corresponds to Frege proofs of $2^{O(|n|^{2})}$-size, and therefore the corresponding tautologies have quasipolynomial-sized Frege proofs.

The main contribution of this article is to give upper bounds on the provability of rank properties.
There have been a series of research on comparison of computational complexity of the matrix rank function in the context of parallel algorithms (\cite{ParallelLinearAlgorithms}), and it is known that the rank function is reduced to the matrix determinant by Mulmuley's elegant parallel algorithm \cite{Mulmuley's algorithm}.
Although it is straightforward to formalize the matrix rank function in $\LAP$ since Mulmuley's algorithm does not rely on the coefficient field, it does not directly imply that basic rank properties are provable in the theory.

Indeed, we consider an extension $\LAPPD$ of $\LAP$, which has an additional axiom $\det(AB)=\det(A)\det(B)$ for square matrices $A,B$ with $\FF(X)$-coefficients.
Here, $\FF$ is the universe of the $\fieldsort$ of $\LAP$.
Then we will show that the rank of a matrix with $\FF$-coefficients computed by Mulmuley's algorithm coincides with:
\begin{itemize}
 \item the number of linearly independent row vectors
 \item the number of linearly independent column vectors
 \item the degree of maximal nonzero minors
\end{itemize}

Furthermore, we prove that $\LAPPD$ can be interpreted by $\VNC^{2}$ with $\FF=\QQ$, and therefore quasipolynomial Frege proofs can refer and utilize the matrix rank function in their reasoning.

The organization of the paper is as follows:

In \S \ref{Preliminaries}, we set up our notations and conventions on the bounded arithmetics we use: namely, $\LA_{-},\LA$, $\LAP$, and $\LAP_{-}$. 

In \S \ref{SigmaB0-formula and -definability}, we introduce $\Sigma^{B}_{0}$-definability in the language of $\LAP$ and observe that $\LAP$ has quantifier elimination for $\Sigma^{B}_{0}$-formulae.

In \S \ref{A formalization of polynomials in LAP}, we formalize some basic treatment of polynomials in $\LAP$ based on the results in section \ref{SigmaB0-formula and -definability}.
To be precise, we observe that
 \[\langle \MM, \FF[X], \Mat_{\FF[X]} \rangle \models \LAP_{-}\]
 holds given 
\[\langle \MM, \FF, \Mat_{\FF} \rangle \models \LAP.\]

Formally speaking, we define an interpretation $\llbracket \cdot \rrbracket_{pol}$ of $\mathcal{L}_{\LAP_{-}}$ by $\mathcal{L}_{\LAP}$, which interprets $\fieldsort$ as the universe of polynomials (denoted by $\FF[X]$ as usual) and $\matrixsort$ as the universe of matrices with polynomial coefficients (denoted by $\Mat_{\FF[X]}$).

In \S \ref{Rational functions}, analogously with the previous \S \ref{A formalization of polynomials in LAP}, we formalize some basic treatment of rational functions in $\LAP$.
We introduce $\FF(X)$ and $\Mat_{\FF(X)}$, and
we observe that 
\[\langle \MM, \FF(X), \Mat_{\FF(X)} \rangle \models \LAP.\]

In \S \ref{The theory LAPPD}, we introduce a bounded arithmetic $\LAPPD$, which is $\LAP$ plus multiplicativity of determinants in $\Mat_{\FF(X)}$.
An alternative simpler definition of $\LAPPD$ which just mentions polynomials is also presented.

In \S \ref{Matrix substitution to polynomials}, we formalize matrix substitution to polynomials in our notations.

In \S \ref{A definition of rank and its basic properties}, using the basic properties of rational functions in the previous sections, we formalize the matrix rank function following Mulmuley's algorithm and prove that it coincides with several other notions of linear algebra in $\LAPPD$.

In \S \ref{Interpretation}, we observe that $\VNC^{2}$ can interpret $\LAPPD$ for $\QQ$-coefficients.
The work includes that the definition of the determinant in \cite{Uniform} is provably compatible with Berkowitz' algorithm (\cite{Berkowitz}) in $\VNC^{2}$.
(The latter was the approach of \cite{Soltys}'s formalization.)  

In \S \ref{Some combinatorial results}, we mention several combinatorial principles whose known proofs utilize linear algebra methods and which are provable in $\VNC^{2}$.
We start with immediate corollaries of the results of \cite{Uniform}: namely, the Oddtown theorem, the Graham-Pollak theorem, the Fisher's inequality, and a version of Ray-Chaudhuri-Wilson theorem.
Furthermore, as an application of our formalization of the matrix rank function, we formalize Grolmusz's explicit construction of Ramsey graphs. 

Lastly, \S \ref{Open questions} discusses open questions.

The Appendix includes technical definitions and the details of proofs in sections \ref{Preliminaries}, \ref{A formalization of polynomials in LAP}, and \ref{Rational functions}, and \ref{Matrix substitution to polynomials}.
 
\section{Acknowledgement}
We are grateful for Toshiyasu Arai and Iddo Tzameret for their comments.
The authors deeply thank anonymous referees for their sincere suggestions, which greatly improved the presentation of this paper.
This work was supported by JSPS KAKENHI Grant Numbers 18K03400 and 22J22505.
The first author was also supported by Grant-in-Aid for JSPS Fellows and FoPM program at the University of Tokyo.

\section{Preliminaries}\label{Preliminaries}
 Throughout this article, when we write first-order formulae, we prioritize the readability and often use abbreviation with clear meaning. 
 For a natural number $k$, $[k]$ denotes the $k$-set $\{1,\ldots,k\}$.

We denote the language of a given first-order theory $T$ by $\mathcal{L}_{T}$.
Also, a tuple $(x_{1}, \ldots, x_{k})$ of variables is often denoted by $\bar{x}$ if the length is clear.
For a formula $\varphi(x_{1}, \ldots, x_{k})$ and terms $u_{1}, \ldots, u_{k}$, 
\[\varphi [x_{1} \mapsto u_{1},\ldots, x_{k} \mapsto u_{k}] \ \mbox{or} \  \varphi[\bar{x} \mapsto \bar{u}]\]
 denotes the formula obtained by simultaneously substituting each $u_{i}$ for $x_{i}$.
 If $\bar{x}$ is clear from the context, we just write $\varphi(\bar{u})$.

Similarly, for a term $t(x_{1}, \ldots, x_{k})$ and terms $u_{1}, \ldots, u_{k}$, 
\[t [x_{1} \mapsto u_{1},\ldots, x_{k} \mapsto u_{k}]\ \mbox{or} \  t[\bar{x} \mapsto \bar{u}]\]
 denotes the term obtained by simultaneously substituting each $u_{i}$ for $x_{i}$.
 If $\bar{x}$ is clear from the context, we just write $t(\bar{u})$.

The theory $\LA$ and its extension $\LAP$ are defined as quantifier-free theories in \cite{The proof complexity of linear algebra}, and they are later treated as first-order theories in \cite{LAPasFOtheory} as well.
Unless stated otherwise, we follow the notations of  \cite{The proof complexity of linear algebra}, but we treat $\LA$ and its variations as many-sorted first order theories.

Concretely;

\begin{defn}\label{The language of LA}
The language $\mathcal{L}_{\LA}$ is three-sorted: $\indexsort$, $\fieldsort$, and $\matrixsort$, and it is defined as the limit language $\mathcal{L}_{\LA}:= \bigcup_{k=0}^{\infty} \mathcal{L}_{k}$, where $\mathcal{L}_{k}$ ($k \geq 0$) is defined as follows: $\mathcal{L}_{0}$ consists of the following relation and function symbols:
\begin{enumerate}
 \item $0_{\indexsort}$, $1_{\indexsort}$ (constants of $\indexsort$)
 \item $+_{\indexsort}$, $*_{\indexsort}$, $-_{\indexsort}$, $\divi$, $\rem$ (functions of type $\indexsort \times \indexsort \rightarrow \indexsort$)
 \item $0_{\fieldsort}$, $1_{\fieldsort}$. (constants of $\fieldsort$)
 \item $+_{\fieldsort}$, $*_{\fieldsort}$ (functions of type $\fieldsort \times \fieldsort \rightarrow \fieldsort$)
 \item $-_{\fieldsort}$, $(\cdot)^{-1}$ (a function of type $\fieldsort \rightarrow \fieldsort$)
 \item $\row$, $\column$ (functions of type $\matrixsort \rightarrow \indexsort$)
 \item $\extract$ (a function of type $\matrixsort \times \indexsort \times \indexsort \rightarrow \fieldsort$)
 \item $\summation$ (a function of type $\matrixsort \rightarrow \fieldsort$)
 \item $\leq_{\indexsort}$, $=_{\indexsort}$ (binary relations on $\indexsort$)
 \item $=_{\fieldsort}$ (a binary relation on $\fieldsort$)
 \item $=_{\matrixsort}$ (a binary relation on $\matrixsort$)
\end{enumerate}

 Now, assuming $\mathcal{L}_{k}$ is already defined, $\mathcal{L}_{k+1}$ is obtained by adding the following new function symbols to $\mathcal{L}_{k}$: 

\begin{enumerate}
 \item A function symbol $\cond_{\indexsort}^{\alpha,\vec{x}}$ of type 
 \[\indexsort^{l+2} \rightarrow \indexsort; (\vec{x},y,z) \mapsto \cond_{\indexsort}^{\alpha,\vec{x}}(\vec{x},y,z)\]
  for each open $\mathcal{L}_{k}$-formula $\alpha$ and a vector $\vec{x}=(x_{1},\ldots,x_{l})$ ($l \geq 0$) of $\indexsort$-variables, where the variables occurring in $\alpha$ are among the elements of $\vec{x}$. (Some $x_{i}$ might not occur in $\alpha$.)
 \item Similarly as above, a function symbol $\cond_{\fieldsort}^{\alpha,\vec{x}}$ of type 
 \[\indexsort^{l} \times \fieldsort^{2} \rightarrow \fieldsort; (\vec{x},y,z) \mapsto \cond_{\fieldsort}^{\alpha,\vec{x}}(\vec{x},y,z)\]
  for each open $\mathcal{L}_{k}$-formula $\alpha$ and a vector $\vec{x}=(x_{1},\ldots,x_{l})$ of $\indexsort$-variables, where the variables occurring in $\alpha$ are among the elements of $\vec{x}$.
  \item A function symbol $\lamt_{ij} \langle m,n, t \rangle$ for each quintuple $(i,j,m,n,t)$, where $i$ and $j$ are $\indexsort$-variables, $m$ and $n$ are $\mathcal{L}_{k}$-terms whose outputs are of $\indexsort$ in which $i$ and $j$ do not occur, and $t$ is an $\mathcal{L}_{k}$-term whose outputs are of $\fieldsort$.
\end{enumerate}
\end{defn}


\begin{convention}
We often omit subscripts of the symbols (e.g. $=_{\indexsort}$ is denoted by $=$).
Formally speaking, this causes clashes of several symbols, for example, $+_{\indexsort}$ and $+_{\fieldsort}$ are both represented by the same symbol $+$.
Therefore, we follow this convention when there is no danger of confusion.
\end{convention}

Now, $\LA$ is defined as an $\mathcal{L}_{\LA}$-theory as follows (originally, all of the following were presented as deduction rules of sequent-calculus formalization, but we treat them as first-order formulae here):
\begin{defn}\label{DefLA}
\begin{enumerate}
 \item  (cf. $(A1)$-$(A5)$ in \cite{The proof complexity of linear algebra}.) 
 Equality axioms for $=_{\indexsort}$, $=_{\fieldsort}$, $=_{\matrixsort}$.
 \item (cf. Induction rule in \cite{The proof complexity of linear algebra}.)
 Open induction for $\indexsort$. 
  \item\label{axiomsforindex}  (cf. $(A6)$-$(A16)$ in \cite{The proof complexity of linear algebra}. Note that the forms of the conditions here are different than those in \cite{The proof complexity of linear algebra}.) 
   Axioms for $0_{\indexsort}$, $1_{\indexsort}$, $+_{\indexsort}$, $*_{\indexsort}$, $\leq_{\indexsort}$, $-_{\indexsort}$, $\divi$, $\rem$ which, over open induction for $\mathcal{L}_{\indexsort}$-formulae, are equivalent to the conjunction of the following
  (here, we set
 \[\mathcal{L}_{\indexsort}:=\{0_{\indexsort}, 1_{\indexsort},+_{\indexsort},*_{\indexsort},\leq_{\indexsort}, -_{\indexsort}, \divi, \rem\}.\] 
 See \ref{ReformulationofLA} for a proof sketch of the equivalence. 
 There, we refer this conjunction as (\dag)):
 \begin{enumerate}
  \item With $(0_{\indexsort}, 1_{\indexsort}, +_{\indexsort},*_{\indexsort},\leq_{\indexsort})$, $\indexsort$ forms the nonnegative part of a discretely ordered ring, that is, an ordered ring in which $1$ is the smallest positive element.
  \item $-_{\indexsort}$ is the modified-minus.
  \item $\indexsort$ admits a division with respect to $\leq_{\indexsort}$, and $\divi(i,j)$ and $\rem(i,j)$ are the quotient and the remainder respectively of division of $i$ by $j$.
   \end{enumerate}
    \item (cf. $(A17)$ in \cite{The proof complexity of linear algebra}.)
  \[\cond_{\indexsort}^{\alpha,\vec{x}}(\vec{x},y,z)=
  \begin{cases}
  y &\quad \mbox{(if $\alpha(\vec{x})$ holds)}\\
  z &\quad \mbox{(otherwise)}
  \end{cases}\]
  for all $\vec{x},y,z$ of $\indexsort$.
   \item (cf. $(A18)$-$(A26)$ except $(A21)$ in \cite{The proof complexity of linear algebra}.) 
With $(0_{\fieldsort},1_{\fieldsort}, +_{\fieldsort},*_{\fieldsort})$, $\fieldsort$ forms a commutative ring. $-_{\fieldsort}$ returns the opposite of the input.   
 \item\label{inverseaxiom} (cf. $(A21)$ in \cite{The proof complexity of linear algebra}.) 
 $(\cdot)^{-1}$ returns the inverse element of the input. 
    With this, $\fieldsort$ actually forms a field.
   \item (cf. $(A27)$ in \cite{The proof complexity of linear algebra}.) 
   \[\cond_{\fieldsort}^{\alpha,\vec{x}}(\vec{x},y,z)=
  \begin{cases}
  y &\quad \mbox{(if $\alpha(\vec{x})$ holds)}\\
  z &\quad \mbox{(otherwise)}
  \end{cases}
  \]
  for all $\vec{x}$ of $\indexsort$ and $y,z$ of $\fieldsort$.

\item\label{componentconvention} (cf. $(A28)$ in \cite{The proof complexity of linear algebra}.)
For $A$ of $\matrixsort$ and $i,j$ of $\indexsort$, $\extract (A,i,j)=0_{\fieldsort}$ if $i \not \leq_{\indexsort} \row (A)$ or $j \not\leq_{\indexsort} \column (A)$ or $i=_{\indexsort}0_{\indexsort}$ or $j=_{\indexsort}0_{\indexsort}$. 
\item\label{matrixextensionality} (cf. Matrix equality rule in \cite{The proof complexity of linear algebra}.) $A=_{\matrixsort}B$ if and only if $\row (A)=\row (B)$ and $\column (A)=\column (B)$ and 
\[\extract (A,i,j) = \extract (B,i,j)\]
 for each $i \leq_{\indexsort} \row (A)$ and $j \leq_{\indexsort} \column (A)$.
\item (cf. $(A29)$ in \cite{The proof complexity of linear algebra}. Note that $m,n,t$ may include variables other than $i,j$.)
$\lamt_{ij}\langle m,n,t \rangle$ returns a matrix $A$ with $\row (A)=m$, $\column (A)=n$, and 
\[\extract (A,i,j) = t(i,j) \quad (i \leq_{\indexsort} \row (A),\ j\leq_{\indexsort} \column (A)).\]

\item (cf. $(A30)$-$(A34)$ in \cite{The proof complexity of linear algebra}.)
A recursive characterization of $\summation (A)$, whose intended meaning is: 
\[\summation (A) = \sum_{i \leq_{\indexsort} \row (A)}\sum_{j \leq_{\indexsort} \column (A)} \extract (A,i,j).\]
 \end{enumerate}
 \end{defn}

Intuitively, by item \ref{matrixextensionality}, we may regard each $A$ of $\matrixsort$ as an \\
$(\row (A)\times \column (A))$-matrix of components of $\fieldsort$. 
Then $\extract (A,i,j)$ amounts to the $(i,j)$-component of $A$.
The item \ref{componentconvention} says we adopt the convention under which $\extract (A,i,j)=0$ if $(i,j)$ is outside of $[\row (A)] \times [\column (A)]$, the intended range of the indices for $A$.


In \cite{The proof complexity of linear algebra}, $\LA$ is shown to be powerful enough to formalize the $(k \times l)$ zero-matrix $\zeromat_{k,l}$, the identity matrix $\idmat_{k}$ of $(k\times k)$, and the matrix sum $A+B$ (the matrices $A,B$ may have different sizes; in that case, $A+B$ is a $(\mmax\{\row (A),\row (B)\} \times \mmax\{\column (A),\column (B)\})$-matrix, calculated by regarding $A$ and $B$ as the matrices of this size, padding zero to the additional lower right components) the matrix product $A*B$ (here, $A$ and $B$ are allowed to have any size, and $A*B$ is of size $\row (A) \times \column (B)$, calculated by regarding $A$ as $(\row (A) \times \mmax\{\column (A),\row (B)\})$-matrix and $B$ as $(\mmax\{\column (A),\row (B)\} \times \column (B))$-matrix, padded by $0$ for the additional right columns and lower rows respectively) and show some of their basic arithmetical properties, such as that the $(k \times k)$-matrices form a ring with $(\zeromat_{kk}, \idmat_{k}, +, *)$ (not commutative in general) for each $k$ of $\indexsort$.

In order to formalize Berkowitz algorithm (\cite{Berkowitz}), \cite{The proof complexity of linear algebra} introduced an extension $\LAP$ of $\LA$, which is helpful for our work, too.

\begin{defn}\label{The language of LAP}
The language $\mathcal{L}_{\LAP}$ is defined analogously as $\mathcal{L}_{\LA}$, except for one thing; the initial language $\mathcal{L}_{0}$ contains the following additional function symbol:
\begin{itemize}
 \item $\power$ (a function symbol of type $\indexsort \times \matrixsort \rightarrow \matrixsort$)
\end{itemize}
\end{defn}
\begin{rmk}
The intended meaning of $\power$ is: $\power (k,A)=A^{k}$. 
Here, the power is with respect to the matrix product $*$ above.
Note that, formally speaking, $A$ is allowed to be non-square, although we are interested only in when $A$ is square.
Also, note that $\mathcal{L}_{\LAP}$ includes $\lambda$-terms in which $\mathtt{P}$ occurs, which is not the case for $\mathcal{L}_{\LA}$.
\end{rmk}

\begin{defn}\label{DefLAP}
The theory $\LAP$ is $\LA$ plus the following axioms:
\begin{itemize}
 \item $\row (A)=\column (A) \rightarrow \power (0,A)=I_{\row (A)}$.
 \item $\row (A)=\column (A) \rightarrow \power (n+1,A)=\power (n,A)*A$
\end{itemize}
(Note that \cite{The proof complexity of linear algebra} omitted the assumption $\row (A)=\column (A)$ formally.)
\end{defn}

Actually, most theorems of $\LA$ (resp. $\LAP$) presented in \cite{The proof complexity of linear algebra} can be proved in the following $\LA_{-}$ (resp. $\LAP_{-}$):

\begin{defn}
$\mathcal{L}_{\LA_{-}}$ (resp. $\mathcal{L}_{\LAP_{-}}$) is defined similarly as $\mathcal{L}_{\LA}$, except that the function symbol $(\cdot)^{-1}$ is removed from $\mathcal{L}_{0}$ (resp. $(\cdot)^{-1}$ is removed from and $\mathtt{P}$ is added to $\mathcal{L}_{0}$).
Note that infinitely many function symbols of $\lambda$-type are removed.

 Let $\LA_{-}$ (resp. $\LAP_{-}$) be the theory obtained by collecting the axioms of $\LA$ which are $\mathcal{L}_{\LA_{-}}$-formulae (resp. $\mathcal{L}_{\LAP}$-formulae). 
 In particular, $\LA_{-}$ does not have Definition \ref{DefLA}(\ref{inverseaxiom}), and therefore $\fieldsort$-universes can be general rings.
\end{defn}

Thus we develop some theorems in $\LA_{-}$ first, and show that basic treatment of polynomials can be done in the theory (namely, Proposition \ref{[]pol is an interpretation})
Towards it, in order to describe the arithmetic of polynomials neatly, we utilize the notion of \textit{interpretation}:

\begin{defn}\label{interpretation}
 Let $\mathcal{L}$ be a $k$-sorted first-order language, and $s_{1}, \ldots, s_{k}$ be its sorts.
 Similarly, let $\mathcal{L}'$ be a $k'$-sorted first-order language, and $s'_{1}, \ldots, s'_{k'}$ be its sorts.
 We assume $\mathcal{L}$ (resp. $\mathcal{L}'$) includes the equality symbol $=_{s}$ for each $s$ of its sorts.
 Now, suppose we have an $\mathcal{L}^{\prime}$-theory $T^{\prime}$.
 \textit{A syntactic interpretation of $\mathcal{L}$ by $T^{\prime}$} are maps 
 \begin{align*}
 \mathcal{I}_{1} &\colon \{s_{1}, \ldots, s_{k}\} 
 \rightarrow  \{\langle \bar{s}',\varphi \rangle \mid \mbox{$\bar{s}' \bar{\in} \{s'_{1}, \ldots, s'_{k}\}$, and $\varphi$ is an $\mathcal{L'}$-formula}\}\\
 \mathcal{I}_{2} &\colon \{\mbox{variables for $\mathcal{L}$-formulae}\} \rightarrow \{\mbox{tuples of variables for $\mathcal{L}'$-formulae}\}\\
 \mathcal{I}_{3} &\colon \mathcal{L} 
 \rightarrow  \{\mbox{$\mathcal{L'}$-formulae}\}
 \end{align*}
 satisfying the following (below, we omit subscripts of $\mathcal{I}$ for readability):
 \begin{enumerate}
  \item  Given a sort $s_{i}$ ($i \in [k]$), let $\mathcal{I}(s_{i})=\langle \sigma'_{i1}, \ldots, \sigma'_{il_{i}}, \alpha_{s_{i}} \rangle$.
   Then there are variables $x_{1}, \ldots, x_{l_{i}}$, where each $x_{j}$ is of sort $\sigma'_{ij}$, such that
   \begin{enumerate} 
     \item The free variables occurring in $\alpha_{s_{i}}$ are among $\{x_{1},\ldots,x_{l_{i}}\}$.
     \item $T^{\prime} \vdash \exists \bar{x}.\ \alpha_{s_{i}}(\bar{x})$.
     \end{enumerate}
     
  \item For each variable $x$ of sort $s_{i}$, $\mathcal{I}(x)$ is a tuple $(x_{1}, \ldots, x_{l_{i}})$ of variables of sorts $\sigma'_{i1}, \ldots, \sigma'_{il_{i}}$.
  Furthermore, $\mathcal{I}$ assigns disjoint tuples to different variables.
  Below, given a tuple $\bar{x}=(x_{1}, \ldots, x_{k})$ of varibales, we set $\mathcal{I}(\bar{x}) := (\mathcal{I}(x_{1}), \ldots, \mathcal{I}(x_{k}))$.
  \item $\beta_{R}(\mathcal{I}(x_{1},\ldots,x_{n}))$ for each predicate symbol 
  $R$ (including equalities) of arity $s_{i_{1}} \times \cdots \times s_{i_{n}}$, 
  where each 
  $x_{j}$ is an $s_{i_{j}}$-variable.
  \item\label{definablefunction} $\gamma_{f}(\mathcal{I}(x_{1},\ldots,x_{n}),\mathcal{I}(z))$ for each function symbol $f$ of type \\
  $s_{i_{1}} \times \cdots \times s_{i_{n}} \rightarrow s_{i_{n+1}}$, where each $x_{j}$ is an $s_{i_{j}}$-variable, $z$ is an $s_{i_{n+1}}$-variable, and
  \begin{align*}
  T^{\prime} \vdash &\forall \mathcal{I}(\bar{x}). \Bigg( \bigwedge_{j=1}^{n} \alpha_{s_{i_{j}}} (\mathcal{I}(x_{j})) \rightarrow \exists \mathcal{I}(z).\ \left(\alpha_{s_{i_{n+1}}}(\mathcal{I}(z)) \land \gamma_{f}(\mathcal{I}(\bar{x}),\mathcal{I}(z)) \right)\Bigg), \\
   T^{\prime} \vdash &\forall \mathcal{I}(\bar{x}), \mathcal{I}(z_{1}), \mathcal{I}(z_{2}). \Bigg( \bigwedge_{j=1}^{n}\alpha_{s_{i_{j}}}(\mathcal{I}(x_{j})) \land \bigwedge_{j=1}^{2}\alpha_{s_{i_{n+1}}}(\mathcal{I}(z_{j})) \land \\
   &\bigwedge_{j=1}^{2}\gamma_{f}(\mathcal{I}(\bar{x}),\mathcal{I}(z_{j}))\rightarrow \beta_{=_{s_{i_{n+1}}}}(\mathcal{I}(z_{1}), \mathcal{I}(z_{2})) \Bigg).
  \end{align*}
 \end{enumerate}
\end{defn}

  The interpretation above can be readily extended to $\mathcal{L}$-terms and $\mathcal{L}$-formulae:  
  \begin{defn}
  Consider the setting in Definition \ref{interpretation}.
  By structural induction on $\mathcal{L}$-terms, we define a formula $\Gamma_{\mathcal{I},t}(\mathcal{I}(\bar{x}),\mathcal{I}(y))$ for each $\mathcal{L}$-term $t$, a tuple $\bar{x}=(x_{1}, \ldots, x_{n})$ covering all the variables occurring in $t$, and a fresh variable $y$ of the output sort of $t$, as follows:
  \begin{enumerate}
   \item If $t$ is a variable $z$ of sort $s_{i}$, then let $y$ be a fresh variable of sort $s_{i}$, and let $\Gamma_{\mathcal{I},t}(\mathcal{I}(\bar{x}))$ be the conjunction saying $\mathcal{I}(z) = \mathcal{I}(y)$.
   \item If $t=f(t_{1},\ldots, t_{m})$, then let $\Gamma_{\mathcal{I},t}(\mathcal{I}(\bar{x}), \mathcal{I}(y))$ be the formula saying:
   \[\exists \mathcal{I}(y_{1}), \ldots, \mathcal{I}(y_{m}). \left(\bigwedge_{j=1}^{m} \Gamma_{\mathcal{I},t_{j}}(\mathcal{I}(\bar{x}),\mathcal{I}(y_{j})) \land \gamma_{f}(\mathcal{I}(y_{1}), \ldots, \mathcal{I}(y_{m}),\mathcal{I}(y)) \right) \]
   Here, each $y_{i}$ is a distinct fresh variable of the sort of the corresponding input of $f$.
  \end{enumerate}
  
  \end{defn}
 

\begin{defn}\label{newmodel}
Consider the setting in Definition \ref{interpretation}.
 For each $\mathcal{M} \models T^{\prime}$, we define an $\mathcal{L}$-structure $\mathcal{M}_{\mathcal{I}}$.
 For each $i \in [k]$ (recall that $k$ is the number of sorts of $\mathcal{L}$), let $M_{i}$ be the universe of the sort $s_{i}$ of $\mathcal{M}$.
 Now, $\mathcal{M}_{\mathcal{I}}$ is defined as follows:
 \begin{enumerate}
  \item The universe of $\mathcal{M}_{\mathcal{I}}$ is a tuple $\langle N_{1}, \ldots, N_{k} \rangle$ where each $N_{i}$ is for sort $s_{i}$ and given as follows:
  \[N_{i} := \{\langle a_{1} , \ldots, a_{l_{i}} \rangle \mid a_{j} \in M_{\sigma_{ij}} \ (j \in [l_{i}]) \ \mbox{and} \ \mathcal{M} \models \alpha_{s_{i}} (a_{1}, \ldots, a_{l_{i}})\}.\]
  \item Each relation symbol $R$ on $s_{i_{1}}\times \cdots \times s_{i_{n}}$ is interpreted by 
  \begin{align*}
  R_{\mathcal{I}} = \Large\{ \langle \bar{a}_{1}, \ldots, \bar{a}_{n} \rangle \mid &\forall j \in [n]. \bar{a}_{j} \in N_{i_{j}},\\
  &\mathcal{M} \models \bigwedge_{i \in [n] } \alpha_{s_{i}}(\bar{a}_{i}) \land \beta_{R}(\bar{a}_{1}, \ldots, \bar{a}_{n})\Large\}.
  \end{align*}
  \item Each function symbol $f(\bar{a}_{1}, \ldots, \bar{a}_{n})$ of type $s_{i_{1}} \times \cdots \times s_{i_{n}} \rightarrow s_{i_{n+1}}$ in $\mathcal{L}$ is interpreted by the graph:
  \begin{align*}
  G_{f,\mathcal{I}} := \Large\{ \langle \bar{a}_{1}, \ldots, \bar{a}_{n}, \bar{b} \rangle \mid &\forall j \in [n].\ \bar{a}_{j} \in N_{i_{j}},\\
    &\bar{b} \in N_{i_{n+1}}, \mbox{and} \ \mathcal{M} \models \gamma_{f} (\bar{a}_{1}, \ldots, \bar{a}_{n}, \bar{b})\Large\}.
    \end{align*}  
  By Definition \ref{interpretation}, it is guaranteed that for any $\bar{a}_{1}\in M_{s_{i_{1}}}$, \ldots, $\bar{a}_{n} \in M_{s_{i_{n}}}$, there exists a unique $\bar{b} \in M_{s_{i_{n+1}}}$ such that $\langle \bar{a}_{1}, \ldots, \bar{a}_{n},\bar{b}\rangle \in G_{f, \mathcal{I}}$.
  We denote such a unique $\bar{b}$ by $f_{\mathcal{I}}(\bar{a}_{1}, \ldots, \bar{a}_{n})$.
  
   Moreover, for each $\mathcal{L}$-term $t(x_{1}, \ldots, x_{k})$, $G_{t,\mathcal{I}}$ and $t_{\mathcal{I}}$ are analogously defined, and they coincide with the definable function on $\mathcal{M}_{\mathcal{I}}$ given by $t$.
 \end{enumerate}

\end{defn}

\begin{defn}
Consider the setting in Definition \ref{interpretation}.
Given a tuple $\bar{x}=(x_{1},\ldots,x_{n})$ of variables in $\mathcal{L}$, where each $x_{j}$ is of sort $s_{i_{j}}$, and an assignment $\bar{a}$ to $\bar{x}$ in $\mathcal{M}_{\mathcal{I}}$,
let $\mathcal{I}(\bar{a})$ be the induced assignment to $\mathcal{I}(\bar{x})$ in $\mathcal{M}$.
That is, if $\bar{a}$ assigns $\bar{a}_{j} \in N_{i_{j}}$ to $x_{j}$, then $\mathcal{I}(\bar{a})$ assigns $\bar{a}_{j}$ to the tuple $\mathcal{I}(x_{j})$.
\end{defn}

\begin{defn}\label{formulaeinterpretation}
We continue in the setting of Definition \ref{interpretation}.
For any $\mathcal{L}$-formula 
 \[\varphi(x_{1}, \ldots, x_{n}),\]
   define an $\mathcal{L}^{\prime}$-formula 
  \[\varphi_{\mathcal{I}}(\mathcal{I}(x_{1}), \ldots, \mathcal{I}(x_{n}))\]
 inductively as follows:
 \begin{enumerate}
  \item If $\varphi$ is $R(t_{1},\ldots,t_{m})$, where $R$ is an $m$-ary relation symbol of $\mathcal{L}$ and each $t_{i}$ is an $\mathcal{L}$-term, then $\varphi_{\mathcal{I}}$ is:
  \[ \exists \mathcal{I}(y_{1},\ldots,y_{m}). \Bigg( \bigwedge_{j=1}^{m} G_{t,\mathcal{I}}(\mathcal{I}(\bar{x}),\mathcal{I}(y_{j})) \land \beta_{R}(\mathcal{I}(y_{1},\ldots,y_{m})) \Bigg).\]
  
  \item If $\varphi$ is of the form $\psi \land \eta$, $\psi \lor \eta$, or $\lnot \psi$, then $\varphi_{\mathcal{I}}$ is:
  \[  \psi_{\mathcal{I}} \land \eta_{\mathcal{I}}, \psi_{\mathcal{I}} \lor \eta_{\mathcal{I}}, \lnot \psi_{\mathcal{I}} \] 
  respectively.
 
 \item If $\varphi$ is of the form $\forall x \psi(x)$ or $\exists x \psi(x)$, and $x$ is an $s_{i}$-variable, then $\varphi_{\mathcal{I}}$ is:
 \[ \forall \mathcal{I}(x). \left(\alpha_{s_{i}}(\mathcal{I}(x)) \rightarrow \psi_{\mathcal{I}}(\mathcal{I}(x))\right), \quad
  \exists \mathcal{I}(x). \left(\alpha_{s_{i}}(\mathcal{I}(x)) \land \psi_{\mathcal{I}}(\mathcal{I}(x)) \right).\]
  respectively.
 \end{enumerate}

\end{defn}

\begin{prop}
 Let $\mathcal{M} \models T^{\prime}$ and $\mathcal{I}$ be an interpretation of $\mathcal{L}$ in $T^{\prime}$. For any $\mathcal{L}$-formula $\varphi(\bar{x})$ and assignment $\bar{a}$ to $\bar{x}$ in $\mathcal{M}_{\mathcal{I}}$, 
 \[\mathcal{M}_{\mathcal{I}} \models \varphi (\bar{a}) \Leftrightarrow \mathcal{M} \models \varphi_{\mathcal{I}}(\mathcal{I}(\bar{a})),\]
\end{prop}

The proof is straightforward, and we omit it.

\begin{defn}
Consider the setting in Definition \ref{formulaeinterpretation}. 
Let $T$ be an $\mathcal{L}$-theory.
$\mathcal{I}$ is \textit{an interpretation of $T$ by $T^{\prime}$} if and only if 
$T^{\prime}\vdash \varphi_{\mathcal{I}}$ for all $\varphi \in T$, that is,
$\mathcal{M}_{\mathcal{I}} \models T$ for all $\mathcal{M} \models T^{\prime}$.
\end{defn}

\begin{rmk}
Especially, if $\mathcal{L}=\mathcal{L}^{\prime}$, $T=T^{\prime}$, and $\mathcal{M} \models T$, we obtain a sequence of applications of the interpretation $\mathcal{M}_{\mathcal{I}}, (\mathcal{M}_{\mathcal{I}})_{\mathcal{I}}$, all giving the models of $T$. 
\end{rmk}

\section{$\Sigma^{B}_{0}$-formulae and -definablity}\label{SigmaB0-formula and -definability}
In this section, we introduce the notion of $\Sigma^{B}_{0}$ with respect to the  language $\mathcal{L}_{\LAP}$ analogously to \cite{Cook}, and we show that $\LAP$ has quantifier elimination for $\Sigma^{B}_{0}$-formula.
Armed with this helpful fact, we proceed to formalize polynomials, rationals, and their basic treatments in $\LAP$ in the subsequent sections.

\begin{defn}
For an $\mathcal{L}_{\LAP}$-formula $\varphi$,
 $\varphi$ is a \textit{$\Sigma^{B}_{0}$-formula} (abbreviated as $\varphi \in \Sigma^{B}_{0}$) if and only if all of its quantifiers are of $\indexsort$ and bounded by terms.  
\end{defn}


We observe that $\LAP$ has quantifier elimination for $\Sigma^{B}_{0}$-formulae:

\begin{lemma}[$\LAP$]\label{characteristicfunction}
 For each $\varphi(\bar{x}) \in \Sigma^{B}_{0}$, there exists an $\mathcal{L}_{\LAP}$-term $\delta_{\varphi}(\bar{x})$ outputting $\fieldsort$ elements such that 
 \begin{align*}
 \LAP \vdash &\varphi(\bar{x}) \rightarrow \delta_{\varphi}(\bar{x}) = 1\\
             &\lnot \varphi(\bar{x})\rightarrow \delta_{\varphi}(\bar{x}) = 0.
 \end{align*}
 In particular, $\varphi$ is equivalent to an open formula;
 \[\LAP \vdash \varphi(\bar{x}) \leftrightarrow \delta_{\varphi}(\bar{x})=1.\]
\end{lemma}

\begin{proof}
 By induction on the construction of $\varphi$.
 Here, for simplicity, we take $\land$, $\lnot$, $\forall$ as the basic logical symbols.
 \begin{enumerate}
  \item When $\varphi(\bar{x})\equiv R (\bar{x})$, where $R$ is an atomic formula on $\indexsort$ (i.e. of the form $s_{1}=s_{2}$ or $s_{1} \leq s_{2}$ where each $s_{i}$ is a term outputting $\indexsort$ elements) let 
  \[\delta_{\varphi}(\bar{x}) := \cond^{y_{1}=y_{2}, (y_{1},y_{2})}_{\fieldsort}(s_{1}(\bar{x}),s_{2}(\bar{x}),1,0).\]
  or
  \[\delta_{\varphi}(\bar{x}) := \cond^{y_{1}\leq y_{2}, (y_{1},y_{2})}_{\fieldsort}(s_{1}(\bar{x}),s_{2}(\bar{x}),1,0)\]
  respectively.
  \item When $\varphi(\bar{x}) \equiv R(\bar{x})$, where $R$ is an atomic formula of $\fieldsort$ (i.e. of the form $s_{1}=s_{2}$ where each $s_{i}$ is a term outputting $\fieldsort$ elements), let 
  \[\delta_{\varphi}(\bar{x}) := 1-(s_{1}-s_{2})^{-1}(s_{1}-s_{2}).\]
  (Note that if $s_{1}=s_{2}$, $\delta_{\varphi}(\bar{x})=1$ regardless what value $(s_{1}-s_{2})^{-1}$ takes.)
  \item When $\varphi(\bar{x}) \equiv R (\bar{x})$, where $R$ is an atomic formula of $\matrixsort$ (i.e. of the form $s_{1}=s_{2}$ where each $s_{i}$ is a term outputting $\matrixsort$ elements) let $\delta_{\varphi}(\bar{x})$ be:
  \begin{align*}
  \delta_{(\row (s_{1})=\row (s_{2}))}(\bar{x})\delta_{(\column (s_{1})=\column (s_{2}))}(\bar{x})\prod_{i=0}^{\mmax\{\row (s_{1}),\row (s_{2})\}} \prod_{j=0}^{\mmax\{\column (s_{1}),\column (s_{2})\}} \delta_{(\extract (s_{1},i,j) = \extract (s_{2},i,j))} (\bar{x},i,j).
  \end{align*}
   (Recall that iterated product $\prod$ can be realized using the powering function $P$ in $\LAP$. See section 4.1 of \cite{The proof complexity of linear algebra}.)
  
  \item When $\varphi(\bar{x}) = \psi_{1}(\bar{x}) \land \psi_{2}(\bar{x})$,  let $\delta_{\varphi}(\bar{x}) := \delta_{\psi_{1}}(\bar{x}) \delta_{\psi_{2}}(\bar{x})$.
  
  \item When $\varphi(\bar{x}) = \lnot \psi(\bar{x})$, let $\delta_{\varphi}(\bar{x}) := 1-\delta_{\psi}(\bar{x})$.
  
  \item When $\varphi(\bar{x}) = \forall y \leq u(\bar{x}).\ \psi(\bar{x},y)$, let 
  \[\delta_{\varphi}(\bar{x}) := \prod_{i =0}^{u(\bar{x})} \delta_{\psi}(\bar{x},i).\]

 \end{enumerate}
 The verifications of implications
  \begin{align*}
 \LAP \vdash &\varphi(\bar{x}) \rightarrow \delta_{\varphi}(\bar{x}) = 1\\
             &\lnot \varphi(\bar{x})\rightarrow \delta_{\varphi}(\bar{x}) = 0.
 \end{align*}
 are straightforward.
\end{proof}


\begin{cor}
 $\LAP$ admits $\Sigma^{B}_{0}$-induction, that is, for each $\Sigma^{B}_{0}$-formula $\varphi(x)$ (where $x$ is a variable of $\indexsort$),
 \[\LAP \vdash \left( \varphi(0) \land \forall x.(\varphi(x) \rightarrow \varphi(x+1)) \right) \rightarrow \forall x. \varphi(x).\]
\end{cor}

Later, we formalize the matrix rank function in $\LAP$, but if we want an $\indexsort$-output, we cannot construct an $\mathcal{L}_{\LAP}$-term for it. 
Instead, we use $\Sigma^{B}_{0}$-definable functions in $\LAP$.

\begin{defn}
When an $\mathcal{L}_{\LAP}$-theory $T$ and a $\Sigma^{B}_{0}$-formula $\varphi(\bar{x},\bar{a},\bar{A},\bar{y},\bar{b},\bar{B})$ ($\bar{x},\bar{y} \in \MM$, $\bar{a}, \bar{b} \in \FF$, $\bar{A},\bar{B} \in \Mat_{\FF}$) satisfy
\[T \vdash \forall \bar{x}\bar{a}\bar{A} \exists! \bar{y}\bar{b}\bar{B}. \varphi(\bar{x}\bar{a}\bar{A}\bar{y}\bar{b}\bar{B}),\]
we say $\varphi$ \textit{$\Sigma^{B}_{0}$-defines a function over $T$}. 
A function $f$ on a generic structure $(\MM,\FF,\Mat_{\FF})$ is \textit{$\Sigma^{B}_{0}$-definable over $T$} when its graph is defined by a $\Sigma^{B}_{0}$-formula (without parameters) defining a function over $T$.
 
 If $\varphi$ is open, we say $\varphi$ \textit{open-defines a function over $T$}.
 And the corresponding function is said to be \textit{open-definable}.
\end{defn}

\begin{eg}
 $\deg$ is a $\Sigma^{B}_{0}$-definable function over $\LAP$.
\end{eg}

By Lemma \ref{characteristicfunction};

\begin{cor}
Let $T$ be an $\mathcal{L}_{\LAP}$-theory extending $\LAP$.
Then a function is $\Sigma^{B}_{0}$-definable over $T$ if and only if it is open-definable over $T$.
\end{cor}
\begin{cor}
  The function $\deg$ is open-definable over $\LAP$.
\end{cor}

Although it is not the case in general that a given $\Sigma^{B}_{0}$-definable function with an $\indexsort$ output can be represented by an $\mathcal{L}_{\LAP}$-term, it is still possible to construct a term with a $\matrixsort$ output which can be identified the original $\indexsort$ output:

\begin{defn}
An \textit{index vector} is a vector $v \in \Mat_{\FF}(n,1)$ such that
\[ \exists i \leq n (v_{i1} = 1 \land \forall j \leq n (j \neq i \rightarrow v_{j1} = 0)).\]
We identify an index vector $v$ with the unique index $i \leq n$ such that $v_{i1} = 1$.
Note that such an index $\iota(v)$ is $\Sigma^{B}_{0}$-defined as
\[ \iota(v) := \min\{i \leq \row(v) \mid v_{i1} =1\}.\] 
\end{defn}

\section{A formalization of polynomials in $\LAP$ and the interpretation $\llbracket \cdot \rrbracket_{pol}$}\label{A formalization of polynomials in LAP}

In this section, we discuss formalization of polynomials in $\LAP$.
For convenience and readability, we fix a generic model $\mathcal{M}=\langle \MM, \FF, \Mat_{\FF} \rangle$ of $\LAP$, where
\begin{itemize}
 \item $\FF$ denotes a universe of $\fieldsort$.
 \item $\MM$ denotes a universe of $\indexsort$.
 \item $\Mat_{\FF}$ denotes a universe of $\matrixsort$.
\end{itemize}

Furthermore, for each $n,m \in \MM$, $\Mat_{\FF}(n,m)$ be the definable subset of a universe of $\matrixsort$, collecting the elements $A$ such that $\row (A)=n$ and $\column (A)=m$.

Moreover, for $a,b \in \MM$, we set
\begin{itemize}
 \item $[a,b] := \{x \in \MM \mid a \leq x \leq b\}$.
 \item $\left[ a,b \right[ := \{x \in \MM \mid a \leq x < b\}$.
 \item $\left] a,b \right]$ and $\left] a,b \right[$ are defined similarly.
 \item $[b] := [1,b]$.
\end{itemize}

Note that the notations of intervals do not clash the notation $[k]$ ($k \in \NN$) introduced in the beginning of this section, and the above intervals are all definable by open formulae in $\LAP$.

The main goal of this section is to establish the interpretation $\llbracket \cdot \rrbracket_{pol}$ of $\LAP_{-}$ by $\LAP$, which produces the model $\langle \MM, \FF[X], \Mat_{\FF[X]} \rangle \models \LAP_{-}$ from $\langle \MM, \FF, \Mat_{\FF} \rangle \models \LAP$.

First, we informally explain the coding of polynomials and matrices with polynomial coefficients.
Like in \cite{The proof complexity of linear algebra}, we code a polynomial 
\[a_{n}X^{n}+ \cdots + a_{0}\]
 (where $n \in \MM$, and each $a_{i}$ is in $\FF$) by a vector 
\begin{align*}
\begin{bmatrix}
a_{0}, \cdots, a_{n} 
\end{bmatrix}^{t}
 \in \Mat_{\FF}(n+1,1).
\end{align*} 
(Note that the order of coefficients are reversed compared to the presentation of \cite{The proof complexity of linear algebra}, but there is no essential difference.
Furthermore, we sometimes use transposes as above meta-theoretically just for saving space, but the operation transpose itself and its basic treatment are already formalized in $\LA$ in \cite{The proof complexity of linear algebra}, so there is no danger for us to ignore the distinction between meta-theoretical transpose and formalized transpose.)

More generally, a matrix $\widetilde{A}=(f_{ij})_{i\in [m], j \in [n]}$ of polynomials $f_{ij}(X)$ of degree $\leq d \in \MM$ can be coded by a pair $(A,d)$ of the index $d \in \MM$ and the block matrix  
\begin{align*}
A=
\begin{bmatrix}
A_{0} \\
\vdots \\
A_{d} 
\end{bmatrix}
 \in \Mat_{\FF}(m(1+d),n)
\end{align*}
where 
\[A_{k} = (c_{ijk})_{i \in [m], j \in [n]} \in \Mat_{\FF}(m,n) \ \& \ f_{ij} = \sum_{k=0}^{d}c_{ijk}X^{k},\]
that is, informally, we can write
\[\widetilde{A} = A_{0} + A_{1}X + \cdots + A_{d}X^{d} \quad (A_{k} \in \Mat_{\FF}(m,n)).\]

Now, we formally implement the coding described above.
Let $\FF[X]$ be the subset of $\Mat_{\FF}$ of the codes of polynomials:
\begin{defn}[$\LAP$]
A predicate ``$A \in \FF[X]$'' ($A \in \Mat_{\FF}$) is defined as follows; 
 \[A \in \FF[X] :\leftrightarrow \column (A)=1.\]
\end{defn}

The coefficients and the degree of polynomials are defined as follows:

\begin{defn}[$\LAP$]
For $f,g \in \FF[X]$, define the following $\mathcal{L}_{\LAP}$-term and relations:
\begin{enumerate}
 \item $\coeff(f,k) := \extract (f,k+1,1)$.\\
 (Informally, $\coeff(f,k)$ returns the coefficient of degree-$k$ in the polynomial $f$)
 
 \item $\deg(f) = j : \leftrightarrow \lnot \coeff(f,j) =0 \land \forall i \in \left]j,\row (f)\right].\ \coeff(f,i) = 0$.\\
 (By open induction, $\LAP$ can prove its totality for \\
 $f \in \FF[X] \setminus \{0_{k1} \mid k \in \MM\}$.)
 \item $\deg(f) = -\infty : \leftrightarrow \forall i \leq \row (f). \coeff(f,i) = 0$. \\
 ($-\infty$ is just a symbol for convenience.)
\end{enumerate}
\end{defn}

Note that $\deg(f)=j$ and $\deg(f)=-\infty$ are not open formulae but $\Sigma^{B}_{0}$.

Furthermore, we formalize matrices with polynomial coefficients as follows:
\begin{defn}[$\LAP$]\label{matrixofpolynomials}
A predicate ``$(A,d) \in \Mat_{\FF[X]}(m,n)$'' and ``$(A,d) \in \Mat_{\FF[X]}(m,n)$'' ($A \in \Mat_{\FF}$, $m,n,d \in \MM$) are defined as follows; 
 \begin{align*}
 (A,d) \in \Mat_{\FF[X]} &:\leftrightarrow \rem(\row (A),1+d)=0.\\
 (A,d) \in \Mat_{\FF[X]}(m,n) &:\leftrightarrow \row (A)=m(1+d) \land \column (A)=n.
 \end{align*}
\end{defn}

\begin{rmk}
\[\LAP \vdash A \in \FF[X]  \leftrightarrow  (A,\row (A)-1) \in \Mat_{\FF[X]}(1,1).\]
\end{rmk}

We adopt the following notation throughout this paper;
\[ \sum_{j=a}^{b} t(j) := \sum( \lamt_{kl}\langle 1,b-a+1, t(a+l-1) \rangle),\]
where $a,b \in \MM$ and $t(j)$ is a term outputting $\fieldsort$ elements.

Now, we define the interpretation $\llbracket \cdot \rrbracket_{pol}$ of $\LAP_{-}$ by $\LAP$.
We will define $\mathcal{I}_{1}$, $\mathcal{I}_{2}$, and $\mathcal{I}_{3}$ in Definition \ref{interpretation} in the subsequent subsections.

\subsection{The definitions of $\mathcal{I}_{1}$ and $\mathcal{I}_{2}$}\label{I1andI2onL0}
Set $\mathcal{I}_{1}$ as follows:
\begin{enumerate}
 \item $\indexsort \mapsto \langle \indexsort, x=x \rangle$.
 \item $\fieldsort \mapsto \langle \matrixsort, f \in \FF[X] \rangle$.
 \item $\matrixsort \mapsto \langle \matrixsort, \indexsort, (A,d) \in \Mat_{\FF[X]} \rangle$.
\end{enumerate}

As for $\mathcal{I}_{2}$, we take the following map: let $x$ be a variable.
\begin{enumerate}
\item For $x$ of $\indexsort$, assign itself.
\item For $x$ of $\fieldsort$, assign a fresh variable of $\matrixsort$. We denote it by $h_{x}$.
\item For $x$ of $\matrixsort$, assign a tuple $(x,d)$, where $d$ is a fresh variable of $\indexsort$.
 We denote it by $d_{x}$.
\end{enumerate}

\subsection{$\mathcal{I}_{3}$ on $\mathcal{L}_{0}$}\label{I3onL0forpoly}
To define $\mathcal{I}_{3}$ on $\mathcal{L}_{\LAP_{-}}$, we follow induction on the construction of $\mathcal{L}_{\LAP_{-}}$.
For a function symbol $f(x_{1},\ldots,x_{l})$, instead of $\gamma_{f}$ in Definition \ref{interpretation}, we define an $\mathcal{L}_{\LAP}$-definable function $f_{pol}(\mathcal{I}(x_{1}), \ldots, \mathcal{I}(x_{l}))$ outputting an element of $N_{i}$ in Definition \ref{newmodel} corresponding to the output sort of $f$.
Then $\gamma_{f}$ is naturally defined by the conjunction expressing $\mathcal{I}(z)=f_{pol}(\mathcal{I}(x_{1}), \ldots, \mathcal{I}(x_{l}))$.
In this subsection, we deal with $\mathcal{L}_{0}$.
As for $\mathcal{L}_{\indexsort}$ (cf. Definition \ref{DefLA}), $\mathcal{I}_{3}$ is an ``identity mapping:''

\begin{enumerate}
 \item $0_{\indexsort} \mapsto x=0_{\indexsort}$, $1_{\indexsort} \mapsto x=1_{\indexsort}$.
 \item $\odot\mapsto \odot(x_{1},x_{2})$, where $\odot=+_{\indexsort}, *_{\indexsort}, -_{\indexsort}, \divi, \rem$.
 \item $R \mapsto R (x_{1},x_{2})$, where $R$ is one of $=_{\indexsort}, \leq_{\indexsort}$.
 \end{enumerate}
 
Next, we consider the symbols related to $\fieldsort$.
It amounts to implement basic arithmetical operations on polynomials:
 \begin{enumerate}
 \item $=_{\fieldsort} \mapsto f =_{pol} g :\leftrightarrow \forall j\leq \mmax\{\row (f),\row (g)\}.\  \coeff(f,j)=\coeff(g,j)$.
 \item $(0_{\fieldsort})_{pol}:= \zeromat_{11}$, $(1_{\fieldsort})_{pol}:= \idmat_{1}$.

 \item $f(+_{\fieldsort})_{pol}g :=_{pol}\lamt_{kl} \langle \mmax\{\row (f),\row (g)\}, 1,\ \extract (f,k,1) + \extract (g,k,1) \rangle$.
 
 Note that $\mmax$ is already defined in \cite{The proof complexity of linear algebra}.
 We denote the RHS of $=_{pol}$ by just $f+_{pol}g$.
 \item $f(*_{\fieldsort})_{pol} g  := \conv(f,\row (g)-1) g$, where 
 \[ \conv(f,l) := \lamt_{ij} \langle \row (f) + l, 1+l, \extract (f,(i+1)-j,1) \rangle\]
 We denote the RHS of $=_{pol}$ by just $f*_{pol}g$.
 ($\conv(f,l)$ denotes the matrix representing $f$ as a convolution operating on degree-$l$ polynomials, that is, the Toeplitz matrix whose first column is 
 \begin{align*}
 \begin{bmatrix}
 f, 0, \cdots ,0
 \end{bmatrix}^{t}
  \in \Mat_{\FF}(\row (f)+l,1). \quad )
 \end{align*}
 
 \item $(-_{\fieldsort})_{pol}f :=\lamt_{kl} \langle \row (f), 1,\ -\extract (f,k,1) \rangle$.
 \end{enumerate}

Lastly, we deal with the remaining symbols in $\mathcal{L}_{0}$ (of $\LAP_{-}$):

\begin{enumerate}
 \item  $\row_{pol}(A,d) := \divi(\row (A),1+d)$.
 \item $\column_{pol}(A,d):= \column (A)$. 
 \item  
 \[\extract_{pol}(A,d,i,j) := \lamt_{kl}\langle 1+d, 1,  \extract (\mcoeff(A,d,k-1),i,j)\rangle.\]
 Here, 
 \[\mcoeff(A,d,k) := \lamt_{hl}\langle \row_{pol}(A,d), \column_{pol}(A,d), \extract (A,\row_{pol}(A,d)k+h,l) \rangle.\]
 
 (Intuitively, if $(A,d)$ codes $(f_{ij})_{(i,j) \in [\row]\times[\column]}=A_{0}+\cdots+A_{d}X^{d} \in \Mat_{\FF[X]}$, then $\mcoeff(A,d,k)=A_{k}$ and $\extract_{pol}(A,d,i,j)=f_{ij} \in \FF[X]$.)
 
 \item $\summation_{pol}(A,d)$ is the term 
 \[\lamt_{kl}\left\langle 1+d, 1, \sum(\mcoeff(A,d,k-1)) \right\rangle.\]
 
 \item $=_{\matrixsort} \mapsto$
   \begin{align*}
   &\row_{pol}(A,d) = \row_{pol}(B,d^{\prime}) \land \column_{pol}(A,d) = \column_{pol}(B,d^{\prime}) \land \\
   &\forall i \in [\row (A)].\forall j \in [\column (A)].\ \extract (A,i,j) =\extract (B,i,j)
   \end{align*}
We denote the RHS by $(A,d)=_{pol}(B,d^{\prime})$.
It clashes with $=_{pol}$ for $=_{\fieldsort}$, but there is no danger of confusion, so we stick to this readable notation.

\item In this item, we abbriviate $\row_{pol}(A,d)$, $\column_{pol}(A,d)$ as $\row_{pol}$, $\column_{pol}$ respectively.
Furthermore, we suppress the subscript $\indexsort$ of the symbols in $\mathcal{L}_{\indexsort}$ for readability.
Let 
\begin{align*}
&\pconv(A,d,l)\\
 :=& \lamt_{ij} \Bigg\langle \row_{pol}*(d+1+l), \column_{pol}*(l+1),
 e\Big(\mcoeff \big(A,d,\\
 &\cond^{y_{2}\leq y_{1}, (y_{1},y_{2})}_{\indexsort}(y_{1},y_{2},y_{1}-y_{2},d+1) [y_{1},y_{2} \mapsto \divi(i-1,\row_{pol}),\divi(j-1,\column_{pol})] \big), \\ 
 &\rem(i-1,\row_{pol})+1, \rem(j-1,\column_{pol})+1\Big)\Bigg\rangle.
\end{align*}

(Intuitively, if $(A,d)$ codes $A_{0}+\cdots+A_{d}X^{d} \in \Mat_{\FF[X]}$, then $\pconv(A,d,l)$ is the convolution matrix for $(\row_{pol}\times \column_{pol})$-matrices in $\Mat_{\FF[X]}$ of degree $l$; precisely speaking:
\begin{itemize}
  \item $\pconv(A,d,l)$ is a Toeplitz block matrix consisting of $(d+1+l)\times (1+l)$-many $(\row_{pol}\times \column_{pol})$-sized blocks.
  \item The first column block matrix of $\pconv(A,d,l)$ is 
  \[ \begin{bmatrix}
A_{0} \\
\vdots \\
A_{d} \\
\zeromat_{\row_{pol} \column_{pol}}\\
\vdots \\
\zeromat_{\row_{pol} \column_{pol}}
\end{bmatrix} \quad ).\]
\end{itemize}

Furthermore, set
  \begin{align*}
  &\bold{Q}_{pol}(k,A,d) :=\\
   &\pconv(A,d,(k-1)d) \cdots \pconv(A,d,d) \pconv(A,d,0)  \idmat_{\row (A)}.
  \end{align*}
  (Note that we can formalize iterated multiplication of boundedly many matrices using powering function; see \cite{The proof complexity of linear algebra}.)
  
  Now, we set $\mathcal{I}_{3}(\mathtt{P})$ by
  \[\mathtt{P}_{pol}(k,A,d) := (\bold{Q}_{pol} (k,A,d) ,kd).\]
  We often write $(A,d)^{k}$ to denote $\mathtt{P}_{pol}(k,A,d)$ for readability.

\end{enumerate}

\begin{rmk}
It is straightforward to verify:
\begin{align*}
 \LAP \vdash \mathtt{P}_{pol}(k,A,d) \in \Mat_{\FF[X]}(m,m).
 \end{align*}
\end{rmk}

\subsection{Degree bounding terms $\bold{b}[t]$}\label{b[t] for L0}
To proceed to the definition of $\mathcal{I}_{3}$ on $\mathcal{L}_{k}$ ($k \geq 1$),
we simultaneously define a \textit{degree bounding term} $\bold{b}[t]$ for each $\mathcal{L}_{k}$-term ($k \geq 0$) $t$ of output sort $\fieldsort$ or $\matrixsort$, whose variables are among $\mathcal{I}(x)$ of the variables $x$ occurring in $t$.
The precise definition is as follows:
  \begin{enumerate}
  \item For each field-variable $a$, we set $\bold{b}[a] := \row (h_{a})$.
  Recall that $h_{a} = \mathcal{I}_{2}(a)$.
    \item For each matrix-variable $A$, set $\bold{b}[A] := d_{A}$.
    Recall that $(A,d_{A})=\mathcal{I}_{2}(A)$.

    \item we set $\bold{b}[0_{field}] := \bold{b}[1_{field}] := 0$.
    
    \item If $t$ is $t_{1} +_{field} t_{2}$, then $\bold{b}[t] := \mmax\{\bold{b}[t_{1}],\bold{b}[t_{2}]\}$.
    
    \item If $t$ is $t_{1} *_{field} t_{2}$, then
    $\bold{b}[t] := \bold{b}[t_{1}]+\bold{b}[t_{2}]$.
    
    \item If $t$ is $\sum (A)$ or $\extract (A,i,j)$, then
     $\bold{b}[t] :\equiv \bold{b}[A]$.
    
    \item If $t$ is $\mathtt{P}(m,A)$, then $\bold{b}[t] := m*\bold{b}[A]$.

     \item If $t$ is $\cond_{field}^{\alpha, \vec{x}}(t_{1},t_{2})$, set $\bold{b}[t] := \mmax\{\bold{b}[t_{1}], \bold{b}[t_{2}]\}$.
 \item If $t$ is $\lamt_{ij} \langle m,n,u \rangle$, set $\bold{b}[t]:\equiv\bold{b}[u]$.
    
   \end{enumerate} 
   
   \subsection{$\mathcal{I}_{3}$ on $\mathcal{L}_{k}$ and $\bold{b}[t]$ for $\mathcal{L}_{k}$-terms}
   Let $k \geq 1$, and assume $\mathcal{I}_{3}$ on $\mathcal{L}_{k-1}$ is already defined.
   Furthermore, assume that a tuple $f_{pol}(\mathcal{I}(x_{1}), \ldots, \mathcal{I}(x_{l}))$ of $\mathcal{L}_{\LAP}$-terms is defined for each $\mathcal{L}_{k-1}$-function symbol $f$, and $\mathcal{I}(f)$ is defined by $\mathcal{I}(z)=f_{pol}(\mathcal{I}(x_{1}), \ldots, \mathcal{I}(x_{l}))$.
   Note that $\tau \mapsto \tau_{pol}$ naturally extends to the case when $\tau$ is an $\mathcal{L}_{k-1}$-term.
      
   Define $\mathcal{I}_{3}$ on $\mathcal{L}_{k} \setminus \mathcal{L}_{k-1}$ as follows:
   \begin{enumerate}
 \item $\cond_{\indexsort}^{\alpha,\vec{x}}(\vec{x},x',x'') \mapsto \cond_{\indexsort}^{\alpha,\vec{x}}(\vec{x},x',x'')$.
 \item $\cond_{\fieldsort}^{\alpha,\vec{x}}(\vec{x},x',x'') \mapsto  \cond_{\fieldsort}^{\alpha, \vec{x}}(1,0) h_{x'} +\cond_{\fieldsort}^{\alpha,\vec{x}}(0,1) h_{x''}$.
  \item $\lamt_{ij} \langle m,n, u \rangle \mapsto$
   \begin{align*}
    &\Bigg(\lamt_{kl} \Big\langle (1+\bold{b}[u]) m_{pol}, 
     n_{pol} ,\\ 
     &\coeff \big(u_{pol}(k,\rem(l,m_{pol})),\divi(l,m_{pol})+1\big) \Big\rangle , \bold{b}[u] \Bigg),
    \end{align*} 

\end{enumerate}

 \subsection{The interpretation $\llbracket \cdot \rrbracket_{pol}$}
 We have completed the definition of an interpretation $\mathcal{I}$ of $\mathcal{L}_{\LAP_{-}}$ by $\mathcal{L}_{\LAP}$.
 We denote $\mathcal{I}(\cdot)$ and $(\cdot)_{\mathcal{I}}$ by $\llbracket \cdot \rrbracket_{pol}$.
 We observe the following:
\begin{prop}\label{[]pol is an interpretation}
 $\llbracket \cdot \rrbracket_{pol}$ is an interpretation of $\LAP_{-}$ by $\LAP$, that is, for each axiom $\varphi \in \LAP_{-}$,
 \[ \LAP \vdash \llbracket \varphi \rrbracket_{pol}.\]
\end{prop}

Since the formal verification is very long but tedious, so we omit the proof.
For a proof sketch, see \ref{proof of []pol}.

Recall that the definitions of the characteristic polynomial $p_{A}(X)$ (which will be expressed as $\ch(A)$ from now on) and the determinant $\bold{det} (A)$ of $A \in \Mat_{\FF}(m,m)$ in \cite{The proof complexity of linear algebra} are actually carried out in $\LAP_{-}$.
 By the interpretation $\llbracket \cdot \rrbracket_{pol}$, we have:
\begin{itemize}
 \item The characteristic polynomial $(\llbracket \ch \rrbracket_{pol}(A,d),m) \in \Mat_{\FF[X]}(m+1,1)$ for each $(A,d) \in \Mat_{\FF[X]}(m,m)$.
 \item The determinant $\llbracket \bold{det} \rrbracket_{pol} (A,d) \in \FF[X]$ for each $(A,d) \in \Mat_{\FF[X]}(m,m)$.
\end{itemize}

We end this section by showing the following formalized identity theorem for polynomials for future use:

\begin{defn}[$\LAP$]
Let $f,g \in \FF[X]$.
We define $f(g) := \sum_{i=0}^{d} c_{i}g^{i}$, where $f=_{pol}[c_{0}, \ldots, c_{d}]^{t} \in \FF$.
Note that $g^{i}$ uses $\mathtt{P}_{pol}$, and $f(g)$ is well-defined.
\end{defn}

\begin{lemma}[$\LAP$]\label{simplesubstitutionishom}
For $g \in \FF$, $\FF[X] \rightarrow \FF[X];\ f \mapsto f(g)$ is a ring homomorphism.
Furthermore, if $M \in \Mat_{\FF[X]}(n,1)$, then 
\[\left(\sum_{i=1}^{n} f_{i}\right)(g) = \sum_{i=1}^{n} f_{i}(g) \quad \mbox{and} \quad \left(\prod_{i=1}^{n} f_{i}\right)(g) = \prod_{i=1}^{n} f_{i}(g). \]
Here, note that $\prod_{i=1}^{n} f_{i}(g)$ is computed by inputting $f_{i}(g)$'s for the interpretation of the function $[a_{1},\ldots, a_{n}] \mapsto \prod_{i=1}^{n}a_{i}$.
\end{lemma}

\begin{prop}[$\LAP$]\label{identitytheorem}
Let $f(X) \in \FF[X]$, $f \not=_{pol} 0_{pol}$, and $d=\deg f \in \MM$.
Let $v \in \Mat_{\FF}(d+1,1)$.
We denote $\extract(v,i)$ by $v_{i}$ for $i \in [d+1]$.
Assume $v$ satisfies $v_{i} \neq v_{j}$ for $i \neq j \in [d+1]$.
If $f(v_{i})=_{pol} 0_{pol}$ for all $i \in [d+1]$, then $f =_{pol} 0_{pol}$.
\end{prop}

\begin{proof}
We work in $\LAP$.
By Proposition \ref{[]pol is an interpretation}, we have 
\[\langle \MM, \FF[X], \Mat_{\FF[X]} \rangle \models \LAP_{-}.\]
Thus we can treat polynomials $g_{j}(X):=\prod_{i=1}^{j}(X-v_{i})$ $(j \in [d+1])$ in $\FF[X]$. (We also define $g_{0}:=1_{pol}$ in $\FF[X]$.)

For $i=1, \ldots, d$, let 
\begin{align*}
V_{i}&:=\begin{bmatrix}
(X+v_{i})^{d+1-i}, \ldots, (X+v_{i})^{0}
\end{bmatrix} \in \Mat_{\FF}(d+2-i, d+2-i), \\
U_{i}&:=\begin{bmatrix}
(X-v_{i})^{d-i}, \ldots, (X-v_{i})^{0}
\end{bmatrix} \in \Mat_{\FF}(d+1-i, d+1-i).
\end{align*}
Note that each polynomial $(X+v_{i})^{j}$ and $(X-v_{i})^{j}$ are regarded as the vector of the coefficients here.

Furthermore, let
\[S_{i} := \begin{bmatrix}
\zeromat_{d+1-i, 1} & I_{d+1-i}
\end{bmatrix} \in \Mat_{\FF}(d+1-i, d+2-i).\]

Put $f=_{pol}[a_{0}, \ldots, a_{d}]^{t}=:w$, where $a_{i} \in \FF$.
Then 
\[U_{d}S_{d}V_{d}\cdots U_{1}S_{1}V_{1}w \in \Mat_{\FF}(1,1),\]
which can be regarded as a scalar $c$ in $\FF$.
Then $f(X)=c g_{d}(X)$ holds.

Indeed, for polynomials $h_{j}$ represented by the vectors 
\[w_{j}:=U_{j}S_{j}V_{j} \cdots U_{1}S_{1}V_{1}w \in \Mat_{\FF}(d+1-j,1) \quad (0 \leq j \leq d),\]
 we can show $f(X) =_{pol} h_{j}*_{pol} g_{j}$ by induction on $j=0, \ldots, d$ as follows.

The base case $j=0$ is trivial.

Assume that the claim holds for $j < d$.
Then, by Induction Hypothesis, we have $f=_{pol} h_{j} *_{pol} g_{j}$.
It suffices to show $h_{j}=_{pol} h_{j+1} *_{pol} (X-v_{j+1})$.
First, by Induction Hypothesis, we have
\[0_{pol}=_{pol} f(v_{j+1}) =_{pol} h_{j}(v_{j+1}) *_{pol} g_{j}(v_{j+1}).\]
By Lemma \ref{simplesubstitutionishom} and that $\LAP$ can show that 
\[\prod_{i=1}^{l}a_{i}= 0 \rightarrow \exists i \in [l].\ a_{i}=0\]
 for $\fieldsort$-elements $a_{i}$,
 we have $g_{j}(v_{j+1}) \neq_{pol} 0_{pol}$, and thus 
\begin{align}\label{anotherroot}
h_{j}(v_{j+1})=_{pol}0_{pol}
\end{align}
follows.

On the other hand, by definition, we have $w_{j+1}=U_{j+2}S_{j+1}V_{j+1} w_{j}$.
Let $w_{j}=[a_{0}, \ldots, a_{d-j}]^{t}$.
Then $V_{j+1}w_{j}$ represents the polynomial 
\begin{align}\label{defofhj}
\widetilde{h}_{j}:=\sum_{i=0}^{d-j}a_{i}(X+v_{j+1})^{i}.
\end{align}
By the equality (\ref{anotherroot}), we obtain $\widetilde{h}_{j}(0)=_{pol}0$, which means $V_{j+1}w_{j}$ is of the form:
\[V_{j+1}w_{j}=[0, b_{1}, \ldots, b_{d-j}]^{t}.\]
Thus $S_{j+1}V_{j+1} w_{j}= [b_{1}, \ldots, b_{d-j}]^{t}$ represents the polynomial 
\[\widetilde{h}_{j+1}:=\sum_{i=0}^{d-j-1} b_{i+1}X^{i},\]
 that is, 
\[\widetilde{h}_{j} =_{pol} \widetilde{h}_{j+1} *_{pol} X.\]
Since $U_{j+1}S_{j+1}V_{j+1}w_{j}$ represents the polynomial $\widetilde{h}_{j+1}(X-v_{j+1})$, and it is at the same time $h_{j+1}$ by definition, we have
\[\widetilde{h}_{j}(X-v_{j+1})=_{pol} h_{j+1} *_{pol} (X-v_{j+1}).\]

Following the definition (\ref{defofhj}) and Lemma \ref{simplesubstitutionishom}, we obtain $\widetilde{h}_{j}(X-v_{j+1})=_{pol} h_{j}$.
Thus the claim follows.

\end{proof}

\section{Rational functions}\label{Rational functions}

So far, we have shown that $\LAP$ can treat $\FF[X]$ and $\Mat_{\FF[X]}$ properly and prove that they satisfy $\LAP_{-}$. 
 That is, we obtain the new model $(\MM, \FF[X], \Mat_{\FF[X]})$ of $\LAP_{-}$.

Now, since we use some properties of fields in order to formalize $\rank$ function, we develop a formalization of rational functions in $\LAP$, too.
Once we have developed it, we show that rational functions and the matrices of rational coefficients satisfy $\LAP$ itself.
Then, applying the interpretation $\llbracket \cdot \rrbracket_{pol}$ above again, we can obtain a model which can be safely denoted by $(\MM, \FF(X)[Y], \Mat_{\FF(X)[Y]})$.
Since a natural embedding $\FF[X] \rightarrow \FF(X)$ can be established in $\LAP$, we can treat bivariate polynomials in $\LAP$.
Repeating the argument, we are able to treat multivariate polynomials $f \in \FF[X_{1}, \ldots, X_{k}]$ for any $k \in \NN$ in $\LAP$.

From now on, we abbreviate $*_{pol}$ if it is clear from the context.

We code a rational function $f/g$ (where $f,g$ are polynomials) by an ordered pair $(f,g)$.

\begin{defn}[$\LAP$]
 We define a $\Sigma^{B}_{0}$-predicate $(f,g) \in \FF(X)$ as follows;
 \[ (f,g) \in \FF(X) : \leftrightarrow f\in \FF[X] \land g\in \FF[X] \land g \neq_{pol} 0. \]

\end{defn} 

\begin{defn}[$\LAP$]
We define the following $\mathcal{L}_{\LAP}$-terms:
 \begin{enumerate}
 \item $\num(f,g) := f$.
 \item $\den(f,g) := g$.
 \end{enumerate}
\end{defn}


Next, we consider formalization of matrices of rational functions.
We code a matrix $A=(f_{ij}/g)_{i \in [m], j \in [n]}$ of rationals $f_{ij}(X)$ of degree $\leq d \in \MM$ by a triple $(g,A,d)$ where $g$ is the common denominator, and $(A,d)$ is a code of the matrix $(f_{ij})_{i \in [m], j \in [n]}$ of polynomials described in \S \ref{A formalization of polynomials in LAP}.
Informally, we write
\[A = \frac{1}{g}\left(A_{0} + A_{1}X + \cdots + A_{d}X^{d} \right) \quad (A_{k} \in \Mat_{\FF}).\]
Note that $d \in \MM$.

\begin{defn}[$\LA_{-}$]
For $g,A \in \Mat_{\FF}$ and $m,n,d \in \MM$,
an open predicate $(g,A,d) \in \Mat_{\FF(X)}(m,n)$  is defined as follows; 
 \[(g,A,d) \in \Mat_{\FF(X)}(m,n) : \leftrightarrow g \in \FF[X]\setminus\{0\} \land (A,d) \in \Mat_{\FF[X]}(m,n).\]
 Let $\coeff_{rat}(g,A,d,k) := (\coeff(A,d,k) ,g) \in \FF(X)$. 
\end{defn}

Now, we define the interpretation $\llbracket \cdot \rrbracket_{rat}$ of $\LAP_{-}$ by $\LAP$.
We will define $\mathcal{I}_{1}$, $\mathcal{I}_{2}$, and $\mathcal{I}_{3}$ in Definition \ref{interpretation} in the subsequent subsections.

\subsection{The definitions of $\mathcal{I}_{1}$ and $\mathcal{I}_{2}$}
Set $\mathcal{I}_{1}$ as follows:
\begin{enumerate}
 \item $\indexsort \mapsto \langle \indexsort, x=x \rangle$.
 \item $\fieldsort \mapsto \langle \matrixsort, \matrixsort, (f,g) \in \FF(X) \rangle$.
 \item $\matrixsort \mapsto \langle \matrixsort, \matrixsort, \indexsort, (g,A,d) \in \Mat_{\FF(X)} \rangle$.
\end{enumerate}

As for $\mathcal{I}_{2}$, we take the following map: let $x$ be a variable.
\begin{enumerate}
\item For $x$ of $\indexsort$, assign itself.
\item For $x$ of $\fieldsort$, assign a pair of fresh variables of $\matrixsort$. We denote it by $(h_{x},g_{x})$.
\item For $x$ of $\matrixsort$, assign a tuple $(g,x,d)$, where $g$ is a 
fresh variable of $\matrixsort$, and $d$ is a fresh variable of $\indexsort$.
 We denote them by $g_{x}$ and $d_{x}$ respectively.
\end{enumerate}

\subsection{$\mathcal{I}_{3}$ on $\mathcal{L}_{0}$}
To define $\mathcal{I}_{3}$ on $\mathcal{L}_{\LAP}$, we follow induction on the construction of $\mathcal{L}_{\LAP}$.
As in \S \ref{A formalization of polynomials in LAP} for the function symbols $f(x_{1},\ldots,x_{l})$, instead of $\gamma_{f}$ in Definition \ref{interpretation}, we define an $\mathcal{L}_{\LAP}$-term $f_{rat}(\mathcal{I}(x_{1}), \ldots, \mathcal{I}(x_{l}))$.
Then $\gamma_{f}$ is naturally defined by $\mathcal{I}(z)=f_{rat}(\mathcal{I}(x_{1}), \ldots, \mathcal{I}(x_{l}))$.
In this subsection, we deal with $\mathcal{L}_{0}$.
As for $\mathcal{L}_{\indexsort}$, $\mathcal{I}_{3}$ is exactly the same ``identity mapping'' as in \S\S \ref{I3onL0forpoly}.
 
Next, we consider the symbols related to $\fieldsort$.
It amounts to implementing basic arithmetical operations on rational functions (below, we omit $*_{pol}$ if it is clear from the context):
 \begin{enumerate}
 \item  $=_{\fieldsort} \mapsto (f_{1},g_{1}) =_{rat}(f_{2},g_{2}) : \leftrightarrow f_{1}g_{2} =_{pol} f_{2}g_{1}$.
 \item $(0_{\fieldsort})_{rat}:= (0_{pol}, 1_{pol})$, $(1_{\fieldsort})_{rat}:= (1_{pol}, 1_{pol})$.

 \item $(f_{1},g_{1})(+_{\fieldsort})_{rat} (f_{2},g_{2}) := (f_{1}g_{2}+_{pol}f_{2}g_{1},g_{1}g_{2})$.

 We denote the RHS of $=$ by just $(f_{1},g_{1})+_{rat} (f_{2},g_{2}) $.
 \item $(f_{1},g_{1})(*_{\fieldsort})_{rat} (f_{2},g_{2})  := (f_{1}f_{2} , g_{1}  g_{2})$.
 
 We denote the RHS of $=$ by just $(f_{1},g_{1})*_{rat} (f_{2},g_{2})$.

 \item $(-_{\fieldsort})_{rat}(f,g) :=((-_{\fieldsort})_{pol} f,g)$.

 \item 
 \begin{align*}
 (f_{1},g_{1})^{-1}_{rat} := 
 \begin{cases}
 (g_{1},f_{1}) \quad &(f_{1} \neq_{pol} 0_{pol})\\
 (0,1) &(f_{1} =_{pol} 0_{pol}),
 \end{cases}
  \end{align*}
 that is,
 \[(f_{1},g_{1})^{-1}_{rat} := (\delta_{\lnot A=_{pol}0}(A \mapsto f_{1}) g_{1}, \delta_{\lnot A=_{pol}0}(A \mapsto f_{1})f_{1}+_{pol} \delta_{A=_{pol}0}(A \mapsto f_{1}) 1_{pol} ),\]
 where $\delta_{\cdot}$'s are the terms given in Lemma \ref{characteristicfunction}.
 \end{enumerate}

Lastly, we deal with the remaining symbols in $\mathcal{L}_{0}$ (of $\LAP$):

\begin{enumerate}
 \item  $\row_{rat}(g, A,d) := \row_{pol}(A,d)$.
 \item $\column_{rat}(g,A,d):= \column_{pol}(A,d)$. 
 \item  $\extract_{rat}(g,A,d,i,j) := (\extract_{pol}(A,d,i,j),g)$.
 
  ($\extract_{rat}(g,A,d,i,j)$ returns the polynomial $f_{ij}/g$, where $A=(f_{ij}/g)_{ij}$).
 
 \item $\summation_{rat}(g,A,d) := (\summation_{pol}(A,d),g)$.
 
 \item $=_{\matrixsort} \mapsto$
   \begin{align*}
   &\row_{rat}(g,A,d) = \row_{rat}(h,B,d^{\prime}) \land \column_{rat}(g,A,d) = \column_{rat}(h,B,d^{\prime}) \land \\
   &\forall i \in [\row_{rat}(g,A,d)].\forall j \in [\column_{rat}(g,A,d)].\ \extract_{rat}(g,A,d,i,j) =\extract_{rat}(h,B,d',i,j).
   \end{align*}
We denote the RHS by $(g,A,d)=_{rat}(h,B,d^{\prime})$.
It clashes with $=_{rat}$ for $=_{\fieldsort}$, but there is no danger of confusion, so we stick to this readable notation.

\item $\mathtt{P}_{rat}(k,g,A,d) := (g^{k}, \mathtt{P}_{pol}(k,A,d))$ \\
  Here, $g^{k}:=\llbracket a^l \rrbracket_{pol} [h_a, l \mapsto g,k ]$. (For the definition of $h_{a}$, see the definition of $\mathcal{I}_{2}$ for $\llbracket \cdot \rrbracket_{pol}$ given in \S\S \ref{I1andI2onL0}.)
  Note that the powering function 
  \[\fieldsort \times \indexsort \rightarrow \fieldsort; (a,l) \mapsto a^l\]
  is given by the $\mathcal{L}_{\LAP_{-}}$-term $\extract (\mathtt{P}(l,\lamt_{ij}\langle 1,1, a \rangle),1,1)$.
  Recall $(\MM, \FF[X], \Mat_{\FF[X]}) \models \LAP_{-}$.

\end{enumerate}

  \subsection{$\mathcal{I}_{3}$ on $\mathcal{L}_{k}$}
   Let $k \geq 1$, and assume $\mathcal{I}_{3}$ on $\mathcal{L}_{k-1}$ and $\bold{b}[t]$ for $\mathcal{L}_{k-1}$-terms are already defined.
   Furthermore, assume that a tuple $f_{rat}(\mathcal{I}(x_{1}), \ldots, \mathcal{I}(x_{l}))$ of $\mathcal{L}_{\LAP}$-terms is defined for each $\mathcal{L}_{k-1}$-function symbol $f$, and $\mathcal{I}(f)$ is defined by $\mathcal{I}(z)=f_{rat}(\mathcal{I}(x_{1}), \ldots, \mathcal{I}(x_{l}))$.
   Note that $\tau \mapsto \tau_{rat}$ naturally extends to the case when $\tau$ is an $\mathcal{L}_{k-1}$-term.
      
   Define $\mathcal{I}_{3}$ on $\mathcal{L}_{k} \setminus \mathcal{L}_{k-1}$ as follows:
   \begin{enumerate}
 \item $\cond_{\indexsort}^{\alpha,\vec{x}}(\vec{x},x',x'') \mapsto \cond_{\indexsort}^{\alpha,\vec{x}}(\vec{x},x',x'')$.
 \item $\cond_{\fieldsort}^{\alpha,\vec{x}}(\vec{x},x',x'') \mapsto $\\
 $(\cond_{\fieldsort}^{\alpha, \vec{x}}(\vec{x},1,0) h_{x'} +\cond_{\fieldsort}^{\alpha,\vec{x}}(\vec{x}, 0,1) h_{x''}, \cond_{\fieldsort}^{\alpha, \vec{x}}(\vec{x},1,0) g_{x'} +\cond_{\fieldsort}^{\alpha,\vec{x}}(\vec{x}, 0,1) g_{x''})$.
  \item $\lamt_{ij} \langle m,n, u \rangle (x_{1}, \ldots, x_{l})\mapsto$
   \begin{align*}
    \Bigg(&D_{u}(\llbracket m \rrbracket_{rat},\llbracket n \rrbracket_{rat}),\\ 
    &\lamt_{kl}\Big\langle \llbracket m \rrbracket_{rat}(b[u]*\llbracket m \rrbracket_{rat}*\llbracket n \rrbracket_{rat}+1),n,\\
    &\quad \quad \coeff \Big(\num(u_{rat})[i,j \mapsto \rem(k,\llbracket m \rrbracket_{rat}),l] *_{pol} p_u(\llbracket m \rrbracket_{rat},\llbracket n \rrbracket_{rat},\rem(k,\llbracket m \rrbracket_{rat}),l),\\ 
    &\quad \quad \quad \quad \divi(k,\llbracket m \rrbracket_{rat})\Big) \Big\rangle,
    b[u]*\llbracket m \rrbracket_{rat}*\llbracket n \rrbracket_{rat}\Bigg).
    \end{align*} 
    
Here,
  \[D_{u}(m,n) := \prod_{i \in [m],j \in [n]} \den(\llbracket u \rrbracket_{rat}(i,j)) \in \FF[X],\]
  and, for each $(i_{0},j_{0}) \in [m] \times [n]$, we set
  \[p_{u}(m,n,i_{0},j_{0}) := \prod_{(i,j) \in [m] \times [n] \setminus \{(i_{0},j_{0})\}} \den(\llbracket u \rrbracket_{rat}(i,j)).\]
  Informally, if each $u(i,j)$ is a rational $f_{ij}/g_{ij}$, then $\llbracket \lamt_{ij} \langle m,n,u\rangle \rrbracket_{rat}$ returns the matrix $(f_{ij}p_{u}(m,n,i,j) / D_{u}(m,n))_{i \in [m],j \in [n]}$.
\end{enumerate}

\begin{prop}[$\LAP$]
 $\llbracket \cdot \rrbracket_{rat}$ is an interpretation of $\LAP$, that is, for each axiom $\varphi \in \LAP$,
 \[ \LAP \vdash \llbracket \varphi \rrbracket_{rat}.\]
\end{prop}

\begin{proof}
The proof is similar to that of Proposition \ref{[]pol is an interpretation}, and we omit it.
\end{proof}

Similarly as $\llbracket \cdot \rrbracket_{pol}$, we have the following interpretations of $\mathcal{L}_{\LAP}$-terms:
\begin{itemize}
 \item The characteristic polynomial $(\llbracket \ch \rrbracket_{rat}(g,A,d),m) \in \Mat_{\FF[X]}(m+1,1)$ for each $(g,A,d) \in \Mat_{\FF(X)}(m,m)$.
 \item The determinant $\llbracket \det \rrbracket_{rat} (g,A,d) \in \FF(X)$ for each $(g,A,d) \in \Mat_{\FF(X)}(m,m)$.
\end{itemize}

\begin{rmk}\label{multivariaterationals}
 Thus $\LAP$ can treat $\FF(X)$ and $\Mat_{\FF(X)}$ properly and prove that they satisfy $\LAP$ itself: setting $\FF(X_1):=\FF(X)$, then
 \[\langle \MM,\FF(X_1), \Mat_{\FF(X_1)} \rangle \models \LAP.\]
 
 Therefore, applying the interpretation $\llbracket \cdot \rrbracket_{rat}$ again, we obtain 
 \[(\langle \MM, \FF(X_1)(X), \Mat_{\FF(X_1)(X)} \rangle \models \LAP.\]
 Putting $\FF(X_1,\ldots,X_{k+1}):=\FF(X_1,\ldots, X_k) (X)$ ($k$ is a standard natural number) inductively and repeating the argument, we have \[\langle \MM, \FF(X_1,\ldots, X_k), \Mat_{\FF(X_1,\ldots, X_k)} \rangle \models \LAP.\]
\end{rmk}

The polynomials are embedded into the rationals straightforwardly as follows:

\begin{lemma}[$\LAP$]\label{polisrat}
 \begin{enumerate}
  \item\label{substfield} $\FF \rightarrow \FF[X];\ a \mapsto aX^{0}$ is an embedding of a ring.
  \item\label{subring} $\FF[X] \rightarrow \FF(X);\ f \mapsto (f,1)$ is an embedding of a ring.
  \end{enumerate}
\end{lemma}

\begin{defn}[$\LAP$]\label{multivariatepolynomials}
Inductively on $k \geq 1$, we define
\[\FF[X_1, \ldots, X_{k+1}]:=\{f \in \FF(X_1, \ldots, X_k)[X] \mid \forall i \leq \row (f). \coeff(f,i) \in \FF[X_1,\ldots, X_k] \}.\]
\end{defn}

\begin{lemma}[$\LAP$]\label{changeofvar}
  \begin{align*}
  \FF[X_1] &\rightarrow \FF[X_1,X_2];\\
   [a_{0}, \ldots, a_{d}]=a_{0}X_1^{0} + \cdots + a_{d} X_1^{d} 
  &\mapsto \left(
  \begin{bmatrix}
  a_{0}X_1^{0}\\
   \vdots \\
  a_{d}X_1^{0}
  \end{bmatrix} ,0\right)= a_{0}X_2^{0} + \cdots + a_{d} X_2^{d}
  \end{align*}
   is an embedding of a ring.
\end{lemma}

\begin{proof}
The verification is straightforward.
\end{proof}

\section{The theory $\LAPPD$}\label{The theory LAPPD}
In \cite{The proof complexity of linear algebra}, 
\[\LAP + \mbox{$\det$ is multiplicative}\]
 is considered, and it is shown that cofactor expansion and Cayley-Hamilton theorem are proved in the theory. 
Based on it, we consider the following theory ($\bold{PD}$ stands for ``polynomial determinant'');

\begin{defn}
\[ \LAPPD := \LAP + \llbracket \mbox{$\det$ is multiplicative} \rrbracket_{rat},\]
that is,
\begin{align*}
\LAPPD &:= \LAP +\\
 &\Large((g,A,d) ,(h,B,d^{\prime}) \in \Mat_{\FF(X)}(m,m) \rightarrow \\
 &\detrat((g,A,d)*_{rat}(h,B,d^{\prime})) =_{rat} \detrat(g,A,d) *_{rat}\detrat(h,B,d^{\prime}) \Large)
 \end{align*}
\end{defn}
Consider the following axiom:
\begin{defn}
We define the formula $(MDP)$ as:
\[ (A,d) ,(B,d^{\prime}) \in \Mat_{\FF[X]}(m,m) \rightarrow \detpol((A,d)*_{pol}(B,d^{\prime})) =_{pol} \detpol(A,d) *_{pol}\detpol(B,d^{\prime}).\]
\end{defn}

\begin{prop}
$\LAP + (MDP)$ is equivalent to $\LAPPD$.
\end{prop}

\begin{proof}
We first show $\LAPPD \vdash (MDP)$.
Since
\[\LAP \vdash \detrat(1,A,d) =_{rat} (\detpol(A,d),1)\]
follows immediately by open induction along the definition of determinant given in \cite{The proof complexity of linear algebra}, we have
\begin{align*}
(\detpol((A,d)*_{pol}(B,d^{\prime})),1) &=_{rat} \detrat(1,(A,d)*_{pol}(B,d^{\prime}))\\
&=_{rat} \detrat((1,A,d)*_{rat}(1,B,d^{\prime}))\\
&=_{rat} \detrat(1,A,d)*_{rat} \detrat(1,B,d^{\prime})\\
&=_{rat} (\detpol(A,d),1) *_{rat}(\detpol(B,d^{\prime}),1)\\
&=_{rat} (\detpol(A,d)*_{pol}\detpol(B,d^{\prime}),1),
\end{align*}
which implies $\detpol((A,d)*_{pol}(B,d^{\prime}))=\detpol(A,d)*_{pol}\detpol(B,d^{\prime})$.
We used the multiplicativity of $\detrat$ at the third equality.

Next we show $\LAP + (MDP) \vdash \LAPPD$.
We can show that 
\[\LAP \vdash \detrat(g,A,d) =_{rat} (\detpol(A,d) , g^{m}).\]
It is because
\[\LAP \vdash \detrat(1,A,d) =_{rat} (\detpol(A,d),1),\]
and 
\[\LAP \vdash \det(aA) = a^{m} \det(A).\]
Now, in $\LAP+(MDP)$,
\begin{align*}
 \detrat(g,A,d) \cdot_{rat}\detrat (h,B,l) 
 &=_{rat}(\detpol(A,d) , g^{m}) *_{rat} (\detpol(B,l) , h^{m}) \\
 &=_{rat}(\detpol(A,d) *_{pol} \detpol(B,l) , g^{m}h^{m}) \\
 &=_{rat}(\detpol((A,d) *_{pol}(B,l)), (g\cdot h)^{m}) \\
 &=_{rat}  \detrat(gh,(A,d) *_{pol}(B,l)) \\
 &=_{rat}  \detrat ((g,A,d) *_{rat} (h,B,l) ) .
\end{align*}

\end{proof}

We end this section by the fact that $\LAPPD$ shows that a natural definition of $\det(A)$ coincides with the result of Berkowitz's algorithm:
\begin{prop}[$\LAPPD$]\label{twodefsofdet}
For $A \in \Mat_{\FF}(n,n)$, 
\[\detpol(XI-A) =_{pol} (-1)^{n}\ch (A)\]
 as polynomials of $X$.
\end{prop}

\begin{proof}
In this proof, we often omit subscripts $()_{pol}$ for readability.
Work in $\LAPPD$.
We have $\langle \MM, \FF[X], \Mat_{\FF[X]} \rangle \models \LAP + \mbox{$\det$ is multiplicative}$.
Thus $\langle \MM, \FF[X], \Mat_{\FF[X]} \rangle$ satisfies Cayley-Hamilton theorem and cofactor expansion by the results of \cite{The proof complexity of linear algebra}.
Therefore, the proofs of Lemma 4.2.1 and Lemma 4.2.2 in \cite{Soltys} interpreted for matrices with polynomial coefficients can be formalized and proven in $\LAPPD$.

Now, we show $\detpol(XI-A) =_{pol} (-1)^{n}\ch (A)$ by induction on $n$.
 
 Let
 \begin{align*}
  A=\begin{bmatrix}
   a & R \\
   S & M
  \end{bmatrix},
 \end{align*}
 where $a \in \FF$, $M \in \Mat_{\FF}(n-1,n-1)$.
 
 Then 
  \begin{align*}
  XI-A=\begin{bmatrix}
   X-a & -R \\
   -S & XI-M
  \end{bmatrix},
 \end{align*}
 and, by Lemma 4.2.1 of \cite{Soltys}, we have
 \[\ch_{pol}(XI-A) = (YI-(X-a))\ch_{pol}(XI-M) - (-R) (\adj_{pol})_{pol}(YI-(XI-M)) (-S)\]
 in $\FF(X)[Y]$, where $\adj(N)$ is the adjoint of $N$.
 Substituting $0$ for $Y$, we obtain
 \begin{align*}
 \detpol(XI-A) =& -(X-a)\detpol(XI-M) - R\adj_{pol}(-(XI-M))S \\
 =& (-1)^{n}(X-a)\ch(M)-(-1)^{n-2}R\adj_{pol}(XI-M)S
 \end{align*}
 by Induction Hypothesis.
 By Lemma 4.2.1 of \cite{Soltys} again, we obtain 
 \[\detpol(XI-A)=(-1)^{n}\ch(A).\]
 
\end{proof}

\section{Matrix substitution to polynomials}\label{Matrix substitution to polynomials}
 As discussed in \cite{The proof complexity of linear algebra}, in $\LAP_{-}$, we can substitute a matrix $A$ for the indeterminate in $f(X) \in \FF[X]$ (substitution of the field element, i.e. $f(a)$ for $a \in \FF$ can be formalized analogously and is easier than the following, so we omit the details of it);
\begin{defn}[$\LAP_{-}$]
 For  $f(X) \in \FF[X]$ and $A \in \Mat_{\FF}(n,n)$,
 \[f(A) := \sum_{i=0}^{\row (f)-1} \coeff(f,i) A^{i}.\]
 More formally,
 \[f(A) := \sum (M_{f}) \quad (M_{f} := \lamt_{ij} \langle \row (f)n ,n, \extract (f,i,1) \extract (A^{\divi(i-1,n)}, \rem(i-1,n)+1,j ) \rangle).\]
\end{defn}

Fixing $A \in \Mat_{\FF}(n,n)$, we can show that $f(X) \mapsto f(A)$ is a ring homomorphism:

\begin{lemma}[$\LAP_{-}$]\label{substisinv}
 For  $f(X), g(X) \in \FF[X]$ and $A \in \Mat_{\FF}(n,n)$,
 \[f=_{pol}g \rightarrow f(A)=g(A).\]
\end{lemma}

\begin{proof}
  Suppose $f =_{pol}g$.
  By definition,
   \begin{align*}
  f(A) = \sum ( M_{f} ),\   g(A) = \sum ( M_{g} ).  
  \end{align*}
  It is easy to show
  \[\forall i,j.\ \extract (M_{f},i,j) = \extract (M_{g},i,j).\]
  using $f =_{pol}g$.
  
  Moreover, we can show 
  \[\forall i,j.\ \extract (A,i,j) = \extract (B,i,j) \rightarrow \sum (A) = \sum (B)\]
  by open induction.
  
  Hence, the result follows.
\end{proof}

\begin{lemma}[$\LAP_{-}$]
 For  $f(X), g(X) \in \FF[X]$ and $A \in \Mat_{\FF}(n,n)$,
 \begin{enumerate}
  \item\label{unit} $1(A) = I_{n}$.
  \item\label{addition} $f(A)+g(A) = (f+_{pol}g)(A)$.
  \item\label{multiplication} $f(A) * g(A) = (f *_{pol} g) (A)$.
 \end{enumerate}
\end{lemma}

\begin{proof}

(\ref{unit}). $1(A) = I_{n}$ is clear.
 
(\ref{addition}). We show $f(A) + g(A) = (f+g) (A)$.
 By definition,
 \begin{align*}
  f(A) &= \sum (M_{f}),\\
  g(A) &= \sum (M_{g}),\\
  (f+_{pol}g)(A) &= \sum (M_{f+_{pol}g}) .
 \end{align*}
 By Lemma \ref{substisinv}, we may assume that 
 \[\row (f)=\row (g)=\row (f+_{pol}g),\ f+_{pol}g = f+g \ \mbox{ as vectors}.\]
 From these equalities, we can derive that 
 \[M_{f+_{pol}g} = M_{f} + M_{g}.\]
 Since 
 \[\sum(A+B) = \sum(A)+\sum(B),\]
 the result follows.
 
 (\ref{multiplication}). Lastly, we show $f(A) * g(A) = (f *_{pol} g) (A)$.
 
 For $k \in \MM$, let $f_{k} := [a_{0}, \ldots, a_{k}]^{t}$, where $f=[a_{0}, \ldots, a_{d}]^{t}$. 
 More formally,
 \[f_{k} := \lamt_{ij} \langle k+1,1, \extract (f,i,1)\rangle.\]
 We show
 \[f_{k}(A) * g(A) = (f_{k} *_{pol} g) (A)\]
 by induction on $k$ (then the case $k=d$ gives the result).
 
 When $k=0$, $f_{k}$ is a scalar and therefore 
 \begin{align*}
  f_{0} (A) * g(A) = f_{0} g(A)
 \end{align*} 
 and 
 \begin{align*}
  (f_{0} *_{pol} g) (A) = (f_{0}g)(A) = \sum (M_{f_{0}g})=\sum (f_{0}M_{g}) = f_{0} \sum(M_{g}) = f_{0}g(A).
 \end{align*}
 In the middle equation, we used $M_{f_{0}g} = f_{0}M_{g}$, which follows easily from the definition.

 This finishes the case when $k=0$.
 
 Now we show 
  \[f_{k+1}(A) g(A) = (f_{k+1} *_{pol} g) (A).\]
 We can write 
 \[f_{k+1} (X) = f_{k}(X) +_{pol} aX^{k+1}.\]
 Therefore it suffices to show
 \begin{align}\label{commute}
 ( f_{k}(X) +_{pol} aX^{k+1})(A) g(A) = (( f_{k}(X) +_{pol} aX^{k+1}) *_{pol} g) (A).
 \end{align}
 The LHS of (\ref{commute}) is 
 \begin{align*}
 ( f_{k}(X) +_{pol} aX^{k+1})(A)g(A) 
 &= ( f_{k}(A) + aA^{k+1})g(A)\\
 &=f_{k}(A) g(A) + aA^{k+1} g(A).
 \end{align*}
 
 On the other hand, the RHS of (\ref{commute}) is
 \begin{align*}
 (( f_{k}(X) +_{pol} aX^{k+1}) *_{pol} g) (A)
 &=( f_{k}(X)*_{pol}g(X) +_{pol} aX^{k+1}*_{pol}g(X) ) (A) \\
 &= f_{k}(A) g(A) + (aX^{k+1}*_{pol}g(X)) (A).
 \end{align*}
 In the last equality, we have used (\ref{addition}) and the induction hypothesis.
 
 Therefore, it suffices to show
 \[(X^{k+1}*_{pol}g(X))(A) = A^{k+1}* g(A).\]

 Since we know that $A^{n}A^{m} = A^{n+m}$,
 it is enough to show that
 \[\sum_{j=0}^{k} CB_{j} = C \sum_{j=0}^{k}B_{j}.\]
Note that we can show
 \[\sum_{j=0}^{k+1}B_{j} =( \sum_{j=0}^{k} B_{j} ) + B_{k+1},\]
 and this formula implies the equation above by open induction on $k$.

\end{proof}

\section{A definition of $\rank$ and its basic properties}\label{A definition of rank and its basic properties}
Now,  we show that $\LAP$ can $\Sigma^{B}_{0}$-define the rank function.

First, we observe that $\LAP$ can count the number of solutions of open predicate in a vector (however, outputting the appropriate element of $\indexsort$ is impossible because of its vocabulary, so, instead, we design a term outputting a vector whose only non-zero component corresponds to the number):
 
\begin{defn}\label{indexvector}
An \textit{index vector} is a vector $v \in \Mat_{\FF}(n,1)$ such that
\[ \exists i \leq n (v_{i1} = 1 \land \forall j \leq n (j \neq i \rightarrow v_{j1} = 0)).\]
An \textit{index matrix} is a matrix $A= [v_{1}, \ldots, v_{s}]\in \Mat_{\FF}(n,s)$ such that
every $v_{i}$ is an index vector.

We identify an index vector $v$ with the unique index $i \leq n$ such that $v_{i1} = 1$.
Note that such an index is $\Sigma^{B}_{0}$-defined as
\[ \min\{i \leq n \mid v_{i1} =1\}.\] 

\end{defn}
\begin{prop}[$\LAPPD$]
There exists a term $\ct(v,k)$ such that for any $v \in \Mat_{\FF}(n,1)$ and $k \in \MM$, the following hold:
\begin{enumerate}
 \item\label{first} $\ct(v,k) \in \Mat_{\FF}(n+1,1)$, and it is an index vector.
 \item $\ct(v,0) = [1,0,\ldots,0]^{t} \in  \Mat_{\FF}(n+1,1)$.
 \item $\extract (\ct(v,k),i,1)=1 \land \extract (v,k+1) \neq 0 \rightarrow \extract (\ct(v,k+1),i+1,1)=1$.
 \item\label{last} $\extract (\ct(v,k),i,1)=1 \land \extract (v,k+1) = 0 \rightarrow \extract (\ct(v,k+1),i,1)=1$.

\end{enumerate}

\end{prop}
Intuitively, $\ct(v,k)$ counts the number of nonzero components within the first $k$ coordinates of $v$.

\begin{proof}
 We simulate the recursion by matrix powering.
 For $n \in \MM$, let
 \begin{align*}
S_{n+1} := 
 \begin{bmatrix}
\zeromat_{1,n} & \zeromat_{1,1} \\
I_{n} & \zeromat_{n,1}
\end{bmatrix}
\in \Mat_{\FF}(n+1,n+1).
\end{align*}
Given a vector $v \in \Mat_{\FF}(n,1)$, we construct a sequence $Y_{0}, \ldots, Y_{n} \in \Mat_{\FF}(n+1,1)$ such that 
\begin{align}
 Y_{0} &= [1, 0, \cdots, 0]^{t} \\
 Y_{k+1} &= 
 \begin{cases}
  Y_{k} \quad &\mbox{if $\extract (v,k+1,1)=0$} \\
  S_{n+1}Y_{k} \quad &\mbox{if $\extract (v,k+1,1)\neq 0$}.
 \end{cases}\label{recursion}
\end{align}

Firstly, for $A \in \Mat_{\FF}$, let
 \[\chmat (A) := \lamt_{ij} \langle \row (A),\column (A), \extract (A,i,j) *_{field}\extract (A,i,j)^{-1}  \rangle .\]
 
 Then it satisfies that for all $i \leq \row (A)$ and $j \leq \column (A)$, 
 \begin{align*}
  &\extract (A,i,j) = 0 \rightarrow \extract (\chmat(A),i,j) = 0,\\
  &\extract (A,i,j) \neq 0 \rightarrow \extract (\chmat(A),i,j) = 1.
 \end{align*}

Now, let
\begin{align*}
 T_{k}:= (1-\extract (\chmat(v),k,1))I_{n+1} + \extract (\chmat(v),k,1) S_{n+1},
 \end{align*}
 and set
 \begin{align*}
 Y_{0} &:= [1, 0, \cdots, 0]^{t} \\
 Y_{k} &:= T_{k} \cdots T_{1} Y_{0} \quad (k \in [n]).
 \end{align*}
 Finally, define 
 \[\ct(X,k)=Y_{k}.\]
 
 Since $\ct(v,k)$ is defined by iterated matrix product of uniform $\mathcal{L}_{\LAP}$-terms, it is a term in the language of $\mathcal{L}_{\LAP}$.
 
 Moreover, by definition, $\ct(X,k)=Y_{k}$ satisfies the condition \ref{recursion}.
 
 The conditions (\ref{first})-(\ref{last}) in the proposition are proved by $\Sigma^{B}_{0}$-induction.
 
\end{proof}

From now on, we work in $\LAPPD$ and $\Sigma^{B}_{0}$-define a function $\rank (A)$ for $A \in \Mat_{\FF}(m,n)$ and show that it satisfies several basic properties. 

First, note that, by the result of \cite{The proof complexity of linear algebra} and the interpretation above,
\begin{lemma}[$\LAPPD$,\cite{The proof complexity of linear algebra}]
 Cofactor expansion and Cayley-Hamilton theorem for 
 \[A \in \Mat_{\FF[X]}(m,m), \Mat_{\FF(X)}(m,m)\]
  hold.
\end{lemma}

We formalize Mulmuley's algorithm calculating the rank of a given matrix (see \cite{Mulmuley's algorithm} for a simple and clear exposition).

\begin{defn}[$\LAPPD$]\label{formalizedMulmuley}
 For $A \in \Mat_{\FF}(m,n)$, we define the following terms and $\Sigma^{B}_{0}$-definable (or, equivalently, open-definable) functions:
 \begin{enumerate}
  \item \begin{align*}
  &\symm(A):=(1,\lamt_{ij} \langle \row (A)+\column (A), \row (A)+\column (A),\\
  &\cond_{field}(i\leq \row (A) < j,1,0) \extract (A,i,j-\row (A)) \\
  &+\cond_{field}(j \leq \row (A) < i,1,0) \extract (A,j,i-\row (A)) \rangle,0) \in \Mat_{\FF(X)},
  \end{align*}
 which is a term.
   ($\symm(A)$ returns the matrix
  \[ \begin{bmatrix}
\zeromat & A \\
A^{t} & \zeromat \\
\end{bmatrix} \in \Mat_{\FF(X)}.\]
  )
  \item Let the term $\chi$ be
  \[\chi(n) := \llbracket \lamt_{ij}\langle k,l,\cond_{field}^{i=j,(i,j)}(i,j,1,0)a^{i-1} \rangle \rrbracket_{rat} [k,l,h_{a},g_{a} \mapsto n,n,X,1] \in \Mat_{\FF(X)}(n,n)\]
  (here, note that $X$ substituted for $h_{a}$ is a code of the polynomial ``$X$.''
  $\chi(n)$ returns the matrix
  \[ \begin{bmatrix}
X^{0} & 0 &\cdots & 0 \\
0 & X^{1} & \cdots & 0\\
\quad  & \quad & \ddots & \quad \\
0 & \cdots &0 & X^{n-1} 
\end{bmatrix} \in \Mat_{\FF(X)}.\])
We denote $\chi_{n}$ instead of $\chi(n)$ for readability.
  Using this, put
  \[\polize (A):= \chi_{\row (A)+\column (A)} \symm(A),\]
  which is another term.
  \item For $v \in \Mat_{\FF(X)}(n,1) \subset \FF(X)[Y]$, let $i_{v}$ be the minimum $i$ such that $\extract_{rat} (v, i,1) \neq 0$. (Note that $v$ is formally a triplet $(g,A,d)$.)
  Let 
  \[\mul(v):=i_{v}-1,\]
  which is a $\Sigma^{B}_{0}$-definable function.
   ($\mul(v)$ returns the multiplicity of the root $0$ of the polynomial $v(Y)$.)
   Note that $\mul(v) \leq \row_{rat}(v)$.
   \item\label{explicitdefofrank} $\rank (A) := \divi(\row (A)+\column (A) - \mul(\llbracket \ch \rrbracket_{rat} (\polize (A))) , 2)$,
  which is again a $\Sigma^{B}_{0}$-definable function since $\divi(\row (A)+\column (A) - \mul(v), 2)$ is $\Sigma^{B}_{0}$-definable and $\llbracket \ch \rrbracket_{rat} (\polize (A))$ is a term.
 \end{enumerate}
 
 Furthermore, we also introduce the following for convenience in the analysis below:
 for $v \in \Mat_{\FF(X)}(n,1) \subset \FF(X)[Y]$, 
  \[\widetilde{v}:=\llbracket \lamt_{ij}\langle \row (M) - k,1, \extract (M,i+k,1) \rangle \rrbracket_{rat}[(g_{M},M,d_{M}),k \mapsto v,\mul(v)],\]
  which is another $\Sigma^{B}_{0}$-definable function.
  (The polynomial $v(Y)$ is decomposed into $Y^{\mul(v)} \widetilde{v}(Y)$.)
 
\end{defn}

Now, we verify that $\LAPPD$ can prove that the function $\rank(A)$ above satisfies the basic properties of rank function.

First, we recall the following fact:

\begin{lemma}[$\LAP_{-} + \det(AB)=\det(A) \det(B)$ ]\label{CHimpliesnontrivialsolution}
For any $A \in \Mat_{\FF}(n,n)$, there exists $B \in \Mat_{\FF}(n,n) \setminus \{0\}$ such that $AB=\det(A) I_{n}$.
\end{lemma}

\begin{proof}
Let $f(Y):=\ch(A) \in \Mat_{\FF}$.
 By \cite{Soltys}, Cayley-Hamilton theorem, $f(A)=0$, is available.
 Let $m$ be the multiplicity of the root $0$ of $f$, and let $\widetilde{f}$ be the factor of $f$ satisfying $f(Y)=Y^{m}\widetilde{f}(Y)$.
  Then we have
 \[A^{m} \widetilde{f}(A)=0.\]
 
 If $m= 0$, $\det(A) = f(0) \neq 0$. 
 Hence, take $B:=\adj A$. 
 Then $AB=\det(A) I \neq 0$, and therefore $B \neq 0$. 
 
 Consider the case $m \neq 0$, that is, $\det(A)=0$.
 We first show that $\widetilde{f}(A) \neq 0$.
 Suppose otherwise.
 Let
 \[\widetilde{f}(Y) = a_{0} + a_{1}Y + \cdots + a_{k} Y^{k}.\]
 Then $\widetilde{f}(A) =0$ implies  
 \[a_{0}I_{n} =  -a_{1}A - \cdots - a_{k} A^{k} = -A (a_{1} I_{n} + \cdots + a_{k} A^{k}) \]
 Therefore, taking the determinant of both sides, 
 \[0 \neq a_{0}^{n} = (-1)^{n}\det(A) \det (a_{1} I_{n} + \cdots + a_{k} A^{k}) = 0.\]
 This is a contradiction.
 
 Since $\widetilde{f}(A) \neq 0$ and $A^{m} \widetilde{f}(A)=0$, by open induction, we can take the smallest $i+1$ such that
 \[A^{i+1}  \widetilde{f}(A)=0 .\]
 Since 
 \[A^{i}  \widetilde{f}(A) \neq 0,\]
 $B := A^{i} \widetilde{f}(A)$ suffices.
\end{proof}

\begin{cor}[$\LAPPD$]\label{quasiinverse}
For any $A \in \Mat_{\FF[X]}(n,n)$, there exists $B \in \Mat_{\FF[X]}(n,n) \setminus \{0_{pol}\}$ such that $A*_{pol}B=\det_{pol}(A) *_{pol} (I_{n},0)$.

This also holds if we replace $\FF[X]$ by $\FF(X)$, $*_{pol}$ by $*_{rat}$, and $\det_{pol}$ by $\det_{rat}$, and $(I_{n},0)$ by $(1,I_{n},0)$.
\end{cor}

\begin{cor}[$\LAPPD$]\label{dimension}
 For any $A \in \Mat_{\FF[X]}(n,n+1)$, there exists $b \in \FF[X]^{n+1}\setminus \{ 0 \}$ such that $A*_{pol}b=0$.
 
 Regarding $A$ as a set of $(n+1)$-column vectors of dimension $n$, it means that they should be linearly dependent.
 
 This also holds if we replace $\FF[X]$ by $\FF(X)$ and $*_{pol}$ by $*_{rat}$.
\end{cor}

\begin{proof}
In this proof, we omit $*_{pol}$ for readability.
 Let $C$ be the matrix obtained from $A$ by adding a row vector $0$ as the first row.
 We see that 
 \[C \in \Mat_{\FF[X]}(n+1,n+1).\]
 Take $B \neq 0$ such that $CB=\det (C) I$ or $CB = 0$.
 Since $\det(C)$ = 0 because $C$ has $0$ as the first row, we see that $CB=0$.
 $B \neq 0$ implies existence of a nontrivial column vector $b$ of $B$, which satisfies $Cb = 0$.
\end{proof}

\begin{lemma}[$\LAPPD$]\label{zerointersection}
For $A \in \Mat_{\FF}(m,n)$,  we have
  \[\ker \polize (A) \cap \im \polize (A) = 0\]
   in $\FF(X)^{m+n}$, that is, 
  for all $v(X) \in \FF(X)^{m+n}$,
  \[ \polize (A) *_{rat} (\polize (A) *_{rat} v ) =_{rat} 0 \rightarrow \polize (A) *_{rat} v =0.\]
\end{lemma}

\begin{proof}
 In this proof, we omit $*_{rat}$ for readability.
 Furthermore, since we are going to use two indeterminates, we follow the notation of Remark \ref{multivariaterationals} and Definition \ref{multivariatepolynomials}.
 We identify the given indeterminate $X$ with $X_{1}$, and we denote the second indeterminate $X_{2}$ with $Y$.
 
 First, by clearing the denominators (which amounts to computing the bounded product of the denominators) of each component of $v$, we may assume that $v \in \FF[X]^{m+n}$. (Formally speaking, $v$ is in the image of the natural embedding of $\FF[X]^{m+n}$ into $\FF(X)^{m+n}$ induced by Lemma \ref{polisrat}.)
 By definition, 
 \[\polize (A) = \chi_{m+n} S,\]
  where $S:=\symm(A) \in \Mat_{\FF(X)}(m+n,m+n)$.
 Assume
 \[\chi_{m+n}  S \chi_{m+n}  S v = 0.\]
  Hence,
  \[S \chi_{m+n}  S v = 0.\]
  Therefore, embedding $\FF(X)$ into $\FF(X)[Y]$ naturally, we have
  \[v^{t}(Y) S \chi_{m+n} S v = 0\]
  in $\Mat_{\FF(X)[Y]}(m+n,m+n)$.
  Here, $v(Y)$ is the vector in $\Mat_{\FF(X)[Y]}(m+n,1)$ obtained by substituting $Y$ for $X$ in each component of $v$ (Lemma \ref{changeofvar}).
  It means 
  \[ u^{t}(Y) \chi_{m+n} u = 0 \quad (u := Sv) .\]
  (Note that each component of $S$ is in $\FF$, or the image of the natural embedding of $\FF$ into $\FF(X)[Y]$.)
  Therefore,
  \begin{align}\label{criticaleq}
  g:=\sum_{i=1}^{m+n} u_{i}(Y) X^{i-1} u_{i}(X)=0,
  \end{align}
  where each $u_{i}(X)$ is the $i$-th component of $v$, which is a polynomial in $\FF(X)$, and $u_{i}(Y)$ is obtained by $X \mapsto Y$ formalized in Lemma \ref{changeofvar}.
  
  Consider the largest $i_{1}$ such that $\deg u_{i}(X)$ evaluated in $\FF[X]$ takes its maximum.
  We can see that it coincides with $\deg u_{i_{1}}(Y)$ evaluated in $\FF(X)[Y]$.
  Let $d$ be the maximum.
  
  We consider ``the coefficient of $Y^{d}X^{i_{1}-1+d}$'' (formally speaking,\\
   $\coeff(\llbracket \coeff \rrbracket_{pol}(g,d),i_{1}-1+d)$) of the LHS of (\ref{criticaleq}).
  More formally, we consider the coefficient of $X^{i_{1}-1+d}$ of the coefficient of $Y^{d}$.
  For each $i = 1, \ldots, m+n$, if $\deg u_{i}(X) \neq d$, then the coefficient $f(X)$ of $Y^{d}$ in $u_{i}(Y) X^{i-1} u_{i}(X)$ is $0$ by the maximality of $d$, and therefore so is the coefficient of $X^{i_{1}-1+d}$ in $f(X)$.
  If $\deg u_{i}(X) = d$, then the coefficient of $Y^{d}$ in $u_{i}(Y) X^{i-1} u_{i}(X)$ is $C_{i}X^{i-1} u_{i}(X)$, where $C$ is a nonzero constant.
  Especially if $i=i_{1}$, it is $C_{i_{1}}X^{i_{1}-1} u_{i_{1}}(X)$, and its coefficient of $X^{i_{1}-1+d}$ is $C_{i_{1}}^{2}$.
  Otherwise, by the maximality of $i_{1}$, the coefficient of $X^{i-1+d}$ in $C_{i}X^{i-1} u_{i}(X)$ is $0$.
  
Therefore, extracting the coefficient of $X^{i_{1}-1+d}$ of the coefficient of $Y^{d}$ in the both sides of (\ref{criticaleq}), we have $C_{i_{1}}^{2}=0$, a contradiction.
\end{proof}

\begin{lemma}[$\LAPPD$]\label{decomposition}
 For any $A \in \Mat_{\FF(X)}(m,n)$ and $v \in \FF(X)^{m+n}$, let $C := \polize (A)$.
 Then there uniquely exists 
 \[(u_{1},u_{2}) \in \ker C \times \im C \subset (\FF(X)^{m+n})^{2}\]
  such that
 \[ v = u_{1} + u_{2}.\]
\end{lemma}

\begin{proof}
 The uniqueness quickly follows from Lemma \ref{zerointersection}.

 Now we show the existence.
 For readability, we denote $\ch(M)$ ($M \in \Mat_{\FF(X)}(m+n,m+n)$) by $p_{M}$.
 It lies in $\Mat_{\FF(X)[Y]}$.
 When we want to indicate the indeterminate, we write $p_{M}(Y)$.

 We show that 
 \[ u_{1} = \widetilde{p_{C}}(0) ^{-1}(\widetilde{p_{C}} (C) v), \ u_{2}:= \widetilde{p_{C}}(0)^{-1}( \widetilde{p_{C}}(0) - \widetilde{p_{C}} (C)) v\]
 suffice.
 
 It is clear that 
 \[v=u_{1} + u_{2} .\]
 
 We show $u_{1} \in \ker C$ next.
 First, by Cayley-Hamilton theorem,
 \[C^{\mul(p_{C})} \widetilde{p_{C}} (C) v = 0.\]
 Therefore, by Lemma \ref{zerointersection} and open induction,
 \[C^{\mul(p_{C}) -i} \widetilde{p_{C}} (C) v = 0\]
 holds for $i \in [0, \mul(p_{C})-1]$.
 
 Hence, 
 \[Cu_{1} = 0.\]
 
 Lastly, we show $u_{2} \in \im C$.
 Note that 
 \[\mul (\widetilde{p_{C}}(0) - \widetilde{p_{C}} (T)) \geq 1.\]
 Hence, we can write
 \[ \widetilde{p_{C}}(0) - \widetilde{p_{C}} (C) = C R.\]
 Therefore,
 \[u_{2} = CRv \in \im C.\]
 
\end{proof}

\begin{cor}[$\LAPPD$]\label{uniquedecomp}
 For any $A \in \Mat_{\FF[X]}(m,n)$ and $v \in \FF[X]^{m+n}$, let $C := \polize (A)$.
 Then there uniquely exists 
 \[(u_{1},u_{2}) \in \ker C \times \im C \cap (\FF[X]^{m+n})^{2}\]
  such that
 \[  \widetilde{p_{C}}(0) v = u_{1} + u_{2}.\]
\end{cor}

\begin{proof}
By the above lemma, we can write
\[v = u_{1} + u_{2}.\]
Consider 
\[ \widetilde{p_{C}}(0) v =  \widetilde{p_{C}}(0) u_{1}+ \widetilde{p_{C}}(0) u_{2},\]
then the right-hand side gives what we want by the definition of $u_{1},u_{2}$.
\end{proof}

\begin{lemma}[$\LAPPD$]\label{solvability}
There exists a $\Sigma^{B}_{0}$-definable relation $\solve (A,b)$ such that for any $A \in \Mat_{\FF} (m,n)$ and $b \in \FF^{m}$,
 \[\solve (A,b) \longleftrightarrow \exists x \in \FF^{n}.\ Ax = b.\]
\end{lemma}

\begin{proof}
Let 
\begin{align*}
\widetilde{b} := 
 \begin{bmatrix}
 b \\
 0
 \end{bmatrix}
  \in \FF^{m+n},
\end{align*}
and 
\[C := \polize (A).\]

Note that the following equivalences hold:
\begin{align*}
\exists x \in \FF^{n}.\ Ax=b
&\Longleftrightarrow \exists \widetilde{x} \in \FF^{m+n}.\ \symm (A) \widetilde{x} = \widetilde{b} \\
&\Longleftrightarrow \exists \widetilde{x} \in \FF[X]^{m+n}.\ \symm (A) \widetilde{x} = \widetilde{p_{C}}(0)\widetilde{b} \\
&\Longleftrightarrow \exists \widetilde{x} \in \FF[X]^{m+n}.\ \chi_{m+n} \symm (A) \widetilde{x} = \chi_{m+n} \widetilde{p_{C}}(0)\widetilde{b} \\
&\Longleftrightarrow \exists \widetilde{x} \in \FF[X]^{m+n}.\ C \widetilde{x} = \widetilde{p_{C}}(0)\chi_{m+n} \widetilde{b} \\
&\Longleftrightarrow \widetilde{p_{C}}(C) (\chi_{m+n}\widetilde{b})=0.
\end{align*}
 The second equivalence follows by substituting $0$ for the indeterminate $X$, and doing a division in $\FF$.
 The third equivalence follows by observing that $\chi_{m+n}$ is invertible in $\Mat_{\FF(X)}(m+n,m+n)$.
 The last equivalence follows by recalling that the $u_{1}$ in Corollary \ref{uniquedecomp} is given by $u_{1}:=\widetilde{p_{C}}(C) v$.
 
 Now, the last condition gives the desired $\Sigma^{B}_{0}$-predicate.
\end{proof}

\begin{lemma}[$\LAPPD$]
There exists a $\Sigma^{B}_{0}$-definable function $\sol(A,b)$ such that for any $A \in \Mat_{\FF} (m,n)$ and $b \in \FF^{m}$,
 \[\solve (A,b) \Longrightarrow \sol(A,b) \in \FF^{n} \& \ A \sol(A,b) = b.\]
\end{lemma}

\begin{proof}
 Let 
 \[\widetilde{p_{C}}(C) (\chi_{m+n}\widetilde{b})=0.\]
 Then 
 \begin{align*}
 \widetilde{p_{C}}(0)\chi_{m+n} \widetilde{b}
 &=(\widetilde{p_{C}}(0)-\widetilde{p_{C}}(C))\chi_{m+n} \widetilde{b} \\
 &= Cv,
 \end{align*}
 where $v$ is of the form $R\chi_{m+n} \widetilde{b}$.
 Now, going backwards the equivalences in the previous Lemma \ref{solvability},
 \begin{itemize}
  \item We see $\symm (A) v = \widetilde{p_{C}}(0) \widetilde{b}$.
  \item Substituting $0$ for the indeterminate $X$, and dividing the both sides by $\widetilde{p_{C}}(0) (0)$,
  we obtain 
  \[\symm (A) \left(\widetilde{p_{C}}(0)(0)\right)^{-1} v(0) =  \widetilde{b}.\]
  By definition of $\symm(A)$ and $\widetilde{b}$, extracting the lower $n$ components of $\left(\widetilde{p_{C}}(0)(0)\right)^{-1} v(0)$,
  we obtain the desired term $\sol(A,b)$.
 \end{itemize}
\end{proof}

\begin{lemma}[$\LAPPD$]\label{basisofimage}
There exist terms $\rk(A)$ and $\basis(A)$ such that for any $A \in \Mat_{\FF}(m,n)$,
\begin{enumerate}
\item $\rk(A) \in \Mat_{\FF}(m+1,1)$ is an index vector, and
\item each column of $\basis(A) \in \Mat_{\FF}(m,n)$ is either a column vector of $A$ or a zero vector. 
\end{enumerate}
Moreover, there are $\rk(A)$-many nonzero vectors in $\basis(A)$ which form a basis of $\im(A)$, that is,
 \begin{enumerate}
  \item For $v:= (\delta_{\exists i \leq \row (A).\ \basis(A)_{ij} \neq 0})_{j \in [n]}$,
  \[ \rk(A) = \ct(v, n)\]
  holds. (For the definition of $\delta$, see Lemma \ref{characteristicfunction}.)
  \item There exists $B \in \Mat_{\FF}(n,n)$ such that $\basis(A) B = A$.
  \item If $c \in \Mat_{\FF}(n,1)$ satisfies $\basis(A) c =0$, and $v_{j} =1$ for $v$ above, then $c_{j}=0$.
 \end{enumerate}
\end{lemma}

\begin{proof}
 For each $j \in [n]$, consider the condition
 \begin{align}\label{span?}
  \lnot \exists x \in \FF^{j-1}. [A_{1} , \ldots, A_{j-1}] x = A_{j} 
 \end{align}
 (Here, we set $[A_{1}, \ldots, A_{j-1}] = 1$ if $j=1$. In this case, the condition is equivalent to; $A_{1} \neq 0$).
 By Lemma \ref{solvability}, it is equivalent to the following $\Sigma^{B}_{0}$-formula:
 \[\varphi(A,j) : \equiv \lnot \solve ([A_{1} , \ldots, A_{j-1}], A_{j}).\]
 Let $v(A) \in \Mat_{\FF}(m,1)$ be the characteristic vector of the condition, that is,
 \[v(A) := \lamt_{ij} \langle n,1, \delta_{\varphi}(A,i) \rangle.\]
 (For the definition of $\delta_{\varphi}$, see Lemma \ref{characteristicfunction}.)
 
 Moreover, let $C$ be
 \[C := \lamt_{ij} \langle n,n, \cond^{i=j,(i,j)}_{field}(i,j,v(A)_{i},0) \rangle.\]
 Then
 \[\rk(A) := \ct(v(A),n), \basis(A) := A C\]
 suffice our purpose.

 Indeed, if $w \in \Mat_{\FF}(n,1)$ satisfies $\basis(A) w =0$, checking the largest index $i$ of $w$ such that $\extract (w,i,1) \neq 0$ by open induction,
 we see that it contradicts that the $i$-th column vector of $\basis(A)$ satisfies the condition \ref{span?}.
 
 Moreover, we can construct $B \in \Mat_{\FF}(n,n)$ such that $\basis(A) B = A$.
 Let 
 \begin{align*}
  b_{i} := \delta_{\solve (A,b)} \sol(A,b) + \delta_{\lnot\solve (A,b)} e_{i},
 \end{align*}
 (here, $e_{i}$ is a size-$n$ vector such that the $i$-th component is $1$ and others are $0$) and
 \[ B := [b_{1}, \ldots, b_{n}].\]
 Then it is clear that $\basis(A)B^{i}$ is of the form
 \[[A_{1}, \ldots, A_{i} , *].\]
 Therefore, $B^{m}$ is the desired matrix. 
\end{proof}

\begin{defn}
Let $A \in \Mat_{\FF}(m,n)$. 
For index matrices 
\begin{align*}
U:=[u_{1}, \ldots, u_{s}] \in \Mat_{\FF}(m,s), \ \mbox{and} \
V:=[v_{1}, \ldots, v_{s}] \in \Mat_{\FF}(n,s).
\end{align*}
 We define \textit{$(U,V)$-minor of $A$} as 
 \[A[U:V] := U^{t}AV.\]
 It immediately follows that 
\[ \extract (A[U:V], i,j) = A_{u_{i}v_{j}} \quad (i,j \in [s]).\]
Note that we are identifying index vectors and the corresponding numbers as mentioned in Definition \ref{indexvector}.
\end{defn}

\begin{lemma}[$\LAPPD$]\label{maximalnonsingular}
 Let $A \in \Mat_{\FF}(m,n)$, and $r=\rk(A)$, that is, the unique nonzero component of $\rk(A)$ is the $r$-th.
 Then there exist index matrices 
 \begin{align*}
 U=[u_{1}, \ldots, u_{r}] \in \Mat_{\FF}(m,r) \ \mbox{and} \
 V=[v_{1}, \ldots, v_{r}] \in \Mat_{\FF}(n,r)
 \end{align*}
 such that $A[U:V]$ is a maximal nonsingular minor of $A$:
 \begin{enumerate}
  \item $\det(A[U:V]) \neq 0$.
  \item For any $U^{\prime}=[U, u_{r+1}] \in \Mat_{\FF}(m,r+1)$ and $V^{\prime}=[V, v_{r+1}] \in \Mat_{\FF}(n,r+1)$, if $U^{\prime}$ and $V^{\prime}$ are again index matrices, then $\det(A[U^{\prime}:V^{\prime}])=0$.
 \end{enumerate}
\end{lemma}

\begin{proof}
 Let $A \in \Mat_{\FF}(m,n)$, $B := \basis (A)$, $C := \basis(A^{t})$, $r := \rk(A)$, $s:=\rk(A^{t})$.
 For a matrix $M$, we denote its $j$-th column vector by $M_{j}$. 
 We construct $U \in \Mat_{\FF}(m,r)$ and $V \in \Mat_{\FF}(n,s)$ as follows.
 
 Let $w \in \Mat_{\FF}(n,1)$ be such that, for $j \in [n]$, 
\begin{itemize}
 \item If $B_{j} \neq 0$, then $\extract(w,j,1) = 1$.
 \item If $B_{j} = 0$, then $\extract(w,j,1) = 0$.
\end{itemize}
More concretely, 
\[w=\lamt_{ij} \langle n,1, \delta_{\lnot B_{j} = 0} \rangle,\]
where ``$B_{j}=0$'' is a $\Sigma^{B}_{0}$-definable relation.
 
Similarly, let $w^{\prime} \in \Mat_{\FF}(m,1)$ be such that, for $j \in [m]$,
\begin{itemize}
 \item If $C_{j} \neq 0$, then $\extract(w',j,1) = 1$.
 \item If $C_{j} = 0$, then $\extract(w',j,1) = 0$.
\end{itemize}

Consider the sequence of index vectors
\[ \ct(w,0), \ct(w,1) , \ldots, \ct(w,n) \in \Mat_{\FF}(n+1,1).\]
Recall that $\ct(w,k)$ denotes the number of $1$'s in $w=[w_{1}, \ldots, w_{n}]^{t}$.
Let
 \[v_{j} := \min \{k \leq n \mid \ct(w,k)=j \} \quad (j=0,\ldots, r).\]
 Note that this is $\Sigma^{B}_{0}$-definable and provably total by open induction.
 Therefore, there exists an open formula $\varphi$ such that
 \[\varphi(k,j) \leftrightarrow v_{j}=k.\]
 Similarly, set
 \[u_{j} := \min \{ k \leq n \mid \ct(w^{\prime},k)=j\}.\]
 There exists an open formula $\varphi^{\prime}$ such that
  \[\varphi^{\prime}(k,j) \leftrightarrow u_{j}=k.\]
  Finally, set
  \begin{align*}
  U &:= \lamt_{ij} \langle m,s, \delta_{\varphi^{\prime}}(i,j) \rangle, \\
  V &:= \lamt_{ij} \langle n,r, \delta_{\varphi}(i,j) \rangle.
  \end{align*}
   We have 
  \[\basis(A)=AVV^{t}, \basis(A^{t})=A^{t}UU^{t}. \]
  
  We show that these $U$ and $V$ are what we want.
  
  We first show that $r=s$.
 Suppose otherwise.
 Without loss of generality, we may assume that $r < s$.
 Consider 
 \[A[U:V] = U^{t}AV \in \Mat_{\FF}(r,s).\]
 Since $r<s$, there exists $b \in \FF^{s} \setminus \{0\}$ such that $A[U:V]b=0$.
 Furthermore, there exists $B \in \Mat_{\FF}(m, m)$ such that 
 \[\basis(A^{t}) B = A^{t} \in \Mat_{\FF}(n,m).\]
 Hence,
 \[A^{t}UU^{t} B = A^{t}.\] 
 Taking the transpose, we obtain
 \[(B^{t} U) U^{t}A  = A \in \Mat_{\FF}(m,n).\]
 Therefore,
 \[0= (B^{t} U) A[U:V] b= AVb = \basis(A) Vb\]
 This forces $Vb=0$, which leads to a contradiction since $V$ is an index matrix. 
 
 Now, we know that $r=s$, and it shows that $A[U:V]$ is indeed a square matrix.
 If $\det(A[U:V]) = 0$, by a similar reasoning above, we obtain a contradiction.
 Therefore, $A[U:V]$ is regular.
 It is maximal since if we add any pair of a new row and a new column, we can obtain the resulting minor $M$ from $A[U:V]$ padded by $0$ multiplied by appropriate matrices, resulting $\det M = 0$. 
\end{proof}
 
 \begin{lemma}[$\LAPPD$]\label{basisofkernel}
  Let $A \in \Mat_{\FF}(m,n)$ and $r=\rk(A)$.
  Then there exists $B \in \Mat_{\FF}(n,n-r)$ such that 
  \begin{enumerate}
  \item\label{coefficientsoflinearcombination} \[ \forall x \in \Mat_{\FF}(n,1). (Ax = 0 \rightarrow \exists c \in \Mat_{\FF}(n-r,1). Bc=x),\]
  \item 
  \[\forall x \in \FF^{n-r}. (Bx=0 \rightarrow x=0).\]
  \end{enumerate}
   that is, column vectors of $B$ is a basis of $\ker (A)$.
 \end{lemma}
 
 \begin{proof}
  Let $M \in \Mat_{\FF}(r,r)$ be the maximal regular submatrix of $A$ whose existence is assured by Lemma \ref{maximalnonsingular}.
  By permutating the coordinates (using counting), we may assume that $A$ is of the form
  \begin{align*}
  A= \begin{bmatrix}
M & y_{1} & \cdots & y_{n-r} \\
 &  *  &  \\
\end{bmatrix},
  \end{align*}
  where each $y_{i}$ is in $\Mat_{\FF}(r,1)$.
  Let 
  \begin{align*}
   w_{i} := \begin{bmatrix}
M^{-1}y_{i} \\
 0  \\
 \vdots \\
 -1 \\
 \vdots \\
 0
\end{bmatrix} \in \Mat_{\FF}(n,1) \quad (i \in [n-r]),
  \end{align*}
  where the component $-1$ above is the $(r+i)$-th one. 
  Then 
  \[B:= [w_{1}, \ldots, w_{n-r}] \in \Mat_{\FF}(n,n-r)\] is the desired basis.
  Indeed, each $w_{i}$ is in $\ker (A)$ and clearly linearly independent.
  
  If $v \in \ker (A)$, then writing it as
  \begin{align*}
   v := \begin{bmatrix}
 \widetilde{v} \\
 c_{1}  \\
 \vdots \\
 c_{n-r}\\
\end{bmatrix} \in \Mat_{\FF}(n,1), \ \mbox{where} \ \widetilde{v} \in \Mat_{\FF}(r,1),
  \end{align*}
  we obtain 
  \[M\widetilde{v} + c_{1} y_{1} + \cdots + c_{n-r}y_{n-r}=0.\]
  Multiplying $M^{-1}$, 
  \[\widetilde{v} = - c_{1} M^{-1}y_{1} - \cdots - c_{n-r}M^{-1}y_{n-r}\]
  follows.
  Thus, we have
  \[v = \sum_{i=1}^{n-r} c_{i}w_{i}. \]
 \end{proof}

\begin{cor}[$\LAPPD$]\label{uniformspan}
In the setting of the previous Lemma \ref{basisofkernel}, 
\[ \forall X \in \Mat_{\FF}(n,k). (AX = 0 \rightarrow \exists C \in \Mat_{\FF}(n-r,k). BC=X).\]
\end{cor}
\begin{proof}
In the proof of Lemma \ref{basisofkernel} (\ref{coefficientsoflinearcombination}) above, $c$ was uniform for $x$.
\end{proof}

\begin{thm}[$\LAPPD$]\label{basic property}
 For any $A \in \Mat_{\FF}(m,n)$, $\rk(A)=\rank(A)$.
\end{thm}
\begin{rmk}
Note that the LHS is an abused notation introduced in Definition \ref{indexvector}.
\end{rmk}

\begin{proof}
Let $N:=\row (A) + \column (A)$.
 Let $C := \polize (A)$. 
 First, we construct a basis
 \begin{align}\label{basisofker(C)}
 b_{1}, \ldots, b_{s} \in \FF(X)^{N}
 \end{align}
  of $\ker (C)$ as an $\FF(X)$-vector space and a basis
  \begin{align}\label{basisofim(C)}
  b_{s+1}, \ldots, b_{s+t} \in \FF(X)^{N}
  \end{align}
  of $\im(C)$ as an $\FF(X)$-vector space.
  
  To be concrete, we take a basis of $\ker(\symm(A))$ constructed in Lemma \ref{basisofkernel} as the basis (\ref{basisofker(C)}).
  (Note that $\symm(A) \in \Mat_{\FF}$.)
  Indeed, if $b_{1}, \ldots, b_{s}$ is a basis of $\ker (\symm(A))$, then they also belong to $\ker (C)$.
  Moreover, suppose $Cv = 0$ for $v \in \FF(X)^{N}$.
  Recalling that $C= \chi_{N} \symm(A)$ and $\chi_{N}$ is invertible in $\Mat_{\FF(X)} (N,N)$,
  \[\symm(A) v = 0.\]
  Clearing the denominator, there exists $g \in \FF[X]$ such that
  \[ g v \in \FF[X]^{N}.\]
  Write
  \[gv = v_{0} X^{0} + \cdots + v_{d} X^{d},\]
  where each $v_{i} \in \FF^{N}$.
  We see that
  \[0 =g\symm(A) v =\symm(A) (g v) = \symm(A) v_{0} X^{0} + \cdots + \symm(A) v_{d}   X^{d} \]
  in $\Mat_{\FF[X]}$. (cf. Definition \ref{matrixofpolynomials}.)
  
  Looking at the coefficients, we obtain
  \[\symm(A) v_{i} = 0 \quad (0 \leq i \leq d).\]
  Since $b_{1}, \ldots, b_{s}$ is a basis of $\ker (\symm(A))$, there exists $D \in \Mat_{\FF}(s,d+1)$ such that
  \begin{align*}
     [b_{1}, \ldots, b_{s}]D = \begin{bmatrix}
 v_{0}, \ldots, v_{d}
\end{bmatrix} \in \Mat_{\FF}(N,d+1)
  \end{align*}
  by Corollary \ref{uniformspan}.
  Therefore,
  \[gv = [b_{1}, \ldots, b_{s}] D \begin{bmatrix}
 X^{0} & 0 & \cdots & 0  \\
 0 & X_{1} & \cdots & 0 \\
 0 & 0 & \ddots & 0 \\
 0& \cdots & 0& X_{d}
\end{bmatrix}\]
  in $\Mat_{\FF(X)}$.
  Thus $b_{1},\ldots, b_{s}$ serve as the first basis (\ref{basisofker(C)}).
  
  Next, we consider the second basis (\ref{basisofim(C)}).
  It is given by first calculating that of $\symm(A)$ and then multiplying $\chi_{N}$ to them.
  The construction is valid by a similar reason as above.  
 
 Because of Lemma \ref{zerointersection} and Lemma \ref{decomposition}, $b_{1}, \ldots, b_{s+t}$ form a basis of $\FF(X)^{N}$. 
 From the linearly independence, we have $s+t \leq N$ by Lemma \ref{dimension}.
 
 Moreover, we have $s+t=N$.
 Indeed, because of the uniformity of the decomposition given in Lemma \ref{decomposition} and the linear combinations of the bases given in Lemma \ref{basisofimage}, Lemma \ref{basisofkernel} and the argument for $b_{1},\ldots, b_{s}$ above, there exists $C \in \Mat_{\FF(X)}(s+t,N)$ such that
 \[[b_{1},\ldots, b_{s+t}]C=I_{N} \in \Mat_{\FF(X)}(N,N).\]
 If $s+t<N$, then we can pad the matrix $[b_{1},\ldots, b_{s+t}]$ and $C$ by zeros, and obtain
 \[[b_{1},\ldots, b_{s+t},\zeromat_{N,N-(s+t)}] \begin{bmatrix}C \\ \zeromat_{N-(s+t),N} \end{bmatrix} = I_{N},\]
 which is now a product of square matrices. 
 Then, taking the determinant in $\FF(X)$ and using its multiplicativity (which is an axiom of $\LAPPD$), we have $0=1$, a contradiction.

 Let $P := [b_{1}, \ldots, b_{s+t}] \in \Mat_{\FF(X)}(N,N)$.
 Then we can write
 \begin{align*}
 C P = P \begin{bmatrix}
 \zeromat_{s,s} & \zeromat_{s,t} \\
 \zeromat_{t,s} & B \\
 \end{bmatrix},
 \end{align*}
 where $B \in \Mat_{\FF(X)}(t,t)$.
 Hence, 
  \begin{align*}
 P^{-1}C P =
  \begin{bmatrix}
 \zeromat_{s,s} & \zeromat_{s,t} \\
 \zeromat_{t,s} & B \\
 \end{bmatrix}.
 \end{align*}
 Taking the characteristic polynomial of both sides and applying Proposition \ref{twodefsofdet},
 \begin{align*}
 \det(YI - P^{-1}CP) = Y^{s} p_{B}(Y)
  \end{align*}
  in $\FF(X)[Y]$.
  Here, we can further show that
  \begin{align}\label{conjinv}
  \det(YI - P^{-1}CP) = \det(YI - C).
  \end{align}
  Indeed, substituting $X^{i}$ for $Y$ ($i \in \MM$), we have
  \begin{align*}
  \det(X^{i}I - P^{-1}CP) &= \det(P^{-1}(X^{i}I - C)P)\\
  &=\det(P^{-1}) \det(X^{i}I - C) \det(P)\\
  &=\det(X^{i}I - C).
  \end{align*}
  (For the last equalities, we have used $\llbracket \mbox{$\det$ is multiplicative} \rrbracket_{rat}$.)
  Since 
  \[\langle \MM, \FF(X), \Mat_{\FF(X)} \rangle \models \LAP,\] by Proposition \ref{identitytheorem}, the equality (\ref{conjinv}) follows.
  
 Now, recalling Lemma \ref{zerointersection}, we have $\ker (C) \cap \im(C) = 0$.
 Hence, it follows that $\ker (B)=0$, that is, $p_{B}(0) \neq 0$.
 Therefore, 
 \[s = \mul(\llbracket \ch \rrbracket_{rat} (\polize (A))).\]
 Note that the RHS appears in the definition of $\rank(A)$. (cf. Definition \ref{formalizedMulmuley} (\ref{explicitdefofrank}))
 
 On the other hand, by Lemma \ref{basisofkernel}, we can construct $n-r$ many vectors forming a basis of $\ker (A)$ and $m-r$ many vectors forming a basis of $\ker (A^{t})$.
 These give $n+m-2r$ many vectors consisting of a basis of $\ker (C)$.
 
 Hence, $s=n+m-2r$, which gives $r = \rank(A)$.
\end{proof}

\section{Interpretation by $VNC^{2}$}\label{Interpretation}
In this section, we show that linear algebra with rational coefficients formalized in $\VNC^{2}$ satisfies $\LAPPD$ (for a formal definition of the bounded arithmetic $\VNC^{2}$, see \cite{Cook}).

Since we would like to use several $\bold{NC}^{2}$-functions as terms, we utilize $\widehat{\VNC^{2}}$ instead of $\VNC^{2}$.
It is known that $\VNC^{2}$ (or $\widehat{\VNC^{2}}$) is quite powerful to carry out basic arithmetical operations, matrix operations over integers coded by the string-sort, and manipulation of arithmetic circuits (here, an \textit{arithmetic circuit} is a circuit whose gates are among $+,*$ and inputs are either variables or field elements and computes a polynomial. We follow the treatment in \cite{Uniform}). 
More concretely, \cite{Uniform} showed that the soundness of $\PP_{c}(\ZZ)$ can be formalized and proved in $\VNC^{2}$.
Upon this result, we put priority on readability and stick to less formal descriptions compared with the previous sections.

Consider the following interpretation:
\begin{defn}
We define an interpretation $\mathcal{I}$ of $\mathcal{L}_{\LAP}$ by $\VNC^{2}$ as follows:
 \begin{enumerate}
  \item The universe of $\indexsort$ is interpreted by the number-sort of $\VNC^{2}$, that is, 
  \[\alpha_{index}(x) :\equiv (x =_{number} x).\]
  \item The universe of $\fieldsort$ is interpreted by the rationals formalized in $\VNC^{2}$ using the string-sort, that is,
  \[\alpha_{field}(P,Q) :\equiv \mbox{$Q$ is a code of a nonzero integer}\]
  (Informally, $P$ and $Q$ are codes of integers, such as binary expansions of two integers, and the pair $(P,Q)$ is regarded as a representation of the rational $P/Q$).
  \item The universe of $\matrixsort$ is interpreted by the matrix of rational coefficients formalized in $\VNC^{2}$, that is,
  $\alpha_{matrix}(A,m,n,Q)$ is a natural formula expressing:
  $A$ is a code of a ($m \times n$)-matrix of integer coefficients which are at most $|A|$-bits, and $Q$ is a code of a nonzero integer.

  (Informally, $(A,Q)$ represents a matrix $\frac{1}{Q}A$ with rational coefficients).
  \item\label{basics} For each relation symbol $\sigma$ (including $=$), $\llbracket \cdot \rrbracket_{pol}$ are all defined naturally according to the informal interpretations above.
  More concretely:
  \begin{enumerate}
   \item $\beta_{=_{index}}(x,y) :\equiv x=_{num}y$.
   \item $\beta_{\leq_{index}}(x,y) :\equiv x\leq_{num}y$.
   \item $\beta_{=_{field}} (P,Q, X,Y) :\equiv P \times Y = Q \times X$.
   Here, $\times$ on RHS is the integer multiplication formalized in $\VNC^{2}$.
   \item $\beta_{=_{matrix}} (A,m,n,Q,B,k,l,R) := (m=k \land n=l \land R*A = Q*B)$.
   Here, $*$ on RHS is the scalar multiplication of an integer and an integer-coefficient matrix formalized in $\VNC^{2}$.
  \end{enumerate}
  \item In order to define $f_{\mathcal{I}}$ for each function symbol $f$, we define $t_{\mathcal{I}}$ for each $\mathcal{L}_{\LAP}$-term $t$, which is:
  \begin{itemize}
   \item a term with number-sort output if the output of $t$ is $\indexsort$.
   \item a tuple $(P,Q)$ of terms with $\indexsort$ output if the output of $t$ is $\fieldsort$.
   \item if the output of $t$ is $\matrixsort$,
 a tuple $(A,m,n,Q)$ of terms, where the outputs of $m$ and $n$ are number-sort, and those of $A$ and $Q$ are string-sort. 
  \end{itemize}
   will coincide with the interpretation of terms induced by $\mathcal{I}$ in the sense of Definition \ref{interpretation}.   \begin{enumerate}
    \item For each index-variable $x$, prepare a fresh number-sort variable $n_{x}$ and set
    \[x_{\mathcal{I}}:=n_{x}.\]
    \item For each field-variable $a$, prepare fresh string-variables $P_{a}$ and $Q_{a}$ and set 
    \[ a_{\mathcal{I}}:=(P_{a},Q_{a})\]
    \item For each matrix-variable $A$, prepare fresh index-variables $m_{A},n_{A}$ and fresh string-variables $X_{A}$ and $Q_{A}$ and set
    \[A_{\mathcal{I}} := (X_{A},m_{A},n_{A},Q_{A}).\]
    \item If $t=f(t_{1}, \ldots, t_{k})$, where $f$ is among 
    \begin{align*}
    &0_{index}, 1_{index},+_{index},*_{index},-_{index}, \divi, \rem, \\
    &0_{field}, 1_{field}, +_{field}, *_{field}, (\cdot)^{-1}_{field},\\
    & \row, \column, \extract ,\sum, \mathtt{P},
    \end{align*}
     then
    \[t_{\mathcal{I}} := \tau_{f}((t_{1})_{\mathcal{I}}, \ldots, (t_{k})_{\mathcal{I}}).\]
    Here, $\tau_{f}$ denotes a natural $\Sigma^{B}_{1}$-formalization of $f$ carried out in $\bold{NC}^{2}$, implemented as a term in $\widehat{\VNC^{2}}$.

For example, if $t=\row (A)$, then $(A)_{\mathcal{I}}$ should be a quadruple $(X,m,n,Q)$ of terms, and set
    \[t_{\mathcal{I}} := m.\]

Note that the bounded sum $\Sigma$ and the powering function $P$ are both known to be in $\VNC^{2}$, and therefore there are $\mathcal{L}_{\widehat{\VNC^{2}}}$-terms representing these functions.
Furthermore, if we choose natural ones, it is also known that some of their basic properties are formalized and proven in $\VNC^{2}$.
    
   \item If $t=\cond_{index}(\alpha, t_{1},t_{2})$, then
    \[t_{\mathcal{I}} := \cond_{index}(\alpha_{\mathcal{I}}, ( t_{1})_{\mathcal{I}},(t_{2})_{\mathcal{I}} ).\]
    Here, $(\alpha)_{\mathcal{I}}$ denotes the interpretation of $\alpha$ under $\mathcal{I}$, which is already defined by induction.
    Note that $\alpha$ is an open formula whose atomic subformulae are all in $\indexsort$, and therefore it is easy to see by induction that $\alpha_{\mathcal{I}}$ is an open formula whose atomic subformulae are all in number-sort.
    \item If $t = \cond_{field}(\alpha,t_{1},t_{2})$, the definition of $t_{\mathcal{I}}$ is completely analogous.

         
    \item If $t=\lamt_{ij} \langle m,n,u \rangle$ then
we set $t_{\mathcal{I}}$ to be the $(m \times n)$-matrix $\{u(i,j)\}_{i \in [m], j \in [n]}$.
Note that by $\Sigma^{B}_{0}$-CA of $\widehat{VNC^{2}}$, this function is $\Sigma^{B}_{1}$-definable and provably total in $\widehat{\VNC^{2}}$, and therefore there is an $\mathcal{L}_{\widehat{\VNC^{2}}}$-term corresponding to this definition.
    
   \end{enumerate} 
 \end{enumerate}

\end{defn}

\begin{thm}
 $\mathcal{I}$ is an interpretation of $\LAPPD$ by $\VNC^{2}$.
\end{thm}

\begin{proof}
 It is straightforward to check the axioms of $\LAP$.
 We focus on checking the axiom $(MDP)$.
 We work in $\VNC^{2}$.
 
 Let $\DET(A)$ be the arithmetic circuit in $\VNC^{2}$ that computes the determinant as given in \cite{Uniform}.
 In general, we can convert arithmetic circuit without division gates into circuit with rational inputs as follows:

 Let $F(x_{1}, \ldots, x_{k})$ be an arithmetic circuit without division with inputs $x_{1}, \ldots, x_{k}$.
 We transform $F$ into a circuit $F^{*}(u_{1}, \ldots, u_{k},v_{1}, \ldots, v_{k})$ by replacing each input gate $x_{i}$ with a division gate $u_{i}/v_{i}$.
 Let $A$ denote the assignment
 \[x_{1} := \frac{N_{1}}{D_{1}}, \ldots, x_{k} :=\frac{N_{k}}{ D_{k}}\]
  (Where each $N_{i}$ and $D_{i}$ are codes of integers).
 Also, let $A^{*}$ be the assignment 
\[ A^{*}:= u_{1}:=N_{1}, \ldots, u_{k}:=N_{k}, v_{1} := D_{1}, \ldots, v_{k}:=D_{k}.\]

 Consider $\Eval_{alg}$, $\Den(G)$ and $\Num(G)$ given in \cite{Uniform},
 and define 
 \[\Eval(F,A) := \frac{\Eval_{alg}(\Num(F^{*}),A^{*})}{\Eval_{alg}(\Den(F^{*}),A^{*}) }.\]
 
 Intuitively, $\Eval(F,A)$ computes the output of $F$ with rational input $A$.
 
 It can be proved that:
 \begin{claim}[$\VNC^{2}$]\label{soundness}
 If $F=G$ has a $P_{c}(\ZZ)$-proof $\pi$, then for any assignment $A$ to the variables in $\pi$, we have 
 \[\Eval(F,A) \equiv \Eval(G,A).\]
 \end{claim}
 
 \begin{proof}[Proof of the Claim]
Recall that $\Eval_{alg}$ is a $\Sigma^{B}_1$-definable function of $\VNC^{2}$, which implies that 
  we can use an induction on the number of proof-lines in $\pi$ to prove the claim.
  The argument is split into cases along the last rule or axiom (see $(R1)$-$(R4)$, $(A1)$-$(A10)$ and $(C1)$-$(C2)$ in \cite{Uniform}) in $\pi$.
  
  The case $(R1)$-$(R4)$ and $(C1)$-$(C2)$ are easy.
  
  As for the cases $(A1)$-$(A10)$, we check the most difficult case $(A6)$:
  \[F \cdot (G+H) = F \cdot G + F \cdot H.\]
  
  For each circuit $C$, let $\widetilde{C}$ be the circuit substituting $u_{i}/v_{i}$ for each variable $x_{i}$ in $C$, where $u_{i}$ and $v_{i}$ are distinct fresh variables.
  Furthermore, let $B$ be the integer assignment for $u_{i}$'s and $v_{i}$'s induced by the given rational assignment $A$.
  
  Then the evaluations can be computed as follows:
  \begin{align*}
   &\Eval(F \cdot (G+H),A) \\
   &\equiv \frac{\Eval_{alg}(\Num (\widetilde{F \cdot (G+H)}), B)}{\Eval_{alg}(\Den (\widetilde{F \cdot (G+H)}), B)}\\
   &\equiv \frac{\Eval_{alg}(\Num (\widetilde{F}) \cdot \Num (\widetilde{(G+H)}), B)}{\Eval_{alg}(\Den (\widetilde{F}) \cdot \Den(\widetilde{(G+H)}), B)}\\
   &\equiv \frac{\Eval_{alg}(\Num (\widetilde{F}) \cdot (\Num (\widetilde{G})\cdot \Den (\widetilde{H}) +\Num (\widetilde{H})\cdot \Den (\widetilde{G})), B)}{\Eval_{alg}(\Den (\widetilde{F}) \cdot (\Den(\widetilde{G}) \cdot \Den(\widetilde{H})), B)}\\
   &\equiv \frac{\Eval_{alg}(\Num (\widetilde{F}) \cdot \Num (\widetilde{G})\cdot \Den (\widetilde{H}),B) +\Eval_{alg}(\Num (\widetilde{F}) \cdot \Num (\widetilde{H})\cdot \Den (\widetilde{G})), B)}{\Eval_{alg}(\Den (\widetilde{F}) \cdot \Den(\widetilde{G}) \cdot \Den(\widetilde{H}), B)}\\
   &\equiv \frac{\Eval_{alg}(\Num (\widetilde{F}) \cdot \Num (\widetilde{G})\cdot \Den (\widetilde{H}),B)}{\Eval_{alg}(\Den (\widetilde{F}) \cdot \Den(\widetilde{G}) \cdot \Den(\widetilde{H}), B)}
    +\frac{\Eval_{alg}(\Num (\widetilde{F}) \cdot \Num (\widetilde{H})\cdot \Den (\widetilde{G})), B)}{\Eval_{alg}(\Den (\widetilde{F}) \cdot \Den(\widetilde{G}) \cdot \Den(\widetilde{H}), B)}\\
   &\equiv \frac{\Eval_{alg}(\Num (\widetilde{F}) \cdot \Num (\widetilde{G}),B)\cdot \Eval_{alg}(\Den (\widetilde{H}),B)}{\Eval_{alg}(\Den (\widetilde{F}) \cdot \Den(\widetilde{G}),B) \cdot \Eval_{alg}(\Den(\widetilde{H}), B)} \\
    &\quad +\frac{\Eval_{alg}(\Num (\widetilde{F}) \cdot \Num (\widetilde{H}),B) \cdot \Eval_{alg}( \Den (\widetilde{G}), B)}{\Eval_{alg}(\Den (\widetilde{F})  \cdot \Den(\widetilde{H}), B) \cdot \Eval_{alg}(\Den(\widetilde{G}),B)}\\
  &\equiv \frac{\Eval_{alg}(\Num (\widetilde{F}) \cdot \Num (\widetilde{G}),B)}{\Eval_{alg}(\Den (\widetilde{F}) \cdot \Den(\widetilde{G}),B) }
  +\frac{\Eval_{alg}(\Num (\widetilde{F}) \cdot \Num (\widetilde{H}),B)}{\Eval_{alg}(\Den (\widetilde{F})  \cdot \Den(\widetilde{H}), B)}
  \end{align*}
  
  and
  
  \begin{align*}
   &\Eval(F \cdot G + F \cdot H, A) \\
   &\equiv \frac{\Eval_{alg}( \Num(\widetilde{F \cdot G + F \cdot H}), B)}{\Eval_{alg}( \Den(\widetilde{F \cdot G + F \cdot H}), B)}\\
   &\equiv \frac{\Eval_{alg}( \Num(\widetilde{F \cdot G}) \Den(\widetilde{F \cdot H}) + \Num(\widetilde{F \cdot H}) \Den(\widetilde{F \cdot G})), B)}{\Eval_{alg}( \Den(\widetilde{F \cdot G}),B) \Eval_{alg}(\Den (\widetilde{F \cdot H}), B)} \\
   &\equiv \frac{\Eval_{alg}( \Num(\widetilde{F \cdot G}) \Den(\widetilde{F \cdot H}),B)}{\Eval_{alg}( \Den(\widetilde{F \cdot G}),B) \Eval_{alg}(\Den (\widetilde{F \cdot H}), B)} \\
   &\quad + \frac{\Eval_{alg}(\Num(\widetilde{F \cdot H}) \Den(\widetilde{F \cdot G})), B)}{\Eval_{alg}( \Den(\widetilde{F \cdot G}),B) \Eval_{alg}(\Den (\widetilde{F \cdot H}), B)} \\
   &\equiv \frac{\Eval_{alg}( \Num(\widetilde{F \cdot G}),B) \Eval_{alg}( \Den(\widetilde{F \cdot H}),B)}{\Eval_{alg}( \Den(\widetilde{F \cdot G}),B) \Eval_{alg}(\Den (\widetilde{F \cdot H}), B)} \\
   &\quad + \frac{\Eval_{alg}(\Num(\widetilde{F \cdot H}),B) \Eval_{alg}( \Den(\widetilde{F \cdot G})), B)}{\Eval_{alg}( \Den(\widetilde{F \cdot G}),B) \Eval_{alg}(\Den (\widetilde{F \cdot H}), B)} \\
  &\equiv \frac{\Eval_{alg}( \Num(\widetilde{F \cdot G}),B) }{\Eval_{alg}( \Den(\widetilde{F \cdot G}),B) } 
   + \frac{\Eval_{alg}(\Num(\widetilde{F \cdot H}),B) }{\Eval_{alg}(\Den (\widetilde{F \cdot H}), B)} \\
   &\equiv \frac{\Eval_{alg}(\Num (\widetilde{F}) \cdot \Num (\widetilde{G}),B)}{\Eval_{alg}(\Den (\widetilde{F}) \cdot \Den(\widetilde{G}),B) }
  +\frac{\Eval_{alg}(\Num (\widetilde{F}) \cdot \Num (\widetilde{H}),B)}{\Eval_{alg}(\Den (\widetilde{F})  \cdot \Den(\widetilde{H}), B)}
  \end{align*}
  (all equalities above are easy consequences of soundness of $\Eval_{alg}$ with respect to $\mathbb{P}_{c}(\ZZ)$-proofs and the definitions of $\Num$ and $\Den$)
  
  Therefore, the result follows.
 \end{proof}
 
 Now, we go back to $\VNC^{2}$. 
 Let $A=(f_{ij})_{i,j \in [n]}$ be a matrix with polynomial coefficients and $\deg f_{ij} \leq d$ for all $i,j \in [n]$.
 
 From the circuit $\Det_{balanced}$ in \cite{Uniform}, we can construct a circuit $\Det_{\VNC^{2}}$ which receives input $A$ and outputs coefficients of the determinant of $A$.
 Note that $\det(A)$ is a polynomial of degree $\leq dn$.
 
 Since $\Det_{\VNC^{2}}(A)$ receives rational inputs, we use the function evaluate the output of $\Det_{\VNC^{2}}(A)$.
 Using claim, we can show that
 \[\Det_{\VNC^{2}}(AB) =\Det_{\VNC^{2}}(A)\Det_{\VNC^{2}}(B)\]
 is provable in $\VNC^{2}$.
 
 To finish the proof, it suffices to show that the determinant defined by Berkowitz's algorithm (carried out in \cite{The proof complexity of linear algebra}) coincides with $\Det_{\VNC^{2}}(A)$.
  
 Indeed, the proof can be carried out as follows: in $\LAP$, 
\begin{align}
 \ch(A):= \column (1,A) \cdots \column (n,A) \label{characteristicpolynomial},
 \end{align}
 and 
 \begin{align}
  \det(A) = \ch(A)(0). \label{detviach}
 \end{align}
 Here, $\column (k,A)$ is described as follows: let $A = (a_{ij})_{i,j \in [n]}$.
 Set
 \begin{align}
  A_{k} &:= (a_{(i+k-1),(j+k-1)})_{i,j \in [n-k+1]} \quad (1 \leq k \leq n),\\
    \begin{bmatrix}
   a_{kk} & R_{k}\\
   S_{k} & A_{k+1}\\
  \end{bmatrix}
  &:= A_{k} \quad (1 \leq k \leq n-1),\\
  \column (k,A) &:=
  \begin{bmatrix}
   1 & 0  & & 0\\
   -a_{kk}& 1 &&0 \\
   -R_{k}S_{k} & -a_{kk} &\cdots & 0\\
   & \ddots &  \\
   -R_{k}A_{k}^{n-k-2}S_{k} & -R_{k}A_{k}^{n-k-3}S_{k} & & 1 \\
   -R_{k}A_{k}^{n-k-1}S_{k} & -R_{k}A_{k}^{n-k-2}S_{k} & & -a_{kk} 
  \end{bmatrix} \label{generalform}\\
   &\in \Mat_{\FF}(n-k+2,n-k+1)
  \quad (1 \leq k \leq n-1), \\
  \column (n,A) &:= 
  \begin{bmatrix}
   1\\
   -a_{nn}
  \end{bmatrix} \label{cornerform}.
 \end{align}
 
 Since
\begin{align*}
\LAP \vdash (\Lambda,d) &\in \Mat_{\FF[X]}(m,m) \rightarrow  \\
&\detpol(\Lambda,d) = \coeff_{pol}(\column_{pol}(1,\Lambda,d)\cdots \column_{pol}(n,\Lambda,d), 0), 
\end{align*}
by the interpretability of $\LAP$ by $\VNC^{2}$,
\begin{align*}
\VNC^{2} \vdash ((\Lambda,d) &\in \Mat_{\FF[X]}(m,m))_{\mathcal{I}} \rightarrow \\
&(\detpol(\Lambda,d))_{\mathcal{I}} =(\coeff_{pol})_{\mathcal{I}} \left( \column_{pol}(1,\Lambda,d)_{\mathcal{I}} \cdots  \column_{pol}(n,\Lambda,d)_{\mathcal{I}},0 \right) \end{align*}
(note that iterated multiplication of matrices of polynomial coefficient can be formalized in $\VNC^{2}$ and it is used in the RHS).

Therefore, in $\VNC^{2}$, the following holds; let $(M,m(d+1),m,Q)$ be a code of 
a matrix of rational coefficients.
Assume $A$ is of the following form;
\begin{align*}
  A = \begin{bmatrix}
   A_{0} \\
   \vdots \\
   A_{d}
  \end{bmatrix}
  \mbox{(each} \ A_{i} \ \mbox{is an $(m \times m)$-matrix of integer coefficients)},
 \end{align*}
 that is, informally, $A$ codes the matrix 
 \[ \tilde{A} = A_{0}+A_{1}X + \cdots +A_{d}X^{d}\]
 of integer-coefficient polynomial coefficients.
 Note that the original tuple $(M,m(d+1),m,Q)$ codes the matrix $\hat{A}:=\frac{1}{Q}\tilde{A}$ of rational-coefficient polynomial coefficients.
 
 Then $\chi(\hat{A}) := (\detpol)_{\mathcal{I}}(M,m(d+1),m,Q)$ is computed according to Berkowitz algorithm described by (\ref{characteristicpolynomial})-(\ref{cornerform}).
 
Now, our goal is to show that 
\[\chi(\hat{A})=\Det_{\VNC^{2}}(YI-\hat{A})\]
as polynomials of a fresh indeterminate $Y$.

 Once this is established, 
 \begin{align*}
  (\detpol)_{\mathcal{I}}(\hat{A}) 
  &= (-1)^{m}\chi(\hat{A}) [0/Y]\\
  &= (-1)^{m}\Det_{\VNC^{2}}(YI-\hat{A}) [0/Y]\\
  &=  (-1)^{m}\Det_{\VNC^{2}} (-\hat{A})\\
  &= \Det_{\VNC^{2}} (\hat{A})
 \end{align*}
 (the last two equalities follow from the soundness of $\PP_{c}(\ZZ)$-proof established in $\VNC^{2}$).
 
Indeed, by \cite{Uniform}, $\VNC^{2}$ can formalize and prove cofactor-expansion and Cayley-Hamilton theorem for matrices of polynomial coefficient.
By Claim \ref{soundness} above, we can extend this result to the matrices of rational coefficient, too.

 Therefore, the proofs of Lemma 4.2.1 and Lemma 4.2.2 in \cite{Soltys} can be formalized and proven in $\VNC^{2}$.
 
 Let $A$ be an $(n \times n)$-matrix of rational coefficients in general, and
 \begin{align*}
  A=\begin{bmatrix}
   a_{11} & R \\
   S & M
  \end{bmatrix}.
 \end{align*}
 
 Then by Lemma 4.2.1 of \cite{Soltys}, we have
 \[\chi(A)(Y) = (Y-a_{11})\chi(M) - R \adj(YI-M) S,\]
 where $\adj(N)$ is the adjoint of $N$.
 
 Let 
 \[\chi(M) = q_{n-1}Y^{n-1} + \cdots + q_{1}Y+q_{0}\]
 and 
 \[B(Y) = \sum_{2 \leq k \leq n} (q_{n-1}M^{k-2} + \cdots + q_{n-k+1}I)Y^{n-k}.\]
 Then $B(Y) = \adj (YI-M)$.
 
 By substituting $B(Y)$ for $\adj(YI-M)$, we obtain
 \begin{align*}
  \chi(A)(Y) &= (Y-a_{11})\chi(M) - RB(Y) S\\
  &= (Y-a_{11})\chi(M) - \sum_{2 \leq k \leq n} (q_{n-1}RM^{k-2}S + \cdots + q_{n-k+1}RS)Y^{n-k}.
 \end{align*}
 Therefore, we have
 \begin{align*}
  \begin{bmatrix}
   p_{n} \\
   \vdots \\
   p_{0}
  \end{bmatrix}
  =  \column (1,A)
  \begin{bmatrix}
   q_{n-1} \\
   \vdots \\
   q_{0}
  \end{bmatrix}.
 \end{align*}
Hence, by $\Sigma^{B}_{0}$-induction in $\widehat{\VNC^{2}}$, we conclude that
\[\chi(\hat{A})=\Det_{\VNC^{2}}(YI-\hat{A}).\]
 
\end{proof}

\begin{cor}
$\VNC^{2}$ can $\Sigma^{B}_{1}$-define $\rank(A)$ and proves its basic properties like Theorem \ref{basic property}.
\end{cor}

\begin{proof}
Just interpret the all results of $\LAPPD$ by $\mathcal{I}$.
Note that each $\Sigma^{B}_{0}$ formula in $\mathcal{L}_{\LAP}$ is interpreted into a $\Sigma^{B}_{0}$ formula in $\widehat{\VNC^{2}}$.
\end{proof}

\section{Some combinatorial results provable in $\VNC^{2}$}\label{Some combinatorial results}
First, note that the Oddtown theorem, the Graham-Pollak theorem, and the Fisher inequality discussed in \cite{Hard} are provable by quasipolynomial-sized Frege proofs as \cite{Hard} conjectured.
More precisely, they can be proven in $\VNC^{2}$ and this fact is a corollary of the result of \cite{Uniform}.
Indeed, knowing $\VNC^{2}$ proves determinant for matrices of integer coefficients (therefore also $\ZZ_{2}$-coefficients) is multiplicative, we see $\VNC^{2}$ proves that if $\det(A)=0$, then there exists a nontrivial $b$ such that $Ab=0$ (by the fact that the multiplicativity of determinant implies Cayley-Hamilton theorem and Lemma \ref{CHimpliesnontrivialsolution}).
Moreover, $\VNC^{2}$ can treat low-degree polynomials (cf. \cite{Uniform}) and also $\VNC^{2}$ can prove that sums of squares of integers (coded by the string-sort) are nonnegative.

A more careful analysis also shows that the following variation of Ray-Chaudhuri-Wilson theorem, another candidate of possibly difficult statements mentioned in \cite{Hard}, is also provable in $\VNC^{2}$:

\begin{thm}\label{nonuniformRCW}
Let $s \in \NN$. Then $\VNC^{2}$ proves the following:
Let $n,m$ be numbers and $\mathcal{F},L$ be strings.
Suppose $\mathcal{F}$ codes a family $\{A_{1}, \ldots, A_{m}\}$ of subsets of $[n]$ and $L$ codes a set $\{l_{1},\ldots, l_{s}\}$ of numbers.
Assume also $\mathcal{F}$ is $L$-intersecting, that is, $\forall i \neq j \in [m]. \#(A_{i} \cap A_{j}) \in L$.
Then the following holds:
\begin{align}\label{informalRCW}
m \leq \sum_{i=0}^{s} \binom{n}{i}
\end{align}
Note that the RHS is $\Sigma^{B}_{0}$-definable since $s$ is standard.
\end{thm}

Before the proof of the theorem above, we put useful facts on binomial coefficients here:

\begin{lemma}\label{binomialcoeff}
Let $s \in \NN$. Then $\VNC^{2}$ proves the following:
\begin{enumerate} 
 \item\label{beyond} For any numbers $n$ and $i \leq s$, if $i \not\in [0,n]$, then $\binom{n}{i}=0$. Also, $\binom{n}{0}=1$.
 \item\label{leader} For any numbers $n$ and $i\leq s$,
 \[\binom{n+1}{i+1} = \binom{n}{i}+\binom{n}{i+1}.\]
 \item\label{casesformin} For any numbers $n$ and $i\leq s$,
 \[\binom{n+1}{i+1} = \sum_{j=i}^{n}\binom{j}{i}.\]
 Note that bounded sums in the number sort on the RHS are formalized in $VTC^{0}$.
\end{enumerate}
\end{lemma}

\begin{proof}[Proof of Theorem \ref{nonuniformRCW}]
First, for each $t \in [0,s]$, let
\begin{align*}
&\varphi_{t}(i_{1},\ldots,i_{t}) :\equiv (i_{1}<i_{2} < \cdots<i_{t})\land \\
&\left(x=\sum_{j=0}^{t}\binom{n}{j} + \sum_{k=1}^{i_{1}} \binom{n-k}{t} + \sum_{k=i_{1}+1}^{i_{2}}\binom{n-k}{t-1} + \cdots + \sum_{k=i_{t-1}+1}^{i_{t}} \binom{n-k}{1}\right).
\end{align*}

Now, we work in $\VNC^{2}$.
Suppose $n,m,s$, $\mathcal{F}$ and $L$ violate the statement.
Then $m > \sum_{i=0}^{s}\binom{n}{i}$.
By Lemma \ref{binomialcoeff} (\ref{casesformin}), we have
\begin{align*}
&\forall x \in \left[\sum_{i=0}^{s}\binom{n}{i}\right]. \bigvee_{t=0}^{s}\left(\exists (i_{1},\ldots, i_{t}) \in [n]^{t}.  \varphi_{t}(i_{1},\ldots,i_{t}) \right).
\end{align*}
Moreover, for each $x$, the witness $t$ of the disjunction above is unique.
Furthermore, for any string $S$ coding a subset of $[n]$ and satisfying $\#S =t\leq s$, we have an enumeration $(i_{1},\ldots,i_{t}) \in[n]^{t}$ of the elements of $S$ satisfying $i_{1}<\cdots<i_{t}$, and 
\[\sum_{j=0}^{t}\binom{n}{j} + \sum_{k=1}^{i_{1}} \binom{n-k}{t} + \sum_{k=i_{1}+1}^{i_{2}}\binom{n-k}{t-1} + \cdots + \sum_{k=i_{t-1}+1}^{i_{t}} \binom{n-k}{1} \in \left[\sum_{i=0}^{s}\binom{n}{i}\right].\]

Therefore, we can regard the set $\left[\sum_{i=0}^{s}\binom{n}{i}\right]$ as a feasible enumeration of the all subsets of $[n]$ having cardinality $\leq s$.

Now, we follow a standard proof of non-uniform Ray-Chaudhuri-Wilson theorem.
We identify each $A_{i}$ as their characteristic vector $v_{i} \in \{0,1\}^{n}$.
For each $i \in [m]$, consider the multivariate polynomial
\[F_{i}(X_{1},\ldots,X_{n}) := \prod_{k \in [s],l_{k} <\#A_{i}}\left(\sum_{j=1}^{n}(v_{i})_{j} X_{j} - l_{k}\right).\]
coded in the string-sort as a natural arithmetic circuit.

We can evaluate the polynomials of integer-coefficients at $(X_{1},\ldots, X_{n}) = v_{j}$ for each $j \in [m]$ in $\VNC^{2}$, and we have an $(m\times m)$-matrix $U$ given by:
\[U_{ij} := \Eval_{alg}(F_{i},v_{j}) \quad (i,j \in [m]).\]
Since we have
\[\sum_{j=1}^{n}(v_{i})_{j} (v_{i^{\prime}})_{j} = \#(A_{i} \cap A_{i^{\prime}})\]
for $i,i^{\prime} \in [m]$, it is straightforward to see that $M$ is upper-triangular, and each diagonal component is nonzero.
Hence, $\det(U) \neq 0$.

Now, we would like to show that each row vector of $U$ is spanned by the functions on $[m]$ induced by multilinear monomials.
Since $\left[\sum_{i=0}^{s}\binom{n}{i}\right]$ enumerates all the subsets of $[n]$ having cardinality $\leq s$, we can regard it as an enumeration of multilinear monomials of degree $\leq s$ whose variables are among $X_{1}, \ldots, X_{n}$.

We define an arithmetic circuit $\lincoeff(F,\sigma)$ for a multivariate polynomial $F(X_{1},\ldots, X_{n})$ of degree $\leq s$ and a code $\sigma \in \left[\sum_{i=0}^{s}\binom{n}{i}\right]$ of multilinear monomials, which returns the coefficient of $\sigma$ of the multilinearization of $F$:
\begin{enumerate}
 \item \[\lincoeff(X_{i},\sigma) :=\begin{cases}
 &1  \quad (\mbox{if $\sigma$ encodes $\{i\}$}) \\
 &0  \quad (\mbox{otherwise})
 \end{cases}. \]
 \item For a constant $a \in \ZZ$, 
  \[\lincoeff(X_{i},\sigma) :=\begin{cases}
 &a  \quad (\mbox{if $\sigma$ encodes $\emptyset$}) \\
 &0  \quad (\mbox{otherwise})
 \end{cases}. \]
 \item \[\lincoeff(F+G,\sigma) := \lincoeff(F,\sigma) +\lincoeff(G,\sigma). \]
 \item \[\lincoeff(F\cdot G,\sigma) := \sum_{\tau_{1}, \tau_{2}\colon \tau_{1}\cup\tau_{2}=\sigma }\lincoeff(F,\tau_{1}) \cdot \lincoeff(G,\tau_{2}). \]
 Here, $\tau_{1}\cup \tau_{2} = \sigma$ means that $\sigma$ codes the union of the sets coded by $\tau_{1}$ and $\tau_{2}$. 
\end{enumerate}

It is straight forward to see that for any $V \in \{0,1\}^{n}$, by induction on the structure of $F$, 
\[\Eval_{alg}(F,V) = \Eval_{alg} \left(\sum_{\sigma \in \left[\sum_{i=0}^{s}\binom{n}{i}\right]} \lincoeff(F,\sigma)M(\sigma),V\right),\]
where $M(\sigma)$ denotes the multilinear monomial represented by $\sigma$.

Hence, we have $U=CM$, where $C \in \Mat_{\ZZ}(m,\sum_{i=0}^{s}\binom{n}{i}), M \in  \Mat_{\ZZ}(\sum_{i=0}^{s}\binom{n}{i},m)$, and 
\begin{align*}
C_{i\sigma} := \lincoeff(F_{i},\sigma) \ \& \ M_{\sigma j} := \Eval_{alg}(M(\sigma),v_{j}) 
\end{align*}
for each $i,j\in[m]$ and $\sigma \in \left[\sum_{i=0}^{s}\binom{n}{i}\right]$.

Now, using rank function, we can show 
\[m=\rank(U) \leq \rank(M) \leq \min\left\{m,\sum_{i=0}^{s}\binom{n}{i}\right\},\]
 a contradiction.
However, we can obtain a contradiction in a more elementary way here as follows; since we have assumed $m > \sum_{i=0}^{s}\binom{n}{i}$, we can pad $C$ by trivial column vectors to make it $(m \times m)$ so as $M$ by trivial row vectors, and obtain
\[U = \begin{bmatrix}
C & \zeromat 
\end{bmatrix}
\begin{bmatrix}
M \\
\zeromat
\end{bmatrix},\]
which is now a multiplication in $\Mat_{\ZZ}(m,m)$.
Then, taking the determinant of both sides, we obtain $0 \neq \det(U)=0$, a contradiction.
\end{proof}
 
As an application of the formalization of matrix rank function, we can obtain that a good lower bound of Ramsey number can be proven in $\VNC^{2}$:
\begin{thm}
There exists a constant $c>0$ such that $\VNC^{2}$ proves; for any number $n$, there exists a graph of $n$-vertices with no clique or independent set larger than $t:=2^{c\sqrt{|{n}| \cdot ||{n}||}}$.

Moreover, the explicit construction of such graphs can be carried out in uniform-$NC^{2}$.
\end{thm}

\begin{proof}[Proof Sketch.]
We observe that the construction in \cite{Grolmusz} can be formalized straightforwardly in $\VNC^{2}$.
We describe where we literally follow \cite{Grolmusz} and where we need some ad-hoc changes.

Definition 1 in \cite{Grolmusz} treats $\Sigma^{B}_{0}$ notions, and it is straightforward to formalize them in $\VNC^{2}$.
As for Definition 2 in \cite{Grolmusz}, we do not formalize it literally.
It refers the minimum number $r$ such that $A$ can be written as $A=BC$, where $B \in \Mat_{\ZZ_{6}}(n,r), C \in  \Mat_{\ZZ_{6}}(r,n)$, which is a $\Delta^{B}_{2}(\VNC^{2})$-property of $r$, and it would be difficult to utilize for further formalization in $\VNC^{2}$ if we adopted it.
Instead, we consider $\rank_{\ZZ_{2}}(A)$ and $\rank_{\ZZ_{3}}(A)$ for a suitable $A$ simultaneously. Note that both quantity are defined and formalized in $\VNC^{2}$ since $\VNC^{2}$ can interpret $\LAPPD$ not only by interpreting the $\fieldsort$ as rationals but also $\ZZ_{2}$ or $\ZZ_{3}$.

As for Lemma 3 in \cite{Grolmusz}, we prove it for $R=\ZZ_{2},\ZZ_{3}$.
It suffices to show that the formalized matrix rank function for $R=\ZZ_{2},\ZZ_{3}$ satisfies Definition 2 in \cite{Grolmusz}. 
(Then we can literally follow the proof given in \cite{Grolmusz}.)
 We work in $\VNC^{2}$. 
 We want to see that $r=\rank(A)$ is the minimum number for which there exists a decomposition $A=BC$ where $B \in \Mat_{\FF}(n,r), C \in \Mat_{\FF}(r,n)$ for $\FF=\ZZ_{2},\ZZ_{3}$.
Indeed, for the existence of a decomposition, consider the decomposition 
 \[A=\basis(A)D,\]
  gather the nontrivial column vectors of $\basis(A)$ and form $B$, and collect the corresponding row vectors of $D$ and form $C$.   
 Furthermore, for the minimality, assume $A=BC$, where $B \in \Mat_{\FF}(n,s), C \in \Mat_{\FF}(s,n)$.
 We show $r \leq s$.
 Suppose otherwise.
 By Lemma \ref{maximalnonsingular}, we can extract a maximal nonsingular submatrix $R \in \Mat_{\FF}(r,r)$ of $A$.
 $\VNC^{2}$ can permute the rows of $A$ and make $R$ the principal submatrix, so, without loss of generality, we may assume that $A$ is of the following form:
 \[A= \begin{bmatrix}
 R & S\\
 T & U
 \end{bmatrix}.\]
Now, we have submatrices $B^{\prime} \in \Mat_{\FF}(r,s)$ of $B$ and $C^{\prime} \in \Mat_{\FF}(s,r)$ of $C$ such that
$R=B^{\prime}C^{\prime}$.
Since $r > s$, we can pad $B^{\prime}$ by trivial column vectors and $C^{\prime}$ by trivial row vectors, and obtain
\[R= \begin{bmatrix}
B^{\prime} & \zeromat
\end{bmatrix}
\begin{bmatrix}
C^{\prime} \\
 \zeromat
\end{bmatrix},\]
which is now a multiplication in $\Mat_{\FF}(r,r)$.
Taking the determinant of both sides, we obtain $0 \neq \det(R) = 0$, a contradiction. 
 
 Now, we proceed to Theorem 4 in \cite{Grolmusz}.
 It suffices to prove the theorem for $p=2,3$ for our purpose this time, and it is straightforward to formalize the proof for these particular cases in $\VNC^{2}$.
 
 Theorem 5 in \cite{Grolmusz} is treated similarly.
 It suffices to prove the theorem for $r=\mmax\{\rank_{\ZZ_{2}}(A), \rank_{\ZZ_{3}}(A)\}$ instead of $r=\rank_{\ZZ_{6}}(A)$.
 Replacing the inequalities 
 \[t \leq r+1 \quad \mbox{and} \quad t \leq \binom{r+1}{2} +1\] 
 with 
 \[t \leq \rank_{\ZZ_{2}}(A)+1 \quad \mbox{and} \quad t \leq \binom{\rank_{\ZZ_{3}}(A)+1}{2} +1\]
 respectively, it is straightforward to carry out the proof of the modified statement in $\VNC^{2}$.
 Note that the graph $G$ is $\Sigma^{B}_{0}$-defined by the given matrix $A$.
 
 Thus the remaining task is to formalize the proof of Theorem 8 in \cite{Grolmusz} so that the co-diagonal matrix $A$ is constructed by a formalized $NC^{2}$-function.
 First, we observe that it suffices to show a modified version of Theorem 8, where $\rank_{\ZZ_{6}}(A)$ replaced with $\mmax\{\rank_{\ZZ_{2}}(A),\rank_{\ZZ_{3}}(A)\}$.
 The proof of Theorem 8 in \cite{Grolmusz} relies on Theorem 10 in \cite{Grolmusz} for $m=6$, $l=2$, and the smallest integer $k$ such that $n \leq k^{k}$.
  Note that we must look at Theorem 2.1 in \cite{polynomialmodcomposite} for the proof of Theorem 10.
 
 Given a number $n$, the number $k$ described previously is available in $\VNC^{2}$ by $\Sigma^{B}_{0}$-Induction.
 We describe how to formalize the proof of Theorem 8 in \cite{Grolmusz}, admitting Theorem 10 in \cite{Grolmusz} for $m=6$, $l=2$ and $k$ at the moment.
 Let $P$ be the polynomial stated in Theorem 10.
 The function $\delta$ and the matrix $\bar{A}$ are $\Sigma^{B}_{0}$-definable in $n$ and $P$, and they are available in $\VNC^{2}$, too.
 By the assumption on $P$, we can see that $\bar{A}$ is co-diagonal as shown in \cite{Grolmusz}.
 As for the paragraph including the formulae (6) and (7), since $k^{c\sqrt{k}} \leq n^{O(1)}$ and $\VNC^{2}$ can formalize counting and bounded sum, the argument can be formalized straightforwardly in $\VNC^{2}$.
 The next two paragraphs decomposing $\bar{A}$ into a sum of $D_{\vec{i}}$'s and transforming each  $D_{\vec{i}}$ to a canonical form can be treated similarly.
 Then, applying Lemma 3 for $R=\ZZ_{2}, \ZZ_{3}$, we obtain the evaluation $\mmax\{\rank_{\ZZ_{2}}(A) , \rank_{\ZZ_{3}}(A)\} \leq k^{2c\sqrt{k}}$.
 The last paragraph is straightforward to implement in $\VNC^{2}$, and it finishes the description of a formalization of the proof of Theorem 8 in \cite{Grolmusz}. 
 
 Now, we look at \cite{polynomialmodcomposite} and describe how to formalize its Theorem 2.1 for $(m,r,N)=(6,2,k)$, which is equivalent to Theorem 10 in \cite{Grolmusz} for $(m,l)=(6,2)$.
 We follow the proof presented in \S 2.2 in \cite{polynomialmodcomposite}, putting $(m,r,N)=(6,2,k)$. Our aim is to construct a polynomial $g \in \ZZ_{6}[X]$ such that $\deg(g) \leq O(k^{\frac{1}{2}})$ and it represents OR in the following sense: for $(x_{1},\ldots, x_{k}) \in \{0,1\}^{k}$, $g(x_{1},\ldots,x_{k})=0$ if and only if $x_{i}=0$ for every $i \in [k]$.  
 
 Let $N \leq k$.
 Let $p^{e}$ be the least power such that $N \leq p^{e}$, which is again available in $\VNC^{2}$ by $\Sigma^{B}_{0}$-Induction.
 For each number $a < p^{e}$, let $v_{a}$ be the following $p^{e}$-dimensional vector:
 \[v_{a}:=\left[\binom{a}{a}, \ldots, \binom{a+p^{e}-1}{a}\right]^{t} \in \ZZ_{p}^{p^{e}}.\]
 (If $a=0$, we set $v_{a}=[1,\ldots, 1]$.)
 We show that $\{v_{a}\}_{a < p^{e}}$ gives a basis of $\ZZ_{p}^{p^{e}}$. 
 
 It suffices to show that $\{v_{a}\}_{a < p^{e}}$ is $\ZZ_{p}$-linearly independent.
 Indeed, if we have the linearly independence, then, by the interpretation of Lemma \ref{CHimpliesnontrivialsolution} in $\VNC^{2}$, we have that $\det[v_{0}, \ldots, v_{p^{e}-1}]\neq 0$, and therefore $\{v_{a}\}_{a < p^{e}}$ spans the whole space $\ZZ_{p}^{p^{e}}$.
 Now, assume $\sum_{a < p^{e}} c_{a}v_{a}=0$, where $c_{a} \in \ZZ_{p}$.
 
 Let $s_{a}(X_{1},\ldots,X_{k})$ be the $a$-th elementary multilinear symmetric function, that is, the sum of all multilinear monomials of degree $a$ in the $k$ input variables.
 Note that the number of its monomials is $\binom{k}{a} \leq k^{p^{e}} \leq n^{O(1)}$, and therefore each $s_{a}$ ($a \leq p^{e}$) is available in $\VNC^{2}$ by $\Sigma^{B}_{0}$-CA.
 Furthermore, since $\VNC^{2}$ can formalize the counting function, it is provable that $s_{a}(\vec{x}) = \binom{a+j-1}{a}$ if 
 \[\#\{i \in [k] \mid x_{i}=1\} = j \geq 1.\]
 
 On the other hand, for $p=2,3$, we have that $j \mapsto \binom{j}{a} \pmod p$ is periodic with period $p^{e}$.
 Furthermore, $s_{0}, \ldots, s_{p^{e}-1}$ are linearly independent modulo $p$ so that they are a basis of the vector space of symmetric functions (i.e. the output depends only on the number of $1$'s in the input) with period $p^{e}$.
 If $N<p^{e}$, the OR of $N$ variables is represented modulo $p$ by the function $f(j)$ with $f(j)=0$ for $j = 0 \pmod p^{e}$ and $f(j)=1$ otherwise.
 This function has degree at most $p^{e}-1$.
 
 But now, putting $N=k^{\frac{1}{2}}$, we obtain $f_{1} \in \ZZ_{2}[X_{1}, \ldots, X_{k}]$ for $p=2$ with period $\geq N$ and $f_{2} \in \ZZ_{3}[X_{1},\ldots, X_{k}]$ for $p=3$ with period $\geq N$.
 Applying Chinese Remainder Theorem, $g:=3f_{1}+2f_{2} \in \ZZ_{6}[X_{1}, \ldots, X_{k}]$ satisfies
 $g(\vec{c})=0 \pmod 6$ if and only if $f_{1}(\vec{c})=0 \pmod 2$ and $f_{2}(\vec{c})=0 \pmod 3$.
 Furthermore, it has period $\geq N^{2}=k$.
 $\deg(g) \leq O(N)=O(k^{\frac{1}{2}})$, and it ends the proof.

\end{proof}

\section{Open questions}\label{Open questions}
As \cite{Uniform} pointed out, it is interesting that known formalization of matrix determinant and rank function still seem to be an overkill.
That is, we have $V\#L$ (cf. \cite{LAPinterpretation}), which corresponds to the complexity class $\#L$ (or $AC^{0}(\#L)$) and can formalize matrix determinants and therefore rank function, but we only know the proofs of their basic properties in $\VNC^{2}$, which corresponds to $NC^{2}$, which includes $\#L$. 
It is interesting to decide whether $V\#L$ is strong enough to carry out these proofs.

Moreover, in the context of linear algebra method in combinatorics, sometimes we need matrices of exponential size with respect to the size of inputs (the incidence matrix of a given set $S$, whose size is $(\#S \times 2^{\#S})$, is a typical example).
This is another possible origin of hardness of the statement which is out of scope of the researches of linear algebra in bounded arithmetics, and worth considering to decide whether we need such huge matrices to carry out the proofs. 

Also, the following question is natural:

\begin{question}\label{strongnonuniformRCW}
Does $\VNC^{2}$ prove the following?:
Let $n,m,s$ be numbers and $\mathcal{F},L$ be strings.
Suppose $\mathcal{F}$ codes a family $\{A_{1}, \ldots, A_{m}\}$ of subsets of $[n]$ and $L$ codes a set $\{l_{1},\ldots, l_{s}\}$ of numbers.
Assume also $\mathcal{F}$ is $L$-intersecting, that is, $\forall i \neq j \in [m]. \#(A_{i} \cap A_{j}) \in L$.
Then the following holds:
\begin{align}\label{informalRCW}
m \leq \sum_{i=0}^{s} \binom{n}{i}
\end{align}
in \textit{string-sort}. 
More precise meaning is as follows. 
For each $i \in [0,n]$, we can define a string $B(n,i)$ as the coefficient of degree-$s$ of the polynomial $(X+1)^{n}$ realized in string-sort.
Moreover, for each number $k$, we have a string counterpart $S(k)$, which is a string coding the binary expansion of $k$.
Now, bounded sum in string-sort is formalized in $\VNC^{1}$, and the comparison of two integers of binary expansions is formalized in $V^{0}$, so we have $\Delta^{B}_{1}(\VNC^{2})$ formula
\[S(m) \leq \sum_{i=0}^{s}B(n,i),\]
and this is the precise meaning of the inequality (\ref{informalRCW}).
\end{question}

Moreover, it is still unclear to what extent we can formalize basic notions and arguments in linear algebra in bounded arithmetics.
Especially,

\begin{question}
 For each matrix $M \in \Mat_{\ZZ}(m,n)$, let $||M||^{2}:= \sum_{(i,j) \in [m]\times[n]}M_{ij}^{2}$.
 Is the following provable in $\VNC^{2}$?:
 Let $A \in \Mat_{\ZZ}(n,n)$, and it is symmetric.
 Suppose a number $d$ satisfies $||Ax||^{2} \leq d||x||^{2}$ for $x \in \Mat_{\ZZ}(n,1)$.
 Then $||A||^{2} \leq \rk(A)d$.  
\end{question}
The statement is a basic comparison between the operator norm and Euclidean norm of $A$ (for a proof, see for example Proposition 13.6 in \cite{Jukna}).
As far as we know, a common proof relies on the notion of eigenvalues and spectral decomposition of symmetric real matrices, and it is at least not straightforward to formalize the argument in bounded arithmetics.
Note that if the answer to the question above is yes, then we can also formalize and prove, for example, the rigidity phenomenon of Hadamard matrices in $\VNC^{2}$. 
See Lemma 14.7 and Theorem 14.8 of \cite{Jukna}.

\appendix
\section{A characterization of the axioms for $\indexsort$ of $\LA$}\label{ReformulationofLA}

In this section, we show that, over Open induction for $\mathcal{L}_{\indexsort}$-formulae, the axioms $(A6)$-$(A16)$ are equivalent to (\dag) (the item \ref{axiomsforindex} in the description of $\LA$ in Section \ref{Preliminaries}).

\begin{proof}
 We prove the equivalence model theoretically. 
 If an $\mathcal{L}_{\indexsort}$-structure $\MM$ satisfies (\dag), then it clearly satisfies $(A6)$-$(A16)$.
 Now, we consider the converse.
 Let $\MM=(M,0,1,+,*,\leq,-,\divi,\rem)$ satisfy $(A6)$-$(A16)$.
 First, we observe the following:
 \begin{claim}\label{arithmeticalproperties}
 \begin{enumerate}
  \item\label{leftzero} For all $j \in M$, $0+j=j$.
  \item\label{associative} For all $i,j,k \in M$, $(i+k)+j=i+(k+j)$.
  \item\label{commutative} For all $i, j \in M$, $i+j=j+i$.
  \item\label{contraction} For all $i,j,k \in M$, $i+j=k+j$ implies $i=k$.
  \item\label{distributive} For all $i,j,k \in M$, $(i+k)*j=i*j+k*j$.
 \end{enumerate}
 \end{claim}
 Indeed, each of them follows from induction on $j$, fixing other variables $i$ and $k$.
 Base cases are dealt with by $(A10)$, the item \ref{commutative} itself, $(A14)$, and $(A9)$.
 Step cases are by $(A12)$, $(A8)$, and $(A7)$.
 
 From now on, we sometimes omit parenthesis when we add multiple elements based on the associative law (\ref{associative})
 
 Furthermore,
 \begin{claim}\label{orderproperties}
 \begin{enumerate}
  \item\label{predecessor} For all $j \in M$, $j=0$ or $(j-1)+1=j$.
  \item\label{discretelyordered} For all $i \in M$, $i=0$ or $1 \leq i$.
  \item\label{defoforder} Forall $i,j \in M$, $i \leq j$ if and only if there exists $k \in M$ such that $i+k=j$.
  \item\label{difference} For all $i,j \in M$, $i \leq j$ implies $(j-i)+i=j$.
  \end{enumerate}
 \end{claim}
 Claim \ref{orderproperties}(\ref{predecessor}) follows from induction on $j$; if $j=0$, the claim is trivial.
 Assume the claim holds for $j$.
 We can show the claim holds for $j+1$ without using IH.
 Since $1+j=j+1$ by Claim \ref{arithmeticalproperties}(\ref{commutative}), we have $(j+1)-1=j$ by $(A15)$.
 Hence, $((j+1)-1)+1=j+1$, as desired.
 
 Claim \ref{orderproperties}(\ref{discretelyordered}) follows immediately from the previous (\ref{predecessor}) as follows. 
 Suppose $i \neq 0$ and $1 \not\leq i$.
 By Claim \ref{orderproperties}(\ref{predecessor}), $(i-1)+1=i$.
 Then $(A9)$ and Claim \ref{arithmeticalproperties}(\ref{commutative}) imply $1 \leq i$.
 
 Claim \ref{orderproperties}(\ref{defoforder}) is proved as follows. 
 The converse is the very thing $(A9)$ states, so we focus on the other way around.
 Suppose $i \leq j$.
 By $(A16)$, we have $j=i*\divi(j,i)+\rem(j,i)$.
 Suppose $\divi(j,i) =0$.
 Then we have $j=\rem(j,i) < i$ by $(A10)$, Claim \ref{arithmeticalproperties} (\ref{leftzero}), and $(A16)$.
 Since $i \leq j$, we have $i=j$ by $(A13)$, a contradiction.
 Hence, we obtain $\divi(j,i) \neq 0$.
 Hence, Claim \ref{orderproperties}(\ref{predecessor}) implies there exists $d \in M$ such that $d+1=\divi(j,i)$.
 Thus we have $j=i*(d+1)+\rem(j,i)=i*d+i+\rem(j,i)=i+(i*d+\rem(j,i))$, as required.

 Claim \ref{orderproperties}(\ref{difference}) follows immediately from the previous (\ref{defoforder}) and $(A15)$. 
 
 We construct an ordered ring to which $M$ is embedded as its nonnegative part.
 Consider the following binary relation $\sim$ on $M$:
 \[(a_{1},b_{1}) \sim (a_{2},b_{2}) : \Leftrightarrow a_{1}+b_{2}=a_{2}+b_{1}.\]
 
 \begin{claim}
  $\sim$ is an equivalence relation on $M$.
 \end{claim}
 Indeed, $\sim$ is reflexive and symmetric by equality axioms, and it is also transitive by the following reasoning:
 assume $(a_{1},b_{1}) \sim (a_{2},b_{2})$ and $(a_{2},b_{2}) \sim (a_{3},b_{3})$.
 Then we have $a_{1}+b_{3}+b_{2} = a_{3}+b_{1}+b_{2}$ since $+$ is commutative by Claim \ref{arithmeticalproperties} (\ref{commutative}).
 Together with Claim \ref{arithmeticalproperties} (\ref{contraction}), we obtain $a_{1}+b_{3} = a_{3}+b_{1}$, that is, $(a_{1},b_{1}) \sim (a_{3},b_{3})$.
 
 Let $R:=M^{2} / \sim$.
 For each $(x,y) \in M^{2}$, we denote its equivalence class by $[(x,y)]$.
 We consider the following operations and relations on $R$:
 \begin{enumerate}
  \item $[(x,y)] +_{R} [(z,w)] := [(x+z,y+w)]$.
  \item $[(x,y)] *_{R} [(z,w)] := [(x*z+y*w, y*z+x*w)]$.
  \item $[(x,y)] \leq_{R} [(z,w)] : \Leftrightarrow x+w \leq y+z$.
 \end{enumerate}
 They are well-defined.
 Indeed, the well-definedness of $+_{R}$ follows from the definition of $\sim$ and the commutativity of $+$ on $M$.
 $*_{R}$ is well-defined because of the following reasoning.
 We omit $*$ for readability.
 Let $(x,y) \sim (x',y')$ and $(z,w) \sim (z',w')$.
 We want to show $(xz+yw,yz+xw) \sim (x'z'+y'w',y'z'+x'w')$.  
 Using the commutativity of $+$ and the distributive law in Claim \ref{arithmeticalproperties}, we have
 \begin{align*}
 (xz+yw+y'z'+x'w') + xw'+yz' = (x'z'+y'w' + yz+xw) + xw'+yz'.
 \end{align*}
 Hence, by Claim \ref{arithmeticalproperties} (\ref{contraction}), we have $(xz+yw,yz+xw) \sim (x'z'+y'w',y'z'+x'w')$.
 
 As for $\leq_{R}$, assume $(x,y) \sim (x',y')$, $(z,w) \sim (z',w')$, and $x+w \leq y+z$.
 Then we can show $x'+w'+y+z \leq y'+z'+y+z$, and hence we have $x'+w' \leq y'+z'$ by Claim \ref{arithmeticalproperties} (\ref{contraction}). 
 
 Now, it is a straightforward exercise to see that $(R,+_{R},*_{R},\leq_{R})$ is an ordered ring;
 $[(0,0)]$ is the zero element and $[(1,0)]$ is the unit element.
 
 The original $M$ is embedded by the following mapping:
 \[h \colon M \rightarrow R;\ a \mapsto [(a,0)].\]
 It is again an easy exercise to show that $h$ preserves addition, multiplication, and ordering.
 Furthermore, it is injective since $[(a,0)]=[(b,0)]$ if and only if $a=a+0=b+0=b$.
 
 Actually, the image $h(M)$ is the nonnegative part of $R$.
 Indeed, if $[(0,0)] \leq_{R} [(a,b)]$ in $R$, then $b \leq a$.
 Hence, $(b-a)+a=b$ by Claim \ref{orderproperties}(\ref{difference}), which means $[(a,b)]=h(b-a)$.
 
 The facts that $-$ is the modified minus on $M$, $M$ admits a division with respect to $\leq$, $\divi$ and $\rem$ are the quotient and the remainder are exactly what $(A15)$ and $(A16)$ say.
 This finishes the proof.
\end{proof}

 \section{A proof sketch of Proposition \ref{[]pol is an interpretation}}\label{proof of []pol}

This section is a continuation of \S \ref{Preliminaries}.

We can show that $=_{pol},0_{pol},1_{pol},+_{pol},-_{pol}$ and $*_{pol}$ satisfy some of the axioms of $\LA$ again:
\begin{lemma}[$\LA_{-}$]
\begin{enumerate}
 \item For $f, g \in \FF[X]$, 
 \[\coeff(f *_{pol} g,k) = \sum_{j=0}^{k} \coeff(f,k-j) \coeff(g,j).\]
 \item $f =_{pol} g \leftrightarrow \forall j.\ \extract (f,j,1)=\extract (g,j,1)$.
 \item $=_{pol}$ is an equivalence relation on $\FF[X]$.
 \item $=_{pol}$ is a congruence relation with respect to $+_{pol},-_{pol}$ and $*_{pol}$.
 \item $\deg$ is invariant under the relation $=_{pol}$.
 \item The relation $=_{pol}$, the constants $0_{pol},1_{pol}$, and the functions $+_{pol},-_{pol}, $ and $*_{pol}$ satisfy the axioms $(A18)$-$(A26)$ except $(A21)$, regarded as the axioms on $\FF[X]$.
 \item $f *_{pol} g = 0 \rightarrow f =_{pol} 0_{pol} \lor g=_{pol}0_{pol}$.
\end{enumerate}
\end{lemma}

\begin{proof}
 \[\coeff(f *_{pol} g,k) = \sum_{j=0}^{k} \coeff(f,k-j) \coeff(g,j)\]
 
 follows from observing each component of $\conv(f, \row (g)-1) g$.

 It is straightforward to see $=_{pol}$ is a congruence relation.
 
 We consider the axioms of $\LA_{-}$.

 $(A18)$-$(A20)$, $(A22)$, $(A24)$, and $(A26)$ are easy. 
 Note that we can also show that 
 \[(f+_{pol}g)*_{pol} h = (f*_{pol}h) +_{pol} (g*_{pol} h)\]
 
 since $=_{pol}$ is a congruence relation, and
 \[\row (f)=\row (g) \rightarrow \conv(f+g,l) = \conv(f,l) + \conv(g,l).\]
 
 We show $(A23)$ first. 
 Let $f,g \in \FF[X]$.
 We would like to show $f*_{pol} g =_{pol} g*_{pol}f$.
 
 Comparing the coefficients of both sides, it suffices to show that 
 \[\sum_{j=0}^{k} a_{k-j}b_{j} = \sum_{j=0}^{k} b_{k-j}a_{j}\]
 for general $[a_{0}, \ldots, a_{k}]$ and $[b_{0}, \ldots, b_{k}]$.
 
 This follows from open induction, commutativity low in $\fieldsort$ and the following theorems of $\LA_{-}$:
 \begin{align*}
  \sum[c_{0}, \ldots, c_{k+1}] &= \sum[c_{0}, \ldots, c_{k}] + c_{k+1},\\
  \sum[c_{0}, \ldots, c_{k+1}] &= c_{0}+ \sum[c_{1}, \ldots, c_{k+1}] \\
  \sum[c_{0}, \ldots, c_{k}] &= \sum[c_{k}, \ldots, c_{0}] .
 \end{align*}
 
 We show $(A25)$. 
 First, we observe
 \begin{align*}
 \coeff(f*_{pol}(g*_{pol}h) ,k) 
 &= \sum_{j=0}^{k} \coeff(f,k-j)\coeff(g*_{pol}h , j) \\
 &= \sum_{j=0}^{k} \coeff(f,k-j)\left(\sum_{i=0}^{j}\coeff(g , j-i) \coeff(h , i) \right)\\
 &= \sum_{j=0}^{k} \sum_{i=0}^{j} \coeff(f,k-j)\coeff(g , j-i) \coeff(h , i) ,
 \end{align*}
 and 
 \begin{align*}
  \coeff((f*_{pol}g)*_{pol}h ,k) 
  &= \sum_{j=0}^{k} \coeff(f*_{pol}g,k-j)\coeff(h , j) \\
  &= \sum_{j=0}^{k} \left(\sum_{i=0}^{k-j}\coeff(f , k-j-i) \coeff(g , i) \right) \coeff(h , j) \\
  &= \sum_{j=0}^{k} \sum_{i=0}^{k-j} \coeff(f,k-j-i)\coeff(g , i) \coeff(h , j) .
 \end{align*}
Now, it is easy to see that
 \[\sum_{j=l_{1}}^{l_{2}} c_{j} = \sum_{j=l_{1}+a}^{l_{2}+a} c_{j-a} \]
 by open induction and the axioms on $\sum$, we obtain that  
 \begin{align*}
  &\sum_{j=0}^{k} \sum_{i=0}^{k-j} \coeff(f,k-j-i)\coeff(g , i) \coeff(h , j) \\
  =& \sum_{j=0}^{k} \sum_{i=j}^{k} \coeff(f,k-i)\coeff(g , i-j) \coeff(h , j) \\
  =&\sum_{i=0}^{k} \sum_{j=i}^{k} \coeff(f,k-j)\coeff(g , j-i) \coeff(h , i) .
 \end{align*}
 The last equation is obtained by just renaming the dummy variables.
 
 Thus, putting
 \[A:= \lamt_{ij} \langle k,k, \cond( i\leq j, \coeff(f,k-j) \coeff(g , j-i) \coeff(h,i), 0 )\rangle,\]
 we obtain 
 \begin{align*}
 \sum(A) &= \coeff((f*_{pol}g)*_{pol}h ,k), \\
 \sum(A^{t}) &= \coeff(f*_{pol}(g*_{pol}h) ,k) 
 \end{align*}
 Therefore, since $\sum(A)=\sum(A^{t})$ (see \cite{The proof complexity of linear algebra}), we obtain $(A25)$.
 \end{proof}
 
 For $(A,d) \in \Mat_{\FF[X]}(m,n)$ and $(B,d^{\prime}) \in \Mat_{\FF[X]}(n,l)$,
  \[(A,d) *_{pol} (B,d^{\prime}) := (\conv(A,d,d^{\prime})B, d+d^{\prime}).\]

\begin{lemma}\label{Ring properties for pol}
  $\LA_{-} \vdash  \conv(A,d,0) = A$.
  
  Moreover, $\LA_{-}$ can prove that $(0_{m,n},0), (I_{m},0), +_{pol}, *_{pol}, \mmax, \summation_{pol}, =_{pol}$ satisfy the Ring properties $(T1)$-$(T15)$ given in \cite{The proof complexity of linear algebra}.
 \end{lemma}
 
 \begin{proof}
  The essential part is to prove the properties of $*_{pol}$, and it can be done following the lines of \cite{The proof complexity of linear algebra} if we can establish 
  \[\coeff((A,d) *_{pol} (B,d^{\prime}), k ) = \sum_{j=0}^{k} \coeff(A,d,j) \coeff(B,d^{\prime},k-j)\]
    for each $k$.
    Now, the verification of this equality is straightforward.
 \end{proof}

Now, 

\begin{lemma}[$\LAP$]
 $\mathtt{P}_{pol}$ satisfies $(A34), (A35)$:
 \begin{enumerate}
  \item[(A34)] \[\mathtt{P}_{pol}(0,A,d) =_{pol} (I_{m},0) \in \Mat_{\FF[X]}(m,m).\]
  \item[(A35)] \[(A,d) \in \Mat_{\FF[X]}(m,m) \rightarrow (A,d)^{k+1} =_{pol} (A,d)^{k} *_{pol} (A,d).\]
 \end{enumerate}
 Moreover, 
 \[(A,d) \in \Mat_{\FF[X]}(m,m) \rightarrow (A,d)^{k} *_{pol} (A,d)^{l} =_{pol} (A,d)^{k+l}.\]
\end{lemma}

\begin{proof}
 For $(A34)$, we have the following equalities:
 \begin{align*}
  \mathtt{P}_{pol}(0,A,d)=(Q_{pol}(0,A,d),0) = (\pconv(A,d,0)I_{\row(A)}, 0) = ((A,d),0).
 \end{align*}
 
 For $(A35)$, first it is easy to see that 
 \[(Q_{pol}(m+1,A,d),(m+1)d) =_{pol} (A,d) *_{pol} (Q_{pol}(m,A,d),md)\]
  by definitions.
 
 Now, we show the ``moreover'' part by induction on $k$.
 
 The case when $k=0$ is clear.
 The case of $k+1$ follows easily from the case of $k$ and induction hypothesis.
 Indeed,
 \begin{align}
  (A,d)^{k+1} *_{pol} (A,d)^{l} 
  &=_{pol} ((A,d) *_{pol} (A,d)^{k}) *_{pol} (A,d)^{l} \label{by def of Ppol and *pol}\\
  &=_{pol} (A,d) *_{pol} ((A,d)^{k} *_{pol} (A,d)^{l}) \label{by Ring properties} \\
  &=_{pol} (A,d) *_{pol} (A,d)^{k+l} \label{by IH} \\
  &=_{pol} (A,d)^{k+l+1} \label{by def of Ppol and *pol2}
 \end{align}
 Here, the equalities (\ref{by def of Ppol and *pol}) and (\ref{by def of Ppol and *pol2}) follow from the definitions of $\mathtt{P}_{pol}$ and $*_{pol}$, 
 (\ref{by Ring properties}) from the associativity law established in Lemma \ref{Ring properties for pol},
 and (\ref{by IH}) from the induction hypothesis.

 $(A35)$ follows from ``moreover'' part.
\end{proof}

Now, we observe the following:

\begin{lemma}\label{bounding}
For a term $t$ outputting a $\fieldsort$ element, 
\[\LAP \vdash \deg (\llbracket t \rrbracket_{pol}) \leq \bold{b}[t],\]
that is, the universal closure of LHS is true under $\LAP$.
For a term $t$ outputting a $\matrixsort$ element,
\[\LAP \vdash \deg(\extract (\llbracket t \rrbracket_{pol}, i, j))  \leq \bold{b}[t] .\]
\end{lemma}

\begin{lemma}
Let $\varphi$ be an open $\LAP$-formula.
Then $\llbracket \varphi \rrbracket_{pol}$ is equivalent to a $\Sigma^{B}_{0}$-formula in $\LAP$. 

\end{lemma}

The proofs can be carried out by the structural induction on $t$ and $\varphi$, and we omit them.
Armed with these and Lemma \ref{characteristicfunction}, it follows that $\langle \MM, \FF[X], \Mat_{\FF[X]} \rangle$ satisfies Open Induction.
Other axioms in $\LAP_{-}$ are more straightforward to verify, and we have Proposition \ref{[]pol is an interpretation}.




\end{document}